%% file: corec-calc.tex
\documentclass{jfp}
\usepackage[utf8]{inputenc}
\usepackage[T1]{fontenc}
\usepackage{microtype}
\usepackage{bold-extra}
\usepackage{xcolor}
\usepackage{caption}
\usepackage{subcaption}
\usepackage{balance}
\usepackage{scalerel}
\usepackage{booktabs} 
\usepackage{amssymb}
\usepackage{stmaryrd}
\usepackage[emu]{cmll}
\usepackage[inline,short labels]{enumitem}
\usepackage{thmtools,thm-restate}
\usepackage{proof}
\usepackage{cleveref}
\usepackage{preamble}

\makeatletter
\let\JFP@linecountL\relax
\makeatother

\renewcommand{\cite}{\citep}

\crefname{fact}{fact}{facts}
\Crefname{fact}{Fact}{Facts}

\newcommand{\paul}[1]{}
\newcommand{\zena}[1]{}

\allowdisplaybreaks[1]

\begin{document}

\journaltitle{arXiv}
\cpr{}
\doival{}
\def\doitext{}

\lefttitle{P.~Downen and Z.M.~Ariola}
\righttitle{Classical (Co)Recursion: Mechanics}

\newcommand{\hi}[2][black!15]{{\setlength{\fboxsep}{1pt}\kern-1pt{\colorbox{#1}{$#2$}}}}

\totalpg{\pageref{lastpage01}}
\jnlDoiYr{2021}

\title{Classical (Co)Recursion: Mechanics}

\begin{authgrp}
\author{Paul Downen} and~ \author{Zena M. Ariola}
\affiliation{University of Oregon \\
        (\email{\{pdownen,ariola\}@cs.uoregon.edu})}
\end{authgrp}


\begin{abstract}
  Primitive recursion is a mature, well-understood topic in the theory and
  practice of programming.  Yet its dual, primitive co\-recursion, is
  underappreciated and still seen as exotic.  We aim to put them both on equal
  footing by giving a foundation for primitive co\-recursion based on
  computation, giving a terminating calculus analogous to the original
  computational foundation of recursion.  We show how the implementation details
  in an abstract machine strengthens their connection, syntactically deriving
  co\-recursion from recursion via logical duality.  We also observe the impact
  of evaluation strategy on the computational complexity of primitive
  (co\-)recursive combinators: call-by-name allows for more efficient recursion,
  but call-by-value allows for more efficient co\-recursion.
\end{abstract}

\maketitle

\input{sec_introduction}

\input{sec_rec-fun}

\input{sec_rec-mach}

\input{sec_corec-mach}

\input{sec_corec-vs-coiter}

\input{sec_corec-mach-correct}

\input{sec_related-work}

\input{sec_conclusion}

\section*{Acknowledgments}
\noindent
This work is supported by the National Science Foundation under Grant
No.~1719158.

\bibliographystyle{JFPlike}
\bibliography{rec}

\clearpage
\appendix

\allowdisplaybreaks[3]

\input{sec_model}

\label{lastpage01}

\end{document}

%% file: sec_introduction.tex
\section{Introduction}
\label{sec:introduction}

Primitive recursion has a solid foundation in a variety of different fields.  In
the categorical setting, it can be seen in the structures of algebras.  In the
logical setting, it corresponds to proofs by induction.  And in the
computational setting, it can be phrased in terms of languages and type theories
with terminating loops, like G\"odel's System T \cite{SystemT}.  The latter
viewpoint of computation reveals a fine-grained lens with which we can study
the subtle impact of the primitive combinators that capture different forms
of recursion.  For example, the recursive combinators given by
\citet{Mendler87,Mendler88} yield a computational complexity for certain
programs when compared to encodings in System F
\cite{Bhm1985AutomaticSO,ProofsAndTypes}.  Recursive combinators have desirable
properties---like the fact that they always terminate---which make them useful for
the design of well-behaved programs \cite{MFP91,OrigamiProgramming}, also for
optimizations made possible by applying those properties and theorems
\cite{Malcom90}.

The current treatment of the dual of primitive recursion---\emph{primitive
  co\-recursion}---is not so fortunate.  Being the much less understood of the
two, co\-recursion is usually only viewed in light of this duality.
Consequently, co\-recursion tends to be relegated to a notion of co\-algebras
\cite{RuttenMethodofCoalgebra}, because only the language of category theory
speaks clearly enough about their duality.  This can be seen in the study of
co\-recursion schemes, where co\-algebraic ``anamorphisms'' \cite{MFP91} and
``apomorphisms'' \cite{Vene98functionalprogramming} are the dual counterparts to
algebraic ``paramorphisms'' \cite{Meertens92} and ``catamorphisms''
\cite{Hinze13}.  Yet the logical and computational status of co\-recursion is not
so clear.  For example, the introduction of stream objects is sometimes
described as the ``dual'' to the elimination of natural numbers
\cite{SangiorgiIntroCoinduction,RoyOPLSSCoinduction}, but how is this so?

The goal of this paper is to provide a purely computational and logical
foundation for primitive co\-recursion based on classical logic.  Specifically, we
will express different principles of co\-recursion in a small core calculus,
analogous to the canonical computational presentation of recursion
\cite{SystemT}.  Instead of the (co)\-algebraic approach, we derive the symmetry
between recursion and co\-recursion through the mechanics of programming
language implementations, formalized in terms of an abstract machine.  This
symmetry is encapsulated by the duality \cite{HaginoCodata,StructuralRecursion}
between \emph{data types}---defined by the \emph{structure} of objects---and
\emph{co\-data types}---defined by the \emph{behavior} of objects.

We begin in \cref{sec:rec-fun} with a review of the formalization of primitive
recursion in terms of a foundational calculus: System T \cite{SystemT}.  We
point out the impact of evaluation strategy on different primitive recursion
combinators, namely the \emph{recursor} and the \emph{iterator}:
\begin{itemize}[leftmargin=2em]
\item In call-by-value, the recursor is just as (in)efficient as the iterator.
\item In call-by-name, the recursor may end early; an asymptotic complexity
  improvement.
\end{itemize}
\Cref{sec:rec-mach} presents an abstract machine for both call-by-value and
-by-name evaluation, and unifies both into a single presentation
\cite{SequentMachines,SequentTutorial}. The lower level nature of the abstract
machine explicitly expresses how the recursor of inductive types, like numbers,
accumulates a continuation during evaluation, maintaining the progress of
recursion. This is implicit in the operational model of System T.  The machine
is shown correct, in the sense that a well typed program will always terminate
and produce an observable value (\cref{thm:t-type-safety,thm:t-termination}),
which in our case is a number.

\Cref{sec:corec-mach} continues by extending the abstract machine with the
primitive co\-recursor for streams.  The novelty is that this machine is derived
by applying syntactic duality to the previous one, naturally leading us to a
classical co\-recursive combinator with multiple outputs, modeled as multiple
continuations.  From duality in the machine, we can see that the co\-recursor
relies on a value accumulator; this is logically dual to the recursor's return
continuation.  Like recursion versus iteration, in \cref{sec:corec-vs-coiter} we
compare co\-recursion versus co\-iteration: co\-recursion can be more efficient
than co\-iteration by letting co\-recursive processes stop early.  Since
co\-recursion is dual to recursion, and call-by-value is dual to call-by-name
\cite{DualityOfComputation,CBVDualToCBN}, this improvement in algorithmic
complexity is only seen in call-by-value co\-recursion.  Namely:
\begin{itemize}[leftmargin=2em]
\item In call-by-name, the co\-recursor is just as (in)efficient as the
  co\-iterator.
\item In call-by-value, the co\-recursor may end early; an asymptotic complexity
  improvement.
\end{itemize}
Yet, even though we have added infinite streams, we don't want to ruin System
T's desirable properties.  So in \cref{sec:safety&termination}, we give an
interpretation of the type system which extends previous models of finite types
\cite{ClassicalStrongNormalization,DualityOfIntersectonUnionTypes} with the
(co)\-recursive types of numbers and streams.  The novel key step in reasoning
about (co)\-recursive types is in reconciling two well-known fixed point
constructions---Kleene's and Knaster-Tarski's---which is non-trivial for
classical programs with control effects.  This lets us show that, even with
infinite streams, our abstract machine is terminating and type safe
(\cref{thm:mach-corec-type-safety,thm:mach-corec-termination}).

Proofs to all theorems that follow are given in the appendix.


%% file: sec_rec-fun.tex
\section{Recursion on Natural Numbers: System T}
\label{sec:rec-fun}


\label{sec:system-t}


We start with G\"odel's System T \cite{SystemT}, a core calculus which allows us to define functions by  structural recursion.
Its syntax
is given in \cref{fig:t-syntax}.  It is a canonical extension of the
simply-typed $\lambda$-calculus, whose focus is on functions of type $A \to B$,
with ways to construct natural numbers of type $\Nat$.
The $\Nat$ type comes equipped with two
constructors $\Zero$ and $\Succ$, and 
a built-in recursor, which we write as
$\Rec M \As{} \{\Zero \to N \mid \Succ x \to y.N'\}$.
 This $\Rec$-expression analyzes $M$ to determine if it has the shape
$\Zero$ or $\Succ x$, and the matching branch is returned.  In addition to
binding the predecessor of $M$ to $x$ in the $\Succ x$ branch, the
\emph{recursive result}---calculated by replacing $M$ with its predecessor
$x$---is bound to $y$.

\begin{figure}
\begin{alignat*}{2}
  \<Type> &\ni{}&
  A, B
  &::= A \to B
  \Alt \Nat
  \\
  \<Term> &\ni{}&
  M, N
  &::= x
  \Alt \lambda x. M
  \Alt M ~ N
  \Alt \Zero
  \Alt \Succ M
  \Alt \Rec M \As \  \{\Zero \to N \mid \Succ x \to y. M\}
\end{alignat*}
\caption[System T]{System T: $\lambda$-calculus with numbers and recursion.}
\label{fig:t-syntax}
\end{figure}

\begin{figure}
\begin{minipage}{1.0\linewidth}
\begin{gather*}
  \axiom[\<Var>]
  {\Gamma, x \givestype A \entails x \givestype A}
  \\[\ruleskip]
  \infer[{\to}I]
  {\Gamma \entails \lambda x. M \givestype A \to B}
  {\Gamma, x \givestype A \entails M \givestype B}
  \quad
  \infer[{\to}E]
  {\Gamma \entails M ~ N \givestype B}
  {
    \Gamma \entails M \givestype A \to B
    &
    \Gamma \entails N \givestype A
  }
  \\[\ruleskip]
  \axiom[{\Nat}I_{\Zero}]
  {\Gamma \entails \Zero \givestype \Nat}
  \qquad
  \infer[{\Nat}I_{\Succ}]
  {\Gamma \entails \Succ M \givestype \Nat}
  {\Gamma \entails M \givestype \Nat}
  \\[\ruleskip]
  \infer[{\Nat}E]
  {
    \Gamma \entails
    \Rec M \As{}  \{\Zero \to N \mid \Succ x \to y.N'\} \givestype A
  }
  {
    \Gamma \entails M \givestype \Nat
    &
    \Gamma \entails N \givestype A
    &
    \Gamma, x \givestype \Nat, y \givestype A \entails N' \givestype A
  }
\end{gather*}
\end{minipage}
\caption{Type system of System T.}
\label{fig:t-type-system}
\end{figure}

\begin{figure}
\emph{Call-by-name values ($V$) and evaluation contexts ($E$)}:
\begin{align*}
  \<Value> \ni V, W &::= M
  &
  \<EvalCxt> \ni E
  &::= \hole
  \Alt E ~ N
  \Alt \Rec E \As \ \{ \Zero \to N \mid \Succ x \to y. N' \}
\end{align*}

\emph{Call-by-value values ($V$) and evaluation contexts ($E$)}:
\begin{alignat*}{3}
  \<Value> &\ni{}&
  V, W
  &::= x
  \Alt \lambda x. M
  \Alt \Zero
  \Alt \Succ V
  \\
  \<EvalCxt> &\ni{}&
  E
  &::= \hole
  \Alt E ~ N
  \Alt V ~ E
  \Alt \Succ E
  \Alt \Rec E \As \ \{ \Zero \to N \mid \Succ x \to y. N' \}
\end{alignat*}

\emph{Operational rules}
\begin{align*}
  (\beta_\to)&&
  (\lambda x. M) ~ V &\srd M\subst{x}{V}
  \\
  (\beta_{\Zero})&&
  \begin{aligned}
    \Rec& \Zero \As
    \\[-1ex]
    \{& \Zero \to N
    \\[-1ex]
    \mid& \Succ x \to y. N' \}    
  \end{aligned}
  &\srd
  N
  \\
  (\beta_{\Succ})&&
  \begin{aligned}
    \Rec& \Succ V \As
    \\[-1ex]
    \{& \Zero \to N
    \\[-1ex]
    \mid& \Succ x \to y. N' \}    
  \end{aligned}
  &\srd
  (\fn y N'\subst{x}{V})
  ~
  \left(
    \begin{aligned}
      \Rec ~& V \As
      \\[-1ex]
      \{& \Zero \to N
      \\[-1ex]
      \mid& \Succ x \to y. N' \}    
    \end{aligned}
  \right)
\end{align*}
\caption{Call-by-name and Call-by-value Operational semantics of System T.}
\label{fig:t-operation}
\end{figure}

The type system of System T is given in \cref{fig:t-type-system}.  The $\<Var>$,
${\to}I$ and ${\to}E$ typing rules are from the simply typed $\lambda$-calculus.
The two ${\Nat}I$ introduction rules give the types of the constructors of
$\Nat$, and the
${\Nat}E$ elimination rule types the $\Nat$ recursor.

System T's call-by-name and -value operational semantics are given in
\cref{fig:t-operation}.  Both of these evaluation strategies share operational
rules of the same form, with $\beta_\to$ being the well-known $\beta$ rule of
the $\lambda$-calculus, and $\beta_{\Zero}$ and $\beta_{\Succ}$ defining
recursion on the two $\Nat$ constructors.  The only difference between
call-by-value and -name evaluation lies in their notion of \emph{values} $V$
(\ie those terms which can be substituted for variables) and \emph{evaluation
  contexts} (\ie the location of the next reduction step to perform).  Note that
we take this notion seriously, and \emph{never} substitute a non-value for a
variable.  As such, the $\beta_{\Succ}$ rule does not substitute the recursive
computation $\Rec V \As \ \{ \Zero \to N \mid \Succ x \to y. N' \}$ for $y$,
since it might not be a value (in call-by-value).
The next reduction step depends on the evaluation strategy.  In call-by-name,
this next step is indeed to substitute
$\Rec V \As \ \{ \Zero \to N \mid \Succ x \to y. N' \}$ for $y$, and so we have:
\begin{spacing}
\begin{align*}
  \begin{aligned}
    \Rec& \Succ M \As
    \\
    \{& \Zero \to N
    \\
    \mid& \Succ x \to y. N' \}    
  \end{aligned}
  &\srds
  N'
  \bigsubs{
    \asub{x}{M},
    \abigsub{y}
    {
      \left(
      \begin{aligned}
        \Rec& ~M \As
        \\
        \{& \Zero \to N
        \\
        \mid& \Succ x \to y. N' \}    
      \end{aligned}
      \right)
    }
  }
\end{align*}
\end{spacing}
So call-by-name recursion is computed starting with the current (largest) number
first and ending with the smallest number needed (possibly the base case for
$\Zero$).  If a recursive result is not needed then it is not computed at all,
allowing for an early end of the recursion.  In contrast, call-by-value must
evaluate the recursive result first before it can be substituted for $y$.  As
such, call-by-value recursion \emph{always} starts by computing the base case
for $\Zero$ (whether or not it is needed), and the intermediate results are
propagated backwards until the case for the initial number is reached.  So
call-by-value allows for no opportunity to end the computation of $\Rec$ early.

\begin{example}
\label{ex:t-arith}

The common arithmetic functions $\<plus>$, $\<times>$, $\<pred>$, and $\<fact>$
can be written in System T as follows:
\begin{align*}
  \<plus>
  &=
  \lambda x. \lambda y.
  \Rec x \As \
  \{
  \Zero \to y
  \mid
  \Succ \blank \to z. \Succ z
  \}
  \\
  \<times>
  &=
  \lambda x. \lambda y.
  \Rec x \As \
  \{
  \Zero \to \Zero
  \mid
  \Succ \blank \to z. \<plus>~y~z
  \}
  \\
  \<pred>
  &= \lambda x.
  \Rec x \As \ \{ \Zero \to \Zero \mid \Succ x \to z.x \}
  \\
  \<fact>
  &=
  \lambda x.
  \Rec x \As \
  \{
  \Zero \to \Succ\Zero
  \mid
  \Succ y \to z. \<times>~(\Succ y)~z
  \}
\end{align*}
Executing $pred~(\Succ(\Succ\Zero))$ in call-by-name proceeds like so:
\begin{spacing}
\begin{align*}
  &
  \<pred>~(\Succ(\Succ\Zero))
  \\
  &\mapsto 
  \Rec \Succ(\Succ\Zero) \As \
  \{ \Zero \to \Zero \mid \Succ x \to z.x \}
  &&(\beta_\to)
  \\
  &\mapsto
  (\lambda z. \Succ\Zero)
  ~
  (\Rec \Succ\Zero \As \
  \{ \Zero \to \Zero \mid \Succ x \to z.x \})
  &&(\beta_{\Succ})
  \\
  &\mapsto
  \Succ\Zero
  &&(\beta_\to)
\end{align*}
\end{spacing}
Whereas, in call-by-value, the predecessor of both $\Succ \Zero$ and $\Zero$ is
computed even though these intermediate results are not needed in the end:
\begin{spacing}
\begin{align*}
  &
  \<pred>~(\Succ(\Succ\Zero))
  \\
  &\mapsto
  \Rec \Succ(\Succ\Zero) \As \
  \{ \Zero \to \Zero \mid \Succ x \to z.x \}
  &&(\beta_\to)
  \\
  &\mapsto
  (\lambda z. \Succ\Zero)
  ~
  (\Rec \Succ\Zero  \As \ \{ \Zero \to \Zero \mid \Succ x \to z.x \})
  &&(\beta_{\Succ})
  \\
  &\mapsto
  (\lambda z. \Succ\Zero)
  ~
  ((\lambda z.\Zero)
  ~
  (\Rec \Zero \As \
  \{ \Zero \to \Zero \mid \Succ x \to z.x \}))
  &&(\beta_{\Succ})
  \\
  &\mapsto
  (\lambda z. \Succ\Zero)~ ((\lambda z.\Zero) ~ \Zero)
  &&(\beta_{\Zero})
  \\
  &\mapsto
  (\lambda z. \Succ\Zero) ~ \Zero
  &&(\beta_\to)
  \\
  &\mapsto
  \Succ\Zero
  &&(\beta_\to)
\end{align*}
\end{spacing}

In general, $\<pred>$ is a constant time ($O(1)$) function over the size of its
argument when following the call-by-name semantics, which computes the
predecessor of any natural number in a fixed number of steps.  In contrast,
$\<pred>$ is a linear time ($O(n)$) function when following the call-by-value
semantics, where $\<pred>~(\Succ^n\Zero)$ executes with a number of steps
proportional to the size $n$ of its argument because it requires at least $n$
applications of the $\beta_{\Succ}$ rule before an answer can be returned.
\end{example}



%% file: sec_rec-mach.tex
\section{Recursion in an Abstract Machine}
\label{sec:rec-mach}

In order to explore the lower-level performance details of recursion, we can use
an \emph{abstract machine} for modeling an implementation of System T.  Unlike
the operational semantics given in \cref{fig:t-operation}, which has to search
arbitrarily deep into an expression for the next redex at every step, an
abstract machine explicitly includes this search in the computation itself.  As
such, every step of the machine can be applied by matching only on the top-level
form of the machine state, which more closely models a real implementation in a
machine where each step is performed sequentially in a fixed amount of time. Thus,
in an abstract machine instead of working with terms one works with
configurations of the form:
$$\cut{M}{E}$$
where $M$ is a term also called a producer, and $E$ is a continuation or evaluation context, also
called a consumer. A state, also called a command, puts together a producer and a consumer, so that the
output of $M$ is given as the input to $E$. 
We first present distinct abstract machines for call-by-name and -value,
we then smooth out the differences in the uniform
abstract machine. 

\subsection{Call-by-Name Abstract Machine}
 The call-by-name abstract machine 
 for System T is based on the Krivine
machine \cite{KrivineMachine}:%
\footnote{Our primary interest in abstract machines here is in the accumulation
  and use of continuations.  For simplicity, we leave out other common details
  sometimes specified by abstract machines, such as modeling a concrete
  representation of substitution and environments.}
\begin{align*}
  \cut{M~N}{E}
  &\srd
  \cut{M}{\app N E}
  \\
  \cut{\Rec M \As \{\dots\} }{E}
  &\srd
  \cut{M}{\Rec \{\dots\}  \With E}
  \\\\
  \cut{\fn x M}{\app N E}
  &\srd
  \cut{M\subst{x}{N}}{E}
  \\
  \bigcut
  {\Zero}
  {
    \begin{aligned}
      &\Rec
      \begin{aligned}[t]
        \{&
        \Zero \to N
        \\[-1ex]
        \mid&
        \Succ x \to y.N'
        \}
      \end{aligned}
      \\[-1ex]
      &\With E
    \end{aligned}
  }
  &\srd
  \cut{N}{E}
  \\
  \bigcut
  {\Succ M}
  {
    \begin{aligned}
      &\Rec
      \begin{aligned}[t]
        \{&
        \Zero \to N
        \\[-1ex]
        \mid&
        \Succ x \to y.N'
        \}
      \end{aligned}
      \\[-1ex]
      &\With E
    \end{aligned}
  }
  &\srd
  \bigcut
  {
    N'
    \bigsubs
    {
      \abigsub{x}{M},
      \abigsub{y}
      {
        \begin{aligned}
          &\Rec M \As
          \begin{aligned}[t]
            \{&
            \Zero \to N
            \\[-1ex]
            \mid&
            \Succ x \to y.N'
            \}
          \end{aligned}
        \end{aligned}
      }
    }
  }
  {E}
\end{align*}
The first two rules are \emph{refocusing} rules that move the attention of the
machine closer to the next reduction building a larger continuation: 
$\app N E$ corresponds to $E[\hole~N]$, and 
$\Rec \{\Zero \to N \mid \Succ x \to y.N'\} \With E$ corresponds to
$E[\Rec \hole \As \{\Zero \to N \mid \Succ x \to y.N'\} ]$. The latter three
rules are \emph{reduction} rules which correspond to steps of the operational
semantics in \cref{fig:t-operation}.  

\subsection{Call-by-Value Abstract Machine}
A CEK-style \cite{CEK}, call-by-value abstract machine for
System T -- which evaluates applications $M_1~M_2~\dots~M_n$ left-to-right to
match the call-by-value semantics in \cref{fig:t-operation}:
\begin{align*}
  \cut{M~N}{E}
  &\srd
  \cut{M}{\app N E}
  \\
  \cut{V}{\app R E}
  &\srd
  \cut{R}{V \comp E}
  \\
  \cut{V}{V' \comp E}
  &\srd
  \cut{V'}{\app V E}
  \\
  \cut{\Rec M \As \{\dots\}}{E}
  &\srd
  \cut{M}{\Rec \{\dots\} \With E}
  \\
  \cut{\Succ R}{E}
  &\srd
  \cut{R}{\Succ \comp E}
  \\
  \cut{V}{\Succ \comp E}
  &\srd
  \cut{\Succ V}{E}
  \\\\
  \cut{\fn x M}{\app V E}
  &\srd
  \cut{M\subst{x}{V}}{E}
  \\
  \bigcut
  {\Zero}
  {
    \begin{aligned}
      &\Rec
      \begin{aligned}[t]
        \{&
        \Zero \to N
        \\[-1ex]
        \mid&
        \Succ x \to y.N'
        \}
      \end{aligned}
      \\[-1ex]
      &\With E
    \end{aligned}
  }
  &\srd
  \cut{N}{E}
  \\
  \bigcut
  {\Succ V}
  {
    \begin{aligned}
      &\Rec
      \begin{aligned}[t]
        \{&
        \Zero \to N
        \\[-1ex]
        \mid&
        \Succ x \to y.N'
        \}
      \end{aligned}
      \\[-1ex]
      &\With E
    \end{aligned}
  }
  &\srd
  \bigcut
  {V}
  {
    \begin{aligned}
      &\Rec
      \begin{aligned}[t]
        \{&
        \Zero \to N
        \\[-1ex]
        \mid&
        \Succ x \to y.N'
        \}
      \end{aligned}
      \\[-1ex]
      &\With{} ((\fn y N') \comp E)
    \end{aligned}
  }
\end{align*}
where $R$ stands for a \emph{non-value} term.  Since the call-by-value
operational semantics has more forms of evaluation contexts, this machine has
additional refocusing rules for accumulating more forms of continuations
including applications of functions ($V \comp E$ corresponding to $E[V~\hole]$)
and the successor constructor ($\Succ \comp E$ corresponding to
$E[\Succ~\hole]$).  Also note that the final reduction rule for the $\Succ$ case
of recursion is different, accounting for the fact that recursion in
call-by-value follows a different order than in call-by-name.  
Indeed, the recursor
must explicitly accumulate and build upon a continuation, ``adding to'' the
place it returns to with every recursive call.
But otherwise,
the reduction rules are the same.

\subsection{Uniform Abstract Machine}
\label{sec:uniform_abstract_machine}

We now unite 
 both evaluation strategies with a common abstract machine, shown in
\cref{fig:t-mach-syntax}. As before, machine
configurations are of the form $$\cut{v}{e}$$
 which  put together a term $v$ and a continuation $e$ (often
referred  to as  a co\-term). However, both
 terms and continuations are more general than before. 
 Our uniform abstract machine is based on the sequent
calculus, a symmetric language reflecting many dualities of classical logic
 \cite{DualityOfComputation,CBVDualToCBN}.

Unlike the previous machines, continuations go beyond evaluation contexts
and include $\inp x \cut{v}{e}$,
which is a continuation
that binds its input value to $x$ and then steps to the machine state
$\cut{v}{e}$.  This new form allows us to express the additional
call-by-value evaluation contexts:
$V \comp E$ becomes 
$\inp x \cut{V}{\app x E}$, and
$\Succ \comp E$ is $\inp x \cut{\Succ x}{E}$.
Evaluation contexts are also more restrictive than before; only values can be pushed
on the calling stack. 
We  represent the application continuation $\app R E$
with a non-value argument $R$ by naming its partner---the generic value $V$---with $y$:
$ \inp y \cut{R}{\inp x \cut{y}{\app x E}}$.

The refocusing rules can be subsumed all together by extending  terms with
a dual form of $\tmu$-binding. The $\mu$-abstraction expression $\outp\alpha \cut{v}{e}$ 
binds its continuation to $\alpha$ and then steps to the machine state
$\cut{v}{e}$.  With $\mu$-bindings, all the refocusing rules for
call-by-name and -value  can be encoded in terms of $\mu$
and $\tmu$.
It is also not necessary to do these steps at run-time, but can all be done before execution
through a compilation step. Indeed,
 all of System T terms can be translated to a
smaller language, as shown in \cref{fig:t-compile} where $R$ stands for a non-value.
The target language of this compilation step becomes the language of the uniform abstract machine. This language does not include anymore 
applications  $M~N$
and the recursive term $\Rec M \As \{\dots\}$.  Also, unlike System T,
the syntax of terms and co\-terms depends on the
definition of values and co\-values;
$\Succ V$, $\Rec \{\dots\} \With E$, and call stacks $\app V E$ 
are valid in both call-by-name and -value, just with different definitions of
$V$ and $E$.
 General terms also include the $\mu$- and
$\tmu$-binders described above: $\outp \alpha c$ 
 is not a value in call-by-value, and $\inp x c$ is not a co\-value in
call-by-name.
So for example, $\Succ (\outp \alpha c)$ is not a legal term
in call-by-value, similarly $\app {(\outp \alpha c)} {\beta}$
is not a legal call stack in call-by-name. 

As with System T,
the notions of values and co-values drive the reduction rules, as shown in
 \cref{fig:t-mach-operation}.
In
particular, the $\mu$ and $\tmu$ rules will only substitute a value for a
variable or a co\-value for a co\-variable, respectively.  Likewise, the
$\beta_\to$ rule implements function calls, but taking the next argument value
off of a call stack and plugging it into the function.
The only remaining rules are $\beta_{\Zero}$ and $\beta_{\Succ}$ for reducing a
recursor when given a number constructed by $\Zero$ or $\Succ$.  While the
$\beta_{\Zero}$ is exactly the same as it was previously in both specialized
machine, notice how $\beta_{\Succ}$ is different.  Rather than committing to one
evaluation order or the other, $\beta_{\Succ}$ is neutral: the recursive
predecessor (expressed as the term
$\outp\alpha\cut{V}{\Rec\{\dots\}\With\alpha}$ on the right-hand side) is
neither given precedence (at the top of the command) nor delayed (by
substituting it for $y$).  Instead, this recursive predecessor is bound to $y$
with a $\tmu$-abstraction.  This way, the correct evaluation order can be
decided in the next step by either an application of $\mu$ or $\tmu$ reduction.

\begin{figure}
\emph{Commands ($c$), general terms ($v$), and general co\-terms ($e$)}:
\begin{gather*}
\begin{aligned}
  \<Command> \ni
  c &::= \cut{v}{e}
  &
  \<Term> \ni
  v,w &::= \outp \alpha c \Alt V
  &
  \<CoTerm> \ni
  e,f &::= \inp x c \Alt E
\end{aligned}
\end{gather*}

\emph{Call-by-name values ($V$) and evaluation contexts ($E$)}:
\begin{alignat*}{3}
  \<Value> &\ni{}&
  V,W
  &::= \outp \alpha c
  \Alt x
  \Alt \fn x v
  \Alt \Zero
  \Alt \Succ V 
  \\
  \<CoValue> &\ni{}&
  E,F
  &::= \alpha
  \Alt \app V E 
  \Alt \Rec \{\Zero \to v \mid \Succ x \to y.w\}  \With E
\end{alignat*}

\emph{Call-by-value values ($V$) and evaluation contexts ($E$)}:
\begin{alignat*}{3}
  \<Value> &\ni{}&
  V,W
  &::= x
  \Alt \fn x v
  \Alt \Zero
  \Alt \Succ V
  \\
  \<CoValue> &\ni{}&
  E,F
  &::= \inp x c
  \Alt \alpha
  \Alt \app V E
  \Alt \Rec \{\Zero \to v \mid \Succ x \to y.w\}    \With E
\end{alignat*}

\newcommand{\murule}{\mu}
\newcommand{\tmurule}{\tmu}
\emph{Operational reduction rules}:
\begin{align*}
  (\murule)&&
  \cut{\outp \alpha c}{E}
  &\srd
  c\subst{\alpha}{E}
  \\
  (\tmurule)&&
  \cut{V}{\inp x c}
  &\srd
  c\subst{x}{V}
  \\
  (\beta_\to)&&
  \cut{\fn x v}{\app V E}
  &\srd
  \cut{v\subst{x}{V}}{E}
  \\
  (\beta_{\Zero})&&
  \BIgcut
  {\Zero}
  {
    \begin{aligned}
      &\Rec
      \begin{aligned}[t]
        \{&
        \Zero \to v
        \\[-0.5ex]
        \mid&
        \Succ x \to y.w
        \}
      \end{aligned}
      \\[-0.5ex]
      &\With E
    \end{aligned}
  }
  &\srd
  \cut{v}{E}
  \\
  (\beta_{\Succ})&&
  \BIgcut
  {\Succ V}
  {
    \begin{aligned}
      &\Rec
      \begin{aligned}[t]
        \{&
        \Zero \to v
        \\[-0.5ex]
        \mid&
        \Succ x \to y.w
        \}
      \end{aligned}
      \\[-0.5ex]
      &\With E
    \end{aligned}
  }
  &\srd
  \BIgcut
  {\outp \alpha
    \BIgcut
    {V}
    {
      \begin{aligned}
        &\Rec
        \begin{aligned}[t]
          \{&
          \Zero \to v
          \\[-0.5ex]
          \mid&
          \Succ x \to y.w
          \}
        \end{aligned}
        \\[-0.5ex]
        &\With \alpha
      \end{aligned}
    }
  }
  {\inp y \cut{w\subst{x}{V}}{E}}
\end{align*}
\caption{Uniform, recursive abstract machine for System T.}
\label{fig:t-mach-syntax}
\label{fig:t-mach-operation}
\end{figure}

\begin{figure}
\small
\begin{gather*}
\begin{aligned}
  \trans{x} &\defeq x
  \\
  \trans{\fn x M} &\defeq \fn x \trans{M}
  \\
  \trans{\Zero} &\defeq \Zero
  \\
  \trans{\Succ V}
  &\defeq
  \Succ \trans{V}
  \\
  \trans{\Succ R}
  &\defeq
  \outp\alpha
  \cut{\trans{R}}{\inp x \cut{\Succ x}{\alpha}}
  \\
  \trans{M ~ N}
  &\defeq
  \outp\alpha
  \cut{\trans{M}}{\inp x \cut{\trans{N}}{\inp y \cut{x}{\app y \alpha}}}
  \\
  \trans{\Rec R \As \ \{\Zero {\to} M \mid \Succ x {\to} y. N\}}
  &\defeq
  \outp\alpha
  \cut
  {\trans{R}}
  {\Rec \{\Zero {\to} \trans{M} \mid \Succ x {\to} y. \trans{N}\} \With \alpha}
\end{aligned}
\end{gather*}
\caption{The translation from System T to the uniform abstract
  machine.}
\label{fig:t-compile}
\end{figure}

\begin{intermezzo}
We can now summarize how some basic concepts of recursion are directly modeled in our syntactic framework:
\begin{itemize}[-]
\item {\em With inductive data types, values are constructed and the
    consumer is a recursive process that uses the data}.  
 Natural numbers are terms or producers, and their use is a process which is triggered when the term becomes a value.
\item {\em Construction of data is finite and its consumption is (potentially)
    infinite, in the sense that there must be no limit to the size of the data
    that a consumer can process}.  We can only build values from a finite number
  of constructor applications.  However, the consumer does not know how big of
  an input it will be given, so it has to be ready to handle data structures of
  any size.  In the end, termination is preserved because only finite values are
  consumed.
\item {\em Recursion uses the data, rather than producing it}.  $\Rec$ is a
  co\-term, not a term.
\item {\em Recursion starts big, and potentially reduces down to a base case}.  As shown in
  the reduction rules, the recursor brakes down the data structure and
might end when
  the base case is reached.
\item {\em The values of a data structures are all independent from each other
    but the results of the recursion potentially depend on each other}.  In the
  reduction rule for the successor case, the result at a number $n$ might depend on
  the result at $n\mbox{-}1$.
\end{itemize}
\end{intermezzo}

\subsection{Examples of Recursion}
\label{sec:rec-examples}

By being restricted to capturing only (co\-)values, the $\mu$ and $\tmu$ rules
effectively implement the chosen evaluation strategy.  For example, consider the
application $(\lambda z. \Succ\Zero)~((\lambda x. x)~\Zero)$.  Call-by-name
evaluation will reduce the outer application first and return $\Succ\Zero$ right
away, whereas call-by-value evaluation will first reduce the inner application
$((\lambda x.x)~\Zero)$.  How is this different order of evaluation made explicit in the
abstract machine, which uses the same set of rules in both cases?  First,
consider the translation of
$\trans{(\lambda z. x)~((\lambda x. x)~y)}$:
\begin{align*}
  \trans{(\lambda z. x)~((\lambda x. x)~y)}
  &\defeq
  \outp\alpha
  \cut{\lambda z. x}
  {
    \inp f
    \cut
    {
      \outp\beta
      \cut
      {\lambda x. x}
      {\inp g \cut{y}{\inp y \cut{g}{\app y \beta}}}
    }
    {\inp z \cut{f}{\app{z}{\alpha}}}
  }
\end{align*}
To execute it, we need to put it in interaction with an actual context.  In our
case, we can simply use a co\-variable $\alpha$.  Call-by-name execution then
proceeds as:
\begin{spacing}
\begin{align*}
  &
  \cut
  {
    \outp\alpha
    \cut{\lambda z. x}
    {
      \inp f
      \cut
      {
        \outp\beta
        \cut
        {\lambda x. x}
        {\inp g \cut{y}{\inp y \cut{g}{\app y \beta}}}
      }
      {\inp z \cut{f}{\app{z}{\alpha}}}
    }
  }
  {\alpha}
  \\
  &\mapsto
  \cut{\lambda z. x}
  {
    \inp f
    \cut
    {
      \outp\beta
      \cut
      {\lambda x. x}
      {\inp g \cut{y}{\inp y \cut{g}{\app y \beta}}}
    }
    {\inp z \cut{f}{\app{z}{\alpha}}}
  }
  &&(\mu)
  \\
  &\mapsto
  \cut
  {
    \outp\beta
    \cut
    {\lambda x. x}
    {\inp g \cut{y}{\inp y \cut{g}{\app y \beta}}}
  }
  {\inp z \cut{\lambda z. x}{\app{z}{\alpha}}}
  &&(\tmu)
  \\
  &\mapsto
  \cut
  {\lambda z. x}
  {
    \app
    {
      \outp\beta
      \cut
      {\lambda x. x}
      {\inp g \cut{y}{\inp y \cut{g}{\app y \beta}}}
    }
    {\alpha}
  }
  &&(\tmu*)
  \\
  &\mapsto
  \cut{x}{\alpha}
  &&(\beta_\to)
\end{align*}
\end{spacing}
Whereas call-by-value execution proceeds as:
\begin{spacing}
\begin{align*}
  &
  \cut
  {
    \outp\alpha
    \cut{\lambda z. x}
    {
      \inp f
      \cut
      {
        \outp\beta
        \cut
        {\lambda x. x}
        {\inp g \cut{y}{\inp y \cut{g}{\app y \beta}}}
      }
      {\inp z \cut{f}{\app{z}{\alpha}}}
    }
  }
  {\alpha}
  \\
  &\mapsto
  \cut{\lambda z. x}
  {
    \inp f
    \cut
    {
      \outp\beta
      \cut
      {\lambda x. x}
      {\inp g \cut{y}{\inp y \cut{g}{\app y \beta}}}
    }
    {\inp z \cut{f}{\app{z}{\alpha}}}
  }
  &&(\mu)
  \\
  &\mapsto
  \cut
  {
    \outp\beta
    \cut
    {\lambda x. x}
    {\inp g \cut{y}{\inp y \cut{g}{\app y \beta}}}
  }
  {\inp z \cut{\lambda z. x}{\app{z}{\alpha}}}
  &&(\tmu)
  \\
  &\mapsto
  \cut
  {\lambda x. x}
  {
    \inp g
    \cut
    {y}
    {
      \inp y
      \cut
      {g}
      {\app{y}{\inp z \cut{\lambda z. x}{\app{z}{\alpha}}}}
    }
  }
  &&(\mu*)
  \\
  &\mapsto
  \cut
  {y}
  {
    \inp y
    \cut
    {\lambda x. x}
    {\app{y}{\inp z \cut{\lambda z. x}{\app{z}{\alpha}}}}
  }
  &&(\tmu)
  \\
  &\mapsto
  \cut
  {\lambda x. x}
  {\app{y}{\inp z \cut{\lambda z. x}{\app{z}{\alpha}}}}
  &&(\tmu)
  \\
  &\mapsto
  \cut
  {y}
  {\inp z \cut{\lambda z. x}{\app{z}{\alpha}}}
  &&(\beta_\to)
  \\
  &\mapsto
  \cut{\lambda z. x}{\app{y}{\alpha}}
  &&(\tmu)
  \\
  &\mapsto
  \cut{x}{\alpha}
  &&(\beta_\to)
\end{align*}
\end{spacing}
The first two steps are the same for either evaluation strategy.  Where the two
begin to diverge is in the third step (marked by a $*$), which is an interaction
between a $\mu$- and a $\tmu$-binder.  In call-by-name, the $\tmu$ rule takes
precedence (because a $\tmu$-co\-term is not a co\-value), which leads to the
next step which throws away the first argument, unevaluated.  In call-by-value,
the $\mu$ rule takes precedence (because a $\mu$-term is not a value), which
leads to the next step which evaluates the first argument.


Consider the System T definition of $\<plus>$ from \cref{ex:t-arith},
which is expressed by the machine term
\begin{align*}
  \<plus>
  &=
  \lambda x.\lambda y.\outp\beta
  \cut
  {x}
  {\Rec\{\Zero \to y \mid \Succ\blank \to z. \Succ z\}\With\beta}
\end{align*}

The application $\<plus>~2~3$ is then expressed as
$ \outp \alpha   \cut   {\<plus>}   {\app{2}{\app{3}{\alpha}}}$,
which is obtained by reducing some $\mu$- and $\tmu$-bindings in
advance.  Putting this term in the context $\alpha$, in call-by-value it
executes (eliding the branches of the $\Rec$-continuation, which are the same in
every following step) like so:
\begin{spacing}
\begin{align*}
  &
  \cut
  {
    \outp \alpha
    \cut
    {\<plus>}
    {\app{2}{\app{3}{\alpha}}}
  }
  {\alpha}
  \\
  &\mapsto
  \cut
  {\<plus>}
  {\app{2}{\app{3}{\alpha}}}
  \\
  &\mapsto
  \cut
  {
    \lambda y.\outp\beta
    \cut
    {2}
    {\Rec\{\Zero \to y \mid \Succ\blank \to z. \Succ z\}\With\beta}
  }
  {\app{3}{\alpha}}
  &&(\beta_\to)
  \\
  &\mapsto
  \cut
  {
    \outp\beta
    \cut
    {2}
    {\Rec\{\Zero \to 3 \mid \Succ\blank \to z. \Succ z\}\With\beta}
  }
  {\alpha}
  &&(\beta_\to)
  \\
  &\mapsto
  \cut
  {\Succ(\Succ\Zero)}
  {
    \Rec\{\Zero \to 3 \mid \Succ\blank \to z. \Succ z\}
    \With{\color{blue}\alpha}
  }
  &&(\mu)
  \\
  &\mapsto
  \cut
  {
    \outp\beta
    \cut
    {\Succ\Zero}
    {\Rec\{\dots\}\With\beta}
  }
  {
    \inp z \cut{\Succ z}{\color{blue}\alpha}
  }
  &&(\beta_{\Succ})
  \\
  &\mapsto
  \cut
  {\Succ\Zero}
  {
    \Rec\{\dots\}
    \With{\color{blue}\inp z \cut{\Succ z}{\alpha}}
  }
  &&(\mu)
  \\
  &\mapsto
  \cut
  {
    \outp\beta
    \cut
    {\Zero}
    {\Rec\{\dots\}\With\beta}
  }
  {
    \inp{z'}\cut{\Succ z'}{\color{blue}\inp z\cut{\Succ z}{\alpha}}
  }
  &&(\beta_{\Succ})
  \\
  &\mapsto
  \cut
  {\Zero}
  {
    \Rec\{\dots\}
    \With{\color{blue}\inp{z'}\cut{\Succ z'}{\inp z \cut{\Succ z}{\alpha}}}
  }
  &&(\mu)
  \\
  &\mapsto
  \cut
  {3}
  {\color{blue}\inp{z'}\cut{\Succ z'}{\inp z \cut{\Succ z}{\alpha}}}
  &&(\beta_{\Zero})
  \\
  &\mapsto
  \cut
  {\Succ 3}
  {\inp z \cut{\Succ z}{\alpha}}
  &&(\tmu)
  \\
  &\mapsto
  \cut{\Succ(\Succ 3)}{\alpha}
  &&(\tmu)
\end{align*}
\end{spacing}
Notice how this execution shows how, during the recursive traversal of the data
structure, the return continuation of the recursor is updated (in blue) to keep
track of the growing context of pending operations, which must be fully
processed before the final value of 5 ($\Succ(\Succ 3)$) can be returned to the
original caller ($\alpha$). This update is implicit in the
$\lambda$-calculus-based System T, but becomes explicit in the abstract machine.
In contrast, call-by-name only computes numbers as far as they are needed,
otherwise stopping at the outermost constructor.  The call-by-name execution of
the above command proceeds as follows, after fast-forwarding to the first application of $\beta_{\Succ}$:
\begin{spacing}
\begin{align*}
&
  \cut
  {
    \outp \alpha
    \cut
    {\<plus>}
    {\app{2}{\app{3}{\alpha}}}
  }
  {\alpha}
\\
& \dmapsto 
  \cut
  {\Succ(\Succ\Zero)}
  {
    \Rec\{\Zero \to 3 \mid \Succ\blank \to z. \Succ z\}
    \With{\alpha}
  }
  \\
  &\mapsto
  \cut
  {
    \outp\beta
    \cut
    {\Succ\Zero}
    {\Rec\{\Zero \to 3 \mid \Succ\blank \to z. \Succ z\}\With\beta}
  }
  {
    \inp z \cut{\Succ z}{\alpha}
  }
  &&(\beta_{\Succ})
  \\
  &\mapsto
  \cut
  {
    \Succ
    (
    \outp\beta
    \cut
    {\Succ\Zero}
    {\Rec\{\Zero \to 3 \mid \Succ\blank \to z. \Succ z\}\With\beta}
    )
  }
  {\alpha}
  &&(\tmu)
\end{align*}
\end{spacing}
Unless $\alpha$ demands to know something about the predecessor of this number,
the term $\outp\beta\cut{\Succ\Zero}{\Rec\{\dots\}\With\beta}$ will not be
computed. \\

Now consider $\<pred>~(\Succ(\Succ\Zero))$, which can be expressed in the
machine as:
\begin{align*}
  &
  \outp \alpha
  \cut
  {\<pred>}
  {\app{\Succ(\Succ\Zero)}{\alpha}}
  \\
  \<pred>
  &=
  \lambda x.
  \outp\beta\cut{x}{\Rec\{\Zero\to\Zero\mid\Succ x\to z.x\}\With\beta}
\end{align*}
In call-by-name it executes with respect to $\alpha$ like so:
\begin{spacing}
\begin{align*}
  &
  \cut
  {
    \outp \alpha
    \cut
    {\<pred>}
    {\app{\Succ(\Succ\Zero)}{\alpha}}
  }
  {\alpha}
  \\
  &\mapsto
    \cut
    {\<pred>}
    {\app{\Succ(\Succ\Zero)}{\alpha}}
  &&(\mu)
  \\
  &\mapsto
  \cut
  {
    \outp\beta
    \cut
    {\Succ(\Succ\Zero)}
    {\Rec\{\Zero\to\Zero\mid\Succ x\to z.x\}\With\beta}
  }
  {\alpha}
  &&(\beta_\to)
  \\
  &\mapsto
  \cut
  {\Succ(\Succ\Zero)}
  {\Rec\{\Zero\to\Zero\mid\Succ x\to z.x\}\With\alpha}
  &&(\mu)
  \\
  &\mapsto
  \cut
  {
    \outp\beta
    \cut
    {\Succ\Zero}
    {\Rec\{\Zero\to\Zero\mid\Succ x\to z.x\}\With\beta}
  }
  {\inp z \cut{\Succ\Zero}{\alpha}}
  &&(\beta_{\Succ})
  \\
  &\mapsto
  \cut{\Succ\Zero}{\alpha}
  &&(\tmu)
\end{align*}
\end{spacing}
Notice how, after the first application of $\beta_{\Succ}$, the computation
finishes in just one $\tmu$ step, even though we began recursing on the number
2.  In call-by-value instead, we have to continue with the recursion even though
its result is not needed.  Fast-forwarding to the first application of the
$\beta_{\Succ}$ rule, we have:
\begin{spacing}
\begin{align*}
  &
  \cut
  {\Succ(\Succ\Zero)}
  {
    \Rec \{\Zero\to\Zero\mid\Succ x\to z.x\}
    \With {\color{blue}\alpha}
  }
  \\
  &\mapsto
  \cut
  {
    \outp\beta
    \cut
    {\Succ\Zero}
    {\Rec\{\dots\}\With\beta}
  }
  {\inp z \cut{\Succ\Zero}{\color{blue}\alpha}}
  &&(\beta_{\Succ})
  \\
  &\mapsto
  \cut
  {\Succ\Zero}
  {
    \Rec \{\dots\}
    \With {\color{blue}\inp z \cut{\Succ\Zero}{\alpha}}
  }
  &&(\mu)
  \\
  &\mapsto
  \cut
  {
    \outp\beta
    \cut
    {\Zero}
    {\Rec\{\dots\}\With\beta}
  }
  {\inp z\cut{\Zero}{\color{blue}\inp z\cut{\Succ\Zero}{\alpha}}}
  &&(\beta_{\Succ})
  \\
  &\mapsto
  \cut
  {\Zero}
  {
    \Rec\{\dots\}
    \With{\color{blue}\inp z\cut{\Succ\Zero}{\inp z\cut{\Succ\Zero}{\alpha}}}
  }
  &&(\mu)
  \\
  &\mapsto
  \cut
  {\Zero}
  {\color{blue}\inp z\cut{\Succ\Zero}{\inp z\cut{\Succ\Zero}{\alpha}}}
  &&(\beta_{\Zero})
  \\
  &\dmapsto
  \cut{\Succ\Zero}{\alpha}
  &&(\tmu)
\end{align*}
\end{spacing}

\subsection{Recursion vs Iteration: Expressiveness and Efficiency}
\label{sec:rec-vs-iter}
\label{sec:rec-vs-iter-mach}

Recall how the recursor performs two jobs at the same time: finding the
predecessor of a natural number as well as calculating the recursive result
given for the predecessor.  These two functionalities can be captured separately
by continuations that perform shallow \emph{case analysis} and the primitive
\emph{iterator}, respectively.  Rather than including them as primitives, both
can be expressed as syntactic sugar in the form macro-expansions in the language
of the abstract machine like so:
\begin{spacing}
\begin{align*}
\begin{aligned}
  &
  \Case
  \begin{alignedat}[t]{2}
    \{&
    \Zero &&\to v
    \\
    \mid&
    \Succ x &&\to w
    \}
  \end{alignedat}
  \\
  &\With E
\end{aligned}
&\defeq\!\!
\begin{aligned}
  &\Rec
  \begin{alignedat}[t]{2}
    \{&
    \Zero &&\to v
    \\
    \mid&
    \Succ x &&\to \blank\,. w
    \}
  \end{alignedat}
  \\
  &\With E
\end{aligned}
&~~
\begin{aligned}
  &
  \Iter
  \begin{alignedat}[t]{2}
    \{&
    \Zero &&\to v
    \\
    \mid&
    \Succ &&\to x.w
    \}
  \end{alignedat}
  \\
  &\With E
\end{aligned}
&\defeq\!\!
\begin{aligned}
  &\Rec
  \begin{alignedat}[t]{2}
    \{&
    \Zero &&\to v
    \\
    \mid&
    \Succ \blank &&\to x.w
    \}
  \end{alignedat}
  \\
  &\With E
\end{aligned}
\end{align*}
\end{spacing}
The only cost of this encoding of $\Case$ and $\Iter$ is an unused variable
binding, which is easily optimized away.  In practice, this encoding of
iteration will perform exactly the same as if we had taken $\Iter$ation as a
primitive.

While it is less obvious, going the other way is still possible.  It is well
known that primitive recursion can be encoded as a macro-expansion of iteration
using pairs.  The usual macro-expansion in System T is:
\begin{spacing}
\begin{align*}
  \begin{aligned}
    &
    \Rec M \As
    \\
    &\quad
    \begin{alignedat}{2}
      &\{~
      \Zero &&\to N
      \\
      &\mid
      \Succ x &&\to y. N'
      \}
    \end{alignedat}
  \end{aligned}
  &\defeq
  \begin{aligned}
    &\Snd~
    (
    \Iter M \As
    \\
    &\qquad\quad
    \begin{alignedat}{2}
      &\{~
      \Zero &&\to (\Zero, N)
      \\
      &\mid
      \Succ &&\to (x,y).\, (\Succ x, N')
      \}
      )
    \end{alignedat}
  \end{aligned}
\end{align*}
\end{spacing}
The trick to this encoding is to use $\Iter$ to compute \emph{both} a
reconstruction of the number being iterated upon (the first component of the
iterative result) alongside the desired result (the second component).  Doing
both at once gives access to the predecessor in the $\Succ$ case, which can be
extracted from the first component of the previous result (given by the variable
$x$ in the pattern-match $(x,y)$).

To express this encoding in the abstract machine, we need to extend it with
pairs, which look like \cite{CBVDualToCBN}:
\begin{align*}
  \cut{(v, w)}{\Fst e} &\srd \cut{v}{E}
  &
  \cut{(v, w)}{\Snd e} &\srd \cut{w}{E}
  &
  (\beta_\times)
\end{align*}
In the syntax of the abstract machine, the analogous encoding of a $\Rec$
continuation as a macro-expansion looks like this:
\begin{spacing}
\begin{align*}
\begin{aligned}
  &\Rec
  \begin{alignedat}[t]{2}
    \{&
    \Zero &&\to v
    \\
    \mid&
    \Succ x &&\to y.w
    \}
  \end{alignedat}
  \\
  &\With E
\end{aligned}
&\defeq
\begin{aligned}
  &\Iter
  \begin{alignedat}[t]{2}
    \{&
    \Zero &&\to (\Zero, v)
    \\
    \mid&
    \Succ &&\to (x,y). (\Succ x, w)
    \}
  \end{alignedat}
  \\
  &\With{} \Snd E
\end{aligned}
\end{align*}
\end{spacing}
Since the inductive case $w$ might refer to both the predecessor $x$ and the
recursive result for the predecessor (named $y$), the two parts must be
extracted from the pair returned from iteration.  Here we express this
extraction in the form of pattern-matching, which is shorthand for:
\begin{equation*}
\label{rec_vs_iter}
  (x,y). (v_1,v_2)
  \defeq
  z.
  \outp\alpha
  \cut{z}{\Fst(\inp x \cut{z}{\Snd(\inp y \cut{(v_1,v_2)}{\alpha})})}
\end{equation*}
Note that the recursor continuation is tasked with passing its final result to
$E$ once it has finished.  In order to give this same result to $E$, the
encoding has to extract the second component of the final pair before passing it
to $E$, which is exactly what $\Snd E$ expresses.

Unfortunately, this encoding of recursion is not always as efficient as the
original.  If the recursive parameter $y$ is never used (such as in the
$\<pred>$ function), then $\Rec$ can provide an answer without computing the
recursive result.  However, when encoding $\Rec$ with $\Iter$, the result of the
recursive value must always be computed before an answer is seen, regardless of
whether or not $y$ is needed.  As such, redefining $\<pred>$ using $\Iter$ in
this way changes it from a constant time ($O(1)$) to a linear time ($O(n)$)
function.  Notice that this difference in cost is only apparent in call-by-name,
which can be asymptotically more efficient when the recursive $y$ is not needed
to compute $N'$, as in $\<pred>$.  In call-by-value, the recursor must descend
to the base case anyway before the incremental recursive steps are propagated
backward.  That is to say, the call-by-value $\Rec$ has the same asymptotic
complexity as its encoding via $\Iter$.

\subsection{Types and Correctness}

\begin{figure}
\begin{gather*}
  \infer[\<Cut>]
  {\Gamma \entails \cut{v}{e} \contra}
  {
    \Gamma \entails v \givestype A
    &
    \Gamma \entails e \takestype A
  }
  \\[\ruleskip]
  \axiom[\<VarR>]{\Gamma, x \givestype A \entails x \givestype A}
  \qquad
  \axiom[\<VarL>]{\Gamma, \alpha \takestype A \entails \alpha \takestype A}
  \\[\ruleskip]
  \infer[\<ActR>]
  {\Gamma \entails \outp\alpha c \givestype A}
  {\Gamma, \alpha \takestype A \entails c \contra}
  \qquad
  \infer[\<ActL>]
  {\Gamma \entails \inp x c \takestype A}
  {\Gamma, x \givestype A \entails c \contra}
  \\[\ruleskip]
  \infer[{\to}R]
  {\Gamma \entails \fn x v \givestype A \to B}
  {\Gamma, x \givestype A \entails v \givestype B}
  \qquad
  \infer[{\to}L]
  {\Gamma \entails \app v e \takestype A \to B}
  {
    \Gamma \entails v \givestype A
    &
    \Gamma \entails e \takestype B
  }
  \\[\ruleskip]
  \axiom[{\Nat}R_{\Zero}]
  {\Gamma \entails \Zero \givestype \Nat}
  \qquad
  \infer[{\Nat}R_{\Succ}]
  {\Gamma \entails \Succ V \givestype \Nat}
  {\Gamma \entails V \givestype \Nat}
  \\[\ruleskip]
  \infer[{\Nat}L]
  {
    \Gamma \entails
    \Rec \{\Zero \to v \mid \Succ x \to y.w\} \With E
    \takestype \Nat
  }
  {
    \Gamma \entails v \givestype A
    &
    \Gamma, x \givestype \Nat, y \givestype A \entails w \givestype A
    &
    \Gamma \entails E \takestype A
  }
\end{gather*}
\caption{Type system for the uniform, recursive abstract machine.}
\label{fig:t-mach-type-system}
\end{figure}

We can also give a type system directly for the abstract machine, as shown in
\cref{fig:t-mach-type-system}.  This system has judgments for assigning types to
terms as usual: $\Gamma \entails v \givestype A$ says $v$ produces an output of
type $A$.  In addition, there are also judgments for assigning types to
co\-terms ($\Gamma \entails e \takestype A$ says $e$ consumes an input of type
$A$) and commands ($\Gamma \entails c \contra$ says $c$ is safe to compute, and
does not produce or consume anything).

This type system ensures that the machine itself is type safe: well-typed,
executable commands don't get stuck while in the process of computing a final
state.  For our purposes, we will only execute commands $c$ with just one free
co\-variable (say $\alpha$), representing the final, top-level continuation.
Thus, well-typed executable commands will satisfy
$\alpha\takestype\Nat \entails c \contra$.  The only final states of these
executable commands have the form $\cut{\Zero}{\alpha}$, which sends 0 to the
final continuation $\alpha$, or $\cut{\Succ V}{\alpha}$, which sends the
successor of some $V$ to $\alpha$.  But the type system ensures more than just
type safety: all well-typed programs will eventually terminate.  That's because
$\Rec$-expressions, which are the only form of recursion in the language, always
decrement their input by 1 on each recursive step.  So together, every
well-typed executable command will eventually (termination) reach a valid final
state (type safety).
\begin{restatable}[Type safety \& Termination]{theorem}
  {thmtsafetytermination}
\label{thm:t-type-safety}
\label{thm:t-termination}
For any command $c$ of the recursive abstract machine, if
$\alpha\takestype\Nat \entails c \contra$ then $c \srds \cut{\Zero}{\alpha}$ or
$c \srds \cut{\Succ V}{\alpha}$ for some $V$.
\end{restatable}
The truth of this theorem follows directly from the latter development in
\cref{sec:safety&termination}, since it is a special case of
\cref{thm:mach-corec-type-safety,thm:mach-corec-termination}.

\begin{intermezzo}
\label{rm:two-sided-sequent}

Since our abstract machine is based on the logic of Gentzen's sequent calculus
\cite{Gentzen1935UULS1}, the type system in \cref{fig:t-mach-type-system} too
can be viewed as a term assignment for a particular sequent calculus.  In
particular, the statement $v \givestype A$ corresponds to a proof that $A$ is
true.  Dually $e \takestype A$ corresponds to a proof that $A$ is false, and
hence the notation, which can be understood as a built-in negation $-$ in
$e : -A$.  As such, the built-in negation in every $e \takestype A$ (or
$\alpha \takestype A$) can be removed by swapping between the left- and
right-hand sides of the turnstyle ($\entails$), so that $e \takestype A$ on the
right becomes $e : A$ on the left, and $\alpha \takestype A$ on the left becomes
$\alpha : A$ on the right.  Doing so gives a conventional two-sided sequent
calculus as in \cite{SequentMachines,SequentTutorial}, where the rules labeled
$L$ with conclusions of the form
$x_i \givestype B_i, \alpha_j \takestype C_j \entails e \takestype A$ correspond
to left rules of the form $x_i:B_i \mid e : A \entails \alpha_j:C_j$ in the
sequent calculus.
\end{intermezzo}





%% file: sec_corec-mach.tex
\section{Corecursion in an Abstract Machine}
\label{sec:corec-mach}


\begin{figure}
\begin{align*}
\<Type> \ni   A,B &::= A \to B \Alt \Nat
   \mid \hi{\Stream A}
\end{align*}
\begin{gather*}
  \infer[{\Stream}L_{\Head}]
  {\Gamma \entails \Head E \takestype \Stream A}
  {\Gamma \entails E \takestype A}
  \quad
  \infer[{\Stream}L_{\Head}]
  {\Gamma \entails \Tail E \takestype \Stream A}
  {\Gamma \entails E \takestype \Stream A}
  \\[\ruleskip]
  \infer[{\Stream}R]
  {
    \Gamma \entails
    \CoRec \{\Head \alpha \to e \mid \Tail \beta \to \gamma.f\} \With V
    \givestype \Stream A
  }
  {
    \Gamma, \alpha \takestype A \entails e \takestype B
    &
    \Gamma, \beta \takestype \Stream A, \gamma \takestype B
    \entails f \takestype B
    &
    \Gamma \entails V \givestype B
  }
\end{gather*}
\caption{Typing rules for streams in the uniform, (co\-)recursive abstract
  machine.}
\label{fig:stream-mach-type-system}
\end{figure}

\begin{figure}
\emph{Commands ($c$), general terms ($v$), and general co\-terms ($e$)}:
\begin{gather*}
\begin{aligned}
  \<Command> \ni
  c &::= \cut{v}{e}
  &
  \<Term> \ni
  v,w &::= \outp \alpha c \Alt V
  &
  \<CoTerm> \ni
  e,f &::= \inp x c \Alt E
\end{aligned}
\end{gather*}
\emph{Call-by-name values ($V$) and evaluation contexts ($E$)}:
\begin{alignat*}{3}
  \<Value> &\ni{}&
  V,W
  &::= \outp \alpha c
  \Alt x
  \Alt \fn x v
  \Alt \Zero
  \Alt \Succ V 
  \Alt \hi{\CoRec \{\Head\alpha \to e \mid \Tail\beta \to \gamma.f\} \With v}
  \\
  \<CoValue> &\ni{}&
  E,F &::=  \alpha
  \Alt \app V E
  \Alt \Rec \{\Zero \to v \mid \Succ x \to y.w\}  \With E
  \Alt  \hi{\Head E}
  \Alt  \hi{\Tail E}
\end{alignat*}
\emph{Call-by-value values ($V$) and evaluation contexts ($E$)}:
\begin{alignat*}{3}
  \<Value> &\ni{}&
  V,W
  &::= x
  \Alt \fn x v
  \Alt \Zero
  \Alt \Succ V
  \Alt \hi{\CoRec  \{\Head\alpha \to e \mid \Tail\beta \to \gamma.f\} \With V}
  \\
  \<CoValue> &\ni{}&
  E,F &::= \inp x c
  \Alt \alpha
  \Alt \app V E
  \Alt \Rec \{\Zero \to v \mid \Succ x \to y.w\}    \With E
  \Alt \hi{\Head e}
  \Alt \hi{\Tail e}
\end{alignat*}
\newcommand{\murule}{\mu}
\newcommand{\tmurule}{\tmu}
\emph{Operational rules:}
\begin{align*}
 (\murule)&&
  \cut{\outp \alpha c}{E}
  &\srd
  c\subst{\alpha}{E}
  \\
  (\tmurule)&&
  \cut{V}{\inp x c}
  &\srd
  c\subst{x}{V}
  \\
  (\beta_\to)&&
  \cut{\fn x v}{\app V E}
  &\srd
  \cut{v\subst{x}{V}}{E}
  \\
  (\beta_{\Zero})&&
  \BIgcut
  {\Zero}
  {
    \begin{aligned}
      &\Rec
      \begin{aligned}[t]
        \{&
        \Zero \to v
        \\[-0.5ex]
        \mid&
        \Succ x \to y.w
        \}
      \end{aligned}
      \\[-0.5ex]
      &\With E
    \end{aligned}
  }
  &\srd
  \cut{v}{E}
  \\
  (\beta_{\Succ})&&
  \BIgcut
  {\Succ V}
  {
    \begin{aligned}
      &\Rec
      \begin{aligned}[t]
        \{&
        \Zero \to v
        \\[-0.5ex]
        \mid&
        \Succ x \to y.w
        \}
      \end{aligned}
      \\[-0.5ex]
      &\With E
    \end{aligned}
  }
  &\srd
  \BIgcut
  {\outp \alpha
    \BIgcut
    {V}
    {
      \begin{aligned}
        &\Rec
        \begin{aligned}[t]
          \{&
          \Zero \to v
          \\[-0.5ex]
          \mid&
          \Succ x \to y.w
          \}
        \end{aligned}
        \\[-0.5ex]
        &\With \alpha
      \end{aligned}
    }
  }
  {\inp y \cut{w\subst{x}{V}}{E}}
\\
  (\beta_{\Head}) &&
\hi{
  \BIgcut
  {
    \small
    \begin{aligned}
      &
      \begin{aligned}[t]
        \CoRec
        \{&
        \Head\alpha \to e
        \\[-1ex]
        \mid&
        \Tail\beta \to \gamma.f\}
      \end{aligned}
      \\[-1ex]
      &\With V
    \end{aligned}
  }
  {\Head E} }
  &\srd 
  \hi{\cut{V}{e\subst{\alpha}{E}}}
  \\
  (\beta_{\Tail}) && \hi{
  \BIgcut
  {
    \small
    \begin{aligned}
      &
      \begin{aligned}[t]
        \CoRec
        \{&
        \Head\alpha \to e
        \\[-1ex]
        \mid&
        \Tail\beta \to \gamma.f\}
      \end{aligned}
      \\[-1ex]
      &\With V
    \end{aligned}
  }
  {\Tail E}
}
  &\srd
\hi{
  \BIgcut
  {\outp\gamma\cut{V}{f\subst{\beta}{E}}}
  {
    \inp x
    \BIgcut
    {
      \small
      \begin{aligned}
        &
        \begin{aligned}[t]
          \CoRec
          \{&
          \Head\alpha \to e
          \\[-1ex]
          \mid&
          \Tail\beta \to \gamma.f\}
        \end{aligned}
        \\[-1ex]
        &\With x
      \end{aligned}
    }
    {E}
  }}
\end{align*}
\caption{Uniform, (co\-)recursive abstract machine.}
\label{fig:stream-mach-syntax}
\label{fig:stream-mach-operation}
\end{figure}

Instead of coming up with an
extension of System T with corecursion, and then define an abstract machine
we start directly with the abstract machine which we obtain by
applying duality.
As a prototypical example of a co\-inductive type, we consider
infinite streams of values, chosen for
their familiarity (other co\-inductive types work just as well),
which we represent by the type $\Stream A$, as given in
\cref{fig:stream-mach-type-system}. 

%
%
The intention is that  $\Stream A$ is
roughly dual to $\Nat$, and so we will flip
the roles of terms and co\-terms belonging to streams.  In contrast with $\Nat$,
which has constructors for building values, $\Stream A$ has two
\emph{destructors} for building co\-values.  First, the co\-value $\Head E$ (the
base case dual to $\Zero$) projects out the first element of its given stream
and passes its value to $E$.  Second, the co\-value $\Tail E$ (the co\-inductive
case dual to $\Succ V$) discards the first element of the stream and passes the
remainder of the stream to $E$. The \emph{co\-recursor} is defined by dualizing
the recursor, whose general form is:
\begin{alignat*}{10}
  \Rec{}
  &\{
  \text{\it{base case}} &{}\to v
  &{}\mid{}&
  \text{\it{inductive case}} &{}\to y.w
  \}
  &{}\With E
  \\
  \CoRec{}
  &\{
  \text{\it{base case}} &{}\to e
  &{}\mid{}&
  \text{\it{co\-inductive case}} &{}\to \gamma.f
  \}
  &{}\With V
\end{alignat*}
Notice how the internal seed $V$ corresponds to the return continuation $E$.  In
the base case of the recursor, term $v$ is sent  to the current value of the
continuation $E$.  Dually, in the base case of the corecursor, the coterm $e$
receives the current value of the internal seed $V$.  In the recursor's inductive case,
$y$ receives the result of the next recursive step (\ie
the predecessor of the current one).  Whereas  in the co\-recursor's
co\-inductive case, $\gamma$ sends the
updated seed to the next co\-recursive step (\ie the tail of the current one).
The two cases of the recursor match the patterns $\Zero$ (the base case) and
$\Succ x$ (the inductive case).  Analogously, the co\-recursor matches against
the two possible co\-patterns: the base case is $\Head\alpha$, and the
co\-inductive case is $\Tail\beta$.  So the full form of the stream
co\-recursor is:
\[
  \CoRec{} \{\Head\alpha \to e \mid \Tail\beta \to \gamma.f\} \With V
\]


The uniform abstract machine is  given  in \cref{fig:stream-mach-operation},
where we   highlight the
extensions.
Note how
the co\-recursor
generates (on the fly) the values of the stream using $V$ as an incremental
accumulator or seed, saving the progress made through the stream so far.  In
particular, the base case $\Head\alpha \to e$ matching the head projection just
passes the accumulator to $e$, which (may) compute the current element and send it
to $\alpha$.  The co\-recursive case $\Tail\beta \to \gamma.f$ also passes the
accumulator to $f$, which may return an updated accumulator (through $\gamma$)
\emph{or} circumvent further co\-recursion by returning another stream directly
to the remaining projection (via $\beta$).
As with the syntax, the operational semantics is roughly symmetric to the rules for
natural numbers, where
roles of terms and co\-terms have been flipped.

\begin{intermezzo}

We review how the basic concepts of
corecursion are reflected in our syntactic framework.
Note how these observations are dual to the basic concepts
of recursion. 
\begin{itemize}[-]
\item 
{\em With coinductive data types, covalues are constructed and the
producer is a process that generates the data}.
Observers of streams are constructed via the
head and tail projections,
and their construction is a process.
\item
{\em Use of codata is finite and its creation is (potentially)
infinite in the sense that
    there must be no limit to the size of the codata
    that a producer can process}.
 We can only build covalues from a finite number
  of destructor applications.  However, the producer
does not know how big of
  a request it will be given, so it has to be ready to handle codata structures of
  any size.  In the end, termination is preserved because only finite covalues are
  consumed.
\item {\em Corecursion produces the data, rather than using it}.
 $\CoRec$ is a
  term, not a coterm.
\item {\em Corecursion starts from a seed and potentially
produces bigger and bigger internal
states};
the corecursor brakes down the codata structure and might end when
  the base case is reached.
\item {\em The values of a codata structure potentially depend on each other}. The n-th value of a stream might depend on the value at the index  $n\mbox{-}1$. 
\end{itemize}
\end{intermezzo}

\subsection{Examples of corecursion}
\label{sec:corec-examples}

We represent the streams $x, x, x, x, \dots$  and
$x, f x, f^2 x, \cdots$ as:
\begin{spacing}
\begin{align*}
\begin{aligned}[t]
  \<always> \ x  
  =
 \CoRec{}
  \{&
  \Head \alpha \to \alpha
  \\
  \mid&
  \Tail \_ \to \gamma.\gamma
  \}
  \\
  \With{}& x
\end{aligned}
\qquad
\begin{aligned}[t]
  \<repeat>  \  f \ x  
  =
 \CoRec{}
  \{&
  \Head \alpha \to \alpha
  \\
  \mid&
  \Tail  \_ \to \gamma.\inp x \cut {f}{ \app x \gamma}
  \}
  \\
  \With{}& x
\end{aligned}
\end{align*}
\end{spacing}
So when an observer asks $\<always> \ x$ ($\<repeat>\ f \ x$) for its head element (matching the
co\-pattern $\Head\alpha$), $\<always> \ x$ ($\<repeat> \ f \ x$) returns (to $\alpha$) the current
value of the seed $x$.  Otherwise, when an observer asks for its tail (matching
the co\-pattern $\Tail\beta$), $\<always> \ x$  ($\<repeat> \ f \ x$) continues co\-recursing with the
same seed $x$ ($\outp \gamma \cut{f}{\app x \gamma}$).

The infinite streams containing all zeroes, and the infinite stream of all
natural numbers counting up from 0, are then represented as:
\begin{align*}
  \<zeroes> = \outp\alpha\cut{\<always>}{\app\Zero\alpha} \\
  \<nats> =  \outp\alpha\cut{\<repeat>}{\app {succ}{ \app \Zero \alpha}}
\end{align*}
where $\<succ>$ is  defined as
$\lambda x. \outp \alpha \cut{\Succ x} \alpha$.

To better understand execution, let's trace how computation works in both call-by-name and -value.  
Asking for the third element of $\<zeroes>$ proceeds with the following
calculation in call-by-value (we let $\<always>_x$ stands for the stream with $x$ being the seed) :
\begin{spacing}
\begin{align*}
  &
  \cut{\<zeroes>}{\Tail(\Tail(\Head\alpha))}
  \\
  &\srd
  \cut{\<always>}{\app{\Zero}{\Tail(\Tail(\Head\alpha))}}
  &(\mu)
  \\
  &\srd
  \cut{\<always>_{\Zero}}{\Tail(\Tail(\Head\alpha))}
  &(\beta_\to)
  \\
  &\srd
  \cut
  {\outp\gamma \cut{\Zero}{\gamma}}
  {\inp x \cut{\<always>_x}{\Tail(\Head\alpha)}}
  &(\beta_{\Tail})
  \\
  &\srd
  \cut{\Zero}{\inp x \cut{\<always>_x}{\Tail(\Head\alpha)}}
  &(\mu)
  \\
  &\srd
  \cut{\<always>_{\Zero}}{\Tail(\Head\alpha)}
  &(\tmu)
  \\
  &\srds
  \cut{\<always>_{\Zero}}{\Head\alpha}
  &(\beta_{\Tail}\mu\tmu)
  \\
  &\srd
  \cut{\Zero}{\alpha}
  &(\beta_{\Head})
\end{align*}
\end{spacing}
In contrast, notice how the same calculation in call-by-name builds up a delayed
computation in the seed:
\begin{spacing}
\begin{align*}
  &
  \cut{\<zeroes>}{\Tail(\Tail(\Head\alpha))}
  \\
  &\srds
  \cut{\<always>_{\Zero}}{\Tail(\Tail(\Head\alpha))}
  &(\mu\beta_\to)
  \\
  &\srd
  \cut
  {\outp\gamma \cut{\Zero}{\gamma}}
  {\inp x \cut{\<always>_x}{\Tail(\Head\alpha)}}
  &(\beta_{\Tail})
  \\
  &\srd
  \cut{\<always>_{\outp\gamma\cut{\Zero}{\gamma}}}{\Tail(\Head\alpha)}
  &(\tmu)
  \\
  &\srds
  \cut
  {\<always>_{\outp{\gamma'}\cut{\outp\gamma\cut{\Zero}{\gamma}}{\gamma'}}}
  {\Head\alpha}
  &(\beta_{\Tail}\tmu)
  \\
  &\srd
  \cut
  {\outp{\gamma'}\cut{\outp\gamma\cut{\Zero}{\gamma}}{\gamma'}}
  {\alpha}
  &(\beta_{\Head})
  \\
  &\srd
  \cut
  {\outp\gamma\cut{\Zero}{\gamma}}
  {\alpha}
  &(\mu)
  \\
  &\srd
  \cut{\Zero}{\alpha}
  &(\mu)
\end{align*}
\end{spacing}

Consider the function that produces a stream that counts down from some initial
number $n$ to 0, and then staying at 0.  Informally, this function can be
understood as:
\begin{spacing}
\begin{align*}
  \<countDown> &: \Nat \to \Stream \Nat
  \\
  \<countDown>~n &= n, n-1, n-2, \dots, 3, 2, 1, 0, 0, 0, \dots
\end{align*}
\end{spacing}
It is  formally defined  like so:
\begin{spacing}
\begin{align*}
\begin{aligned}[t]
  \<countDown> \ n 
  =
 \CoRec{}
  \{&
  \Head \alpha \to \alpha
  \\
  \mid&
  \Tail  \_ \to \gamma.
  \begin{aligned}[t]
    \Rec{}
    \{&
    \Zero \to \Zero
    \\
    \mid&
    \Succ n \to n
    \}
    \With \gamma \}
  \end{aligned}
  \\
  \With{}& n
\end{aligned}
\end{align*}
\end{spacing}
This definition of $\<countDown>$ can be understood as follows:
\begin{itemize}
\item If the head of the stream is requested (matching the co\-pattern
  $\Head\alpha$), then the current value $x$ of the seed is returned (to
  $\alpha$) as-is.
\item Otherwise, if the tail is requested (matching the co\-pattern
  $\Tail\beta$), then the current value of the seed is inspected: if it is 0,
  then 0 is the value of the updated seed; otherwise,
  the seed is the successor of some $y$, in which case $y$ is
  given as the updated seed.
\end{itemize}

%
%
The previous definition of $\<countDown>$ is not very efficient; once $n$ reaches zero one can
safely returns the $zeroes$ stream thus avoiding the test for each question.
This can be avoided with the power of the corecursor as so:
\begin{spacing}
\begin{align*}
\begin{aligned}[t]
  \<countDown'> \ n 
  =
   \CoRec{}
  \{&
  \Head \alpha \to \alpha
  \\
  \mid&
  \Tail \beta \to \gamma.
  \begin{aligned}[t]
    \Rec{}
    \{&
    \Zero \to \outp \_ \cut{\<zeroes>}{\beta}
    \\
    \mid&
    \Succ n \to \gamma. n
    \} 
    \With \gamma \}
  \end{aligned}
  \\
  \With{}& n
\end{aligned}
\end{align*}
\end{spacing}
Note that in the co\-inductive step, the decision on which continuation is taken
(the observer of the tail $\beta$ or the continuation of co\-recursion $\gamma$)
depends on the current value of the seed. If the seed reaches 0
the corecursion is stopped and the streams of zero's is returned instead.

\subsection{Properly Dual (Co)Recursive Types}
\label{sec:duality}

Although we derived co\-recursion from recursion using duality, our prototypical
examples of natural numbers and streams were not perfectly dual to one another.
While the (co)\-recursive case of $\Tail E$ looks similar enough to $\Succ V$,
the base cases of $\Head E$ and $\Zero$ don't exactly line up, because $\Head$
takes a parameter but $\Zero$ does not.

One way to perfect the duality is to generalize $\Nat$ to $\Numbered A$ which represents
a value of type $A$ labeled with a natural number.  
$\Numbered A$ has two constructors: the base case
$\Zero : A \to \Numbered A$ labels an $A$ value with the number 0, and
$\Succ : \Numbered A \to \Numbered A$ increments the numeric label while leaving
the $A$ value alone, as shown by the  following rules:
\begin{gather*}
  \infer
  {\Gamma \entails \Zero V \givestype \Numbered A}
  {\Gamma \entails V \givestype A}
  \qquad
  \infer
  {\Gamma \entails \Succ V \givestype \Numbered A}
  {\Gamma \entails V \givestype \Numbered A}
\end{gather*}
In order to match the generalized $\Zero$ constructor, the base case of the
recursor needs to be likewise generalized with an extra parameter.  In System T,
this looks like
\begin{math}
  \Rec M \As{} \{\Zero x \to N \mid \Succ y \to z.N'\}
  ,
\end{math}
while in the abstract machine we get the generalized continuation
\begin{math}
  \Rec{} \{\Zero x \to v \mid \Succ y \to z.w\} \With E
  .
\end{math}
The typing rule for this continuation is:
\begin{gather*}
  \infer
  {
    \Gamma
    \entails
    \Rec{} \{\Zero x \to v \mid \Succ y \to z.w\} \With E
    \takestype
    \Numbered A
  }
  {
    \Gamma, x \givestype A \entails v \givestype B
    &
    \Gamma, y \givestype \Numbered A, z : B \entails w \givestype B
    &
    \Gamma \entails E \takestype B
  }
\end{gather*}

It turns out that $\Numbered A$ \emph{is} the proper dual to $\Stream A$.
Notice how the two constructors $\Zero V$ and $\Succ W$ exactly mirror the two
destructors $\Head E$ and $\Tail F$ when we swap the roles of values and
co\-values.  The recursor continuation
\begin{math}
  \Rec{} \{\Zero x \to v \mid \Succ y \to z.w\} \With E
\end{math}
is the perfect mirror image of the stream co\-recursor
\begin{math}
  \CoRec{} \{\Head\alpha \to e \mid \Tail\beta \to \gamma.f\} \With F
\end{math}
when we likewise swap variables with co\-variables in the (co\-)patterns.  In
more detail, we can formalize the duality relation (which we write as $\dualto$)
between values and co\-values.  Assuming that $V \dualto E$, we have the
following duality between the constructors and destructors of these two types:
\begin{align*}
  \Zero V &\dualto \Head E
  &
  \Succ V &\dualto \Tail E
\end{align*}
For (co)\-recursion, we have the following dualities, assuming $v \dualto e$
(under $x \dualto \alpha$) and $w \dualto f$ (under $y \dualto \beta$ and
$z \dualto \gamma$):
\begin{align*}
  \CoRec \{\Head \alpha \to e \mid \Tail \beta \to \gamma. f\}
  &\dualto
  \Rec \{\Zero x \to v \mid \Succ y \to z.w\}
\end{align*}

We could also express the proper duality by restricting streams instead of
generalizing numbers.  In terms of the type defined above, $\Nat$ is isomorphic
to $\Numbered \top$, where $\top$ represents the usual unit type with a single
value (often written as $()$).  Since $\top$ corresponds to logical truth, its
dual is the $\bot$ type corresponding to logical falsehood with no (closed)
values, and a single co\-value that represents an empty continuation.  With this
in mind, the type $\Nat$ is properly dual to $\Stream \bot$, \ie an infinite
stream of computations which cannot return any value to their observer\footnote{
From the point of view of polarity in programming languages
\cite{ZeilbergerPhD,MunchMaccagnoniPhD}, the $\top$ type for truth we use in
``$\Numbered \top$'' should be interpreted as a positive type (written as $1$ in
linear logic \cite{LinearLogic}).  Dually, the $\bot$ type for falsehood in
``$\Stream \bot$'' is a negative type (also called $\bot$ in linear logic).}.


%% file: sec_corec-vs-coiter.tex
\section{Corecursion vs Coiteration: Expressiveness and Efficiency}
\label{sec:corec-vs-coiter}

Similar to recursion, we can define two special cases of corecursion which only
use part of its functionality by just ignoring a parameter in the co\-recursive
branch. We derive the encodings for the creation of streams by applying the
syntactic duality to the encodings presented in \cref{sec:rec-vs-iter-mach}:
\begin{spacing}
\begin{align*}
\begin{aligned}
  &
  \CoCase
  \\
  &\quad
  \begin{alignedat}[t]{2}
    \{&
    \Head \alpha &&\to e
    \\
    \mid&
    \Tail \beta &&\to f
    \}
  \end{alignedat}
  \\
  &\With V
\end{aligned}
&\defeq\!\!
\begin{aligned}
  &\CoRec
  \\
  &\quad
  \begin{alignedat}[t]{2}
    \{&
    \Head \alpha &&\to e
    \\
    \mid&
    \Tail \beta &&\to \blank\,. f
    \}
  \end{alignedat}
  \\
  &\With V
\end{aligned}
&\quad
\begin{aligned}
  &
  \CoIter
  \\
  &\quad
  \begin{alignedat}[t]{2}
    \{&
    \Head \alpha &&\to e
    \\
    \mid&
    \Tail &&\to \gamma.f
    \}
  \end{alignedat}
  \\
  &\With V
\end{aligned}
&\defeq\!\!
\begin{aligned}
  &\CoRec
  \\
  &\quad
  \begin{alignedat}[t]{2}
    \{&
    \Head \alpha &&\to e
    \\
    \mid&
    \Tail \blank &&\to \gamma.f
    \}
  \end{alignedat}
  \\
  &\With V
\end{aligned}
\end{align*}
\end{spacing}
Specifically, $\CoCase$ simply matches on the shape of its projection, which has
the form $\Head\alpha$ or $\Tail\beta$, without co\-recursing at all.  In
contrast, $\CoIter$ \emph{always} co\-recurses by providing an updated
accumulator in the $\Tail\beta$ case without ever referring to $\beta$.  
We have seen already several examples of coiteration. Indeed, all examples
given in the previous section, except  $\<countDown'>$, are examples of
coiteration, as indicated by the absence of the rest of the continuation (named
$\_$).

Recall from \cref{sec:rec-fun} that recursion in call-by-name versus call-by-value have
different algorithmic complexities.  The same holds for co\-recursion, with the
benefit going instead to call-by-value.  Indeed, in call-by-value the simple
cocase continuation avoids the corecursive step entirely in two steps:
\begin{align*}
  &
  \cut
  {\CoCase\{\Head\alpha \to e \mid \Tail\beta \to f\} \With V}
  {\Tail E} \\
& \defeq
\cut
  {\CoRec
    \{ \Head \alpha \to e
    \mid
    \Tail \beta \to \blank\,. f
    \}
  \With V
  }
  {\Tail E}
  \\
  &\srd
  \cut
  {\outp\blank \cut{V}{f\subst{\beta}{E}}}
  {
    \inp x
    \cut
    {\CoCase\{\Head\alpha \to e \mid \Tail\beta \to f\} \With x}
    {E}
  }
  &(\beta_{\Tail})
  \\
  &\srd
  \cut{V}{f\subst{\beta}{E}}
  &(\mu)
\end{align*}
whereas, cocase continues corecursing in call-by-name, despite the fact that
this work will ultimately be thrown away once any element is requested via
$\Head$:
\begin{align*}
  &
  \cut
  {\CoCase \{\Head\alpha \to e \mid \Tail\beta \to f\} \With V}
  {\Tail E}
  \\
& \defeq
\cut
  {\CoRec
    \{ \Head \alpha \to e
    \mid
    \Tail \beta \to \blank\,. f
    \}
  \With V
  }
  {\Tail E}
  \\
  &\srd
  \cut
  {\outp\blank \cut{V}{f\subst{\beta}{E}}}
  {
    \inp x
    \cut
    {\CoCase\{\Head\alpha \to e \mid \Tail\beta \to f\} \With x}
    {E}
  }
  &(\beta_{\Tail})
  \\
  &\srd  
  \cut
  {
    \CoCase
    \{\Head\alpha \to e \mid \Tail\beta \to f\}
    \With
    \outp\blank \cut{V}{f\subst{\beta}{E}}
  }
  {E}
  &(\tmu)
\end{align*}

As an example of  cocase, consider 
the following  stream:
\begin{align*}
  \<scons>~x~(y_1,y_2,\dots) &= x, y_1, y_2, \dots
\end{align*}
$\<scons>~x~s$ appends a new element $x$ on top of the stream $s$.
This informal definition can be formalized in terms of $\CoCase$ like so:
\begin{align*}
  \<scons>
  &\defeq
  \lambda x. \lambda s.
  \CoCase
  \{
  \Head \alpha \to \alpha
  \mid
  \Tail \beta \to \inp \_ \cut{s}{\beta}
  \}
  \With x
  \end{align*}

%
Ideally, $\<scons>$ should not leave a lingering effect on the underlying
stream.  That is to say, the tail of $\<scons>~x~s$ should just be $s$.
This happens directly in call-by-value.
Consider indexing the $n+1^{th}$ element of $\<scons>$ in call-by-value, where
we write $\Tail^n E$ to mean the $n$-fold application of $\Tail$ over $E$ (\ie
$\Tail^0 E = E$ and $\Tail^{n+1} E = \Tail(\Tail^n E)$):
\begin{align*}
  &
  \cut{\<scons>}{\app{x}{\app{s}}{\Tail^{n+1}(\Head\alpha)}}
  \\
  &\srds
  \cut
  {
    \CoCase
    \{\Head\alpha \to \alpha \mid \Tail\beta \to \inp\blank\cut{s}{\beta}\}
    \With x
    \}
  }
  {\Tail(\Tail^{n}(\Head\alpha))}
  &(\beta_\to)
  \\
  &\srds
  \cut
  {x}
  {\inp\blank\cut{s}{\Tail^n(\Head\alpha)}}
  &(\beta_{\Tail}\mu)
  \\
  &\srd
  \cut{s}{\Tail^{n}(\Head\alpha)}
  &(\tmu)
\end{align*}
Notice how, after the first $\Tail$ is resolved, the computation incurred by
$\<scons>$ has completely vanished.  In contrast, the computation of $\<scons>$
continues to linger in call-by-name:
\begin{align*}
  &
  \cut{\<scons>}{\app{x}{\app{s}}{\Tail^{n+1}(\Head\alpha)}}
  \\
  &\srds
  \cut
  {
    \CoCase
    \{\Head\alpha \to \alpha \mid \Tail\beta \to \inp\blank\cut{s}{\beta}\}
    \With x
    \}
  }
  {\Tail(\Tail^{n}(\Head\alpha))}
  &(\beta_\to)
  \\
  &\srds
  \cut
  {
    \CoCase
    \{\dots\}
    \With \outp\blank \cut{x}{\inp\blank\cut{s}{\Tail^n(\Head\alpha)}}
  }
  {\Tail^n(\Head\alpha)}
  &(\beta_{\Tail}\tmu)
\end{align*}
Here, we will spend time over the next $n$ $\Tail$ projections to build up an
ever larger accumulator until the $\Head$ is reached, even though the result
will inevitably just backtrack to directly ask $s$ for $\Tail^n E$.

However, one question remains: how do we accurately measure the cost of a
stream?  The answer is more subtle than the cost of numeric loops, because
streams more closely resemble functions.  With functions, it is not good enough to
just count the steps it takes to calculate the closure.  We also need to count
the steps taken in the body of the function when it is called, \ie what happens
when the function is used.  The same issue occurs with streams, where we also
need to count what happens when the stream is used, \ie the number of steps
taken inside the stream in response to a projection.  Of course, the internal
number of steps in both cases can depend on the ``size'' of the input.  For
functions, this is the size of its argument; in the simple case of
$\Nat \to \Nat$, this size is just the value $n$ of the numeric argument
$\Succ^n\Zero$ itself.  For streams, the input is the stream projection of the
form $\Tail^n(\Head\alpha)$, whose size is proportional to the number $n$ of
$\Tail$ projections.  Therefore, the computational complexity of a stream value
$s$---perhaps defined by a $\CoRec$ term---is expressed as some $O(f(n))$, where
$n$ denotes the depth of the projection given by the number of $\Tail$s.

Now, let us consider the asymptotic cost of $\<scons>$ under both call-by-value
and call-by-name evaluation.  In general, the difference in performance between
$\cut{\<scons>}{\app{x}{\app{s}{E}}}$ and $\cut{s}{E}$ in call-by-value is just
a constant time ($O(1)$) overhead.  So, given that the cost of $s$ is $O(f(n))$,
then the cost of $\cut{\<scons>}{\app{x}{\app{s}{E}}}$ will also be $O(f(n))$.
In contrast, the call-by-name evaluation of $\<scons>$ incurs an additional
linear time ($O(n)$) overhead based on depth of the projection:
$\cut{\<scons>}{\app{x}{\app{s}{\Tail^{n+1}(\Head E)}}}$ takes an additional
number of steps proportional to $n$ compared to the cost of executing
$\cut{s}{\Tail^n(\Head E)}$.  As a consequence, the call-by-name cost of
$\cut{\<scons>}{\app{x}{\app{s}{E}}}$ is $O(n+f(n))$, given that the cost of $s$
is $O(f(n))$.  So the efficiency of co\-recursion is better in call-by-value
than in call-by-name.

To make the analysis more concrete, lets look at an application of $\<scons>$ to
a specific underlying stream.  Recall the $\<countDown>$ function from
\cref{sec:corec-examples} which produces a stream counting down from a given
$n$: $n, n-1, \dots, 2, 1, 0, 0, \dots$.  This function (and the similar
$\<countDown'>$) is defined to immediately return a stream value that captures
the starting number $n$, and then incrementally counts down one at a time with
each $\Tail$ projection.  An alternative method of generating this same stream
is to do all the counting up front: recurse right away on the starting number
and use $\<scons>$ to piece together the stream from each step towards zero.
This algorithm can be expressed in terms of System T-style recursion like so:
\begin{spacing}
\begin{align*}
  \<countNow> \ n 
  =
  \Rec n \As{}
  \{&
  \Zero \to \<zeroes>
  \\
  \mid&
  \Succ x \to xs.~ \<scons>~(\Succ x)~xs
  \}
\end{align*}
\end{spacing}
So that the translation of $\<countNow>$ to the abstract machine language is:
\begin{spacing}
\begin{align*}
  \<countNow> 
  =
  \lambda n.
  \outp\alpha
  \bigcut
  {n}
  {
    \begin{aligned}
      \Rec{}
      \{&
      \Zero \to \<zeroes>
      \\
      \mid&
      \Succ x \to xs.\outp\beta\cut{\<scons>}{\app{\Succ x}{\app{xs}{\beta}}}
      \}
      \\
      \With{}& \alpha
    \end{aligned}
  }
\end{align*}
\end{spacing}
What is the cost of $\<countNow>$? First, the up-front cost of calling the
function with the argument $\Succ^n\Zero$ is unsurprisingly $O(n)$, due to the
use of the recursor.  But this doesn't capture the full cost; what is the
efficiency of using the stream value returned by $\<countNow>$?  To understand
the efficiency of the stream, we have to consider \emph{both} the initial
numeric argument as well as the depth of the projection:
$\cut{\<countNow>}{\app{\Succ^n\Zero}{\Tail^m(\Head\alpha)}}$ in the abstract
machine language, which corresponds to
$\Head(\Tail^m(\<countNow>~(\Succ^n\Zero)))$ in a functional style.  In the base
case, we have the stream $\<zeroes>$, whose cost is $O(m)$ because it must
traverse past all $m$ $\Tail$ projections before it can return a $0$ to the
$\Head$ projection.  On top of this, we apply $n$ applications of $\<scons>$.
Recall that in call-by-value, we said that the cost of $\<scons>$ is the same as
the underlying stream.  Thus, the call-by-value efficiency of the stream
returned by $\cut{\<countNow>}{\app{\Succ^n\Zero}{\Tail^m(\Head\alpha)}}$ is
just $O(m)$.  In call-by-name in contrast, $\<scons>$ adds an additional linear
time overhead to the underlying stream.  Since there are $n$ applications of
$\<scons>$ and the projection is $m$ elements deep, the call-by-name efficiency
of the stream returned by
$\cut{\<countNow>}{\app{\Succ^n\Zero}{\Tail^m(\Head\alpha)}}$ is $O(m(n+1))$.
If the count $n$ and depth $m$ are roughly proportional to each other (so that
$n \approx m$), the difference between call-by-value and call-by-name evaluation
of $\<countNow>$'s stream is a jump between a linear time and quadratic time
computation.

\subsection{Corecursion in Terms of Coiteration}

Recall in \cref{sec:rec-vs-iter-mach} that we were able to encode $\Rec$ in
terms of $\Iter$, though at an increased computational cost.  Applying syntactic
duality, we can write a similar encoding of $\CoRec$ as a macro-expansion in
terms of $\CoIter$.  Doing so requires the dual of pairs: sum types.  Sum types
in the abstract machine look like \cite{CBVDualToCBN}:
\begin{align*}
  \cut{\Left V}{[e_1,e_2]}
  &\srd
  \cut{V}{e_1}
  &
  \cut{\Right V}{[e_1,e_2]}
  &\srd
  \cut{V}{e_2}
  &(\beta_+)
\end{align*}
The macro-expansion for encoding $\CoRec$ as $\CoIter$ derived from duality is:
\begin{spacing}
\begin{align*}
\begin{aligned}
  &
  \CoRec
  \begin{alignedat}[t]{2}
    \{&
    \Head \alpha &&\to e
    \\
    \mid&
    \Tail \beta &&\to \gamma.f
    \}
  \end{alignedat}
  \\
  &
  \With V
\end{aligned}
&\defeq
\begin{aligned}
  &\CoIter
  \begin{alignedat}[t]{2}
    \{&
    \Head \alpha &&\to [\Head\alpha, e]
    \\
    \mid&
    \Tail &&\to [\beta, \gamma].[\Tail\beta, f]
    \}
  \end{alignedat}
  \\
  &
  \With{} \Right V
\end{aligned}
\end{align*}
\end{spacing}
Note how the co\-inductive step of $\CoIter$ has access to both an option to
update the internal seed, or to return some other, fully-formed stream as its
tail.  So dually to the encoding of the recursor, this encoding keeps both
options open by reconstructing the original continuation alongside the
continuation which use its internal seed.  In the base case matching
$\Head\alpha$, it rebuilds the observation $\Head\alpha$ and pairs it up with
the original response $e$ to the base case.  In the co\-inductive case matching
$\Tail\beta$, the projection $\Tail\beta$ is combined with the continuation $f$
which can update the internal seed via $\gamma$.  Since $f$ might refer to one
or both of $\beta$ and $\gamma$, we need to ``extract'' the two parts from the
co\-recursive tail observation.  Above, we express this extraction as
co\-pattern matching \cite{Copatterns}, which is shorthand for
\begin{align*}
  [\beta,\gamma]. [e,f]
  \defeq
  \alpha.
  \inp x
  \cut
  {\Left\outp\beta\cut{\Right\outp\gamma\cut{x}{[e,f]}}{\alpha}}
  {\alpha}
\end{align*}
Because this encoding is building up a pair of two continuations, the internal
seed which is passed to them needs to be a value of the appropriately matching
sum type.  Thus, the encoding has two modes:
\begin{itemize}
\item If the seed has the form $\Right x$, then the $\CoIter$ is simulating the
  original $\CoRec$ process with the seed $x$.  This is because the $\Right$
  continuations for the base and co\-inductive steps are exactly those of the
  encoded recursor ($e$ and $f$, respectively) which get applied to $x$.
\item If the seed has the form $\Left s$, containing some stream $s$, then the
  $\CoIter$ is mimicking $s$ as-is.  This is because the $\Left$ continuations
  for the base and co\-inductive steps exactly mimic the original observations
  ($\Head\alpha$ and $\Tail\beta$, respectively) which get applied to the stream
  $s$.
\end{itemize}
To start the co\-recursive process off, we begin with $\Right V$, which
corresponds to the original seed $V$ given to the co\-recursor.  Then in the
co\-inductive step, if $f$ updates the seed to $V'$ via $\gamma$, the seed is
really updated to $\Right V'$, continuing the simulation of co\-recursion.
Otherwise, if $f$ returns some other stream $s$ to $\beta$, the seed gets
updated to $\Left s$, and the co\-iteration continues but instead mimics $s$ at
each step.


As with recursion, encoding $\CoRec$ in terms of $\CoIter$ forces a performance
penalty for functions like $\<scons>$ which can return a stream directly in the
co\-inductive case instead of updating the accumulator.  Recall how
call-by-value execution was more efficient than call-by-name.  Yet, the above
encoding results in this same inefficient execution even in call-by-value, as in
this command which accesses the $(n+2)^{th}$ element:
\begin{spacing}
\begin{align*}
  &
  \cut{scons}{\app{x}{\app{s}{\Tail^{n+2}(\Head\alpha)}}}
  \\
  &\dmapsto
  \bigcut
  {
    \begin{aligned}
      \CoRec{}
      \{&
      \Head \alpha \to \alpha
      \\
      \mid&
      \Tail \beta \to \blank\,. \inp{\blank\,} \cut{s}{\beta}
      \}
      \\
      \With{}& {\color{blue} x}
    \end{aligned}
  }
  {\Tail^{n+2}(\Head\alpha)}
  &(\beta_\to)
  \\
  &\defeq
  \bigcut
  {
    \begin{aligned}
      \CoIter{}
      \{&
      \Head\alpha \to [\Head\alpha, \alpha]
      \\
      \mid&
      \Tail \to [\beta, \blank]. [\Tail\beta, \inp{\blank}\cut{s}{\beta}]
      \}
      \\
      \With{}& {\color{blue} \Right x}
    \end{aligned}
  }
  {\Tail(\Tail^{n+1}(\Head\alpha))}
  \\
  &\dmapsto
  \cut
  {\color{blue} \Right x}
  {
    [
    \Tail\inp y \cut{\Left y}{\gamma}
    ,
    \inp\blank\cut{s}{\inp y \cut{\Left y}{\gamma}}
    ]
  }
  &(\beta_{\Tail},\dots)
  \\
  &\qquad
  \Where
  \gamma = \inp z \cut{\CoIter\{\dots\}\With z}{\Tail^{n+1}(\Head\alpha)}
  \\
  &\mapsto
  \cut
  {x}
  {
    \inp\blank
    \cut
    {\color{blue} s}
    {
      \inp y
      \cut
      {\color{blue} \Left y}
      {\inp z \cut{\CoIter\{\dots\}\With z}{\Tail^{n+1}(\Head\alpha)}}
    }
  }
  &(\beta_+)
  \\
  &\dmapsto
  \cut{\CoIter\{\dots\}\With {\color{blue} \Left s}}{\Tail^{n+1}(\Head\alpha)}
  &(\tmu)
  \\
  &\dmapsto
  \cut
  {\color{blue} \Left s}
  {
    [
    \Tail \inp y \cut{\Left y}{\gamma}
    ,
    \inp\blank\cut{s}{\inp y \cut{\Left y}{\gamma}}
    ]
  }
  &(\beta_{\Tail},\dots)
  \\
  &\qquad
  \Where
  \gamma = \inp z \cut{\CoIter\{\dots\}\With z}{\Tail^{n}(\Head\alpha)}
  \\
  &\mapsto
  \cut
  {\color{blue} s}
  {
    \Tail \inp y
    \cut
    {\Left y}
    {\inp z \cut{\CoIter\{\dots\}\With z}{\Tail^{n}(\Head\alpha)}}
  }
  &(\beta_+)
\end{align*}
\end{spacing}
Notice how the internal seed (in blue) changes through this computation.  To
begin, the seed is ${\color{blue}\Right x}$.  The first $\Tail$ projection
(triggering the $\beta_{\Tail}$ rule) leads to the decision point (by $\beta_+$)
which chooses to update the seed with ${\color{blue}\Left s}$.  From that point
on, each $\Tail$ projection to follow will trigger the next step of this
coiteration (and another $\beta_{\Tail}$ rule). Each time, this will end up
asking ${\color{blue}s}$ for its $\Tail$, ${\color{blue}s_1}$, which will be
then used to build the next seed, ${\color{blue}\Left s_1}$.

In order to continue, we need to know something about $s$, specifically, how it
responds to a $\Tail$ projection.  For simplicity, assume that the tail of $s$
is $s_1$, \ie $\cut{s}{\Tail E} \srds \cut{s_1}{E}$.  And then for each
following $s_i$, assume its tail is $s_{i+1}$.  Under this assumption, execution
will proceed to the base case $\Head$ projection like so:
\begin{align*}
  &
  \cut
  {\color{blue} s}
  {
    \Tail \inp y
    \cut
    {\Left y}
    {\inp z \cut{\CoIter\{\dots\}\With z}{\Tail^{n}(\Head\alpha)}}
  }
  \\
  &\dmapsto
  \cut{\CoIter\{\dots\}\With {\color{blue} \Left s_1}}{\Tail^{n}(\Head\alpha)}
  &(\Tail s \dmapsto s_1)
  \\
  &\dmapsto
  \cut{\CoIter\{\dots\}\With {\color{blue} \Left s_2}}{\Tail^{n-1}(\Head\alpha)}
  &(\Tail s_1 \dmapsto s_2)
  \\
  &\dmapsto
  \dots
  &(\Tail s_i \dmapsto s_{i+1})
  \\
  &\dmapsto
  \cut{\CoIter\{\dots\}\With {\color{blue} \Left s_{n+1}}}{\Head\alpha}
  &(\Tail s_n \dmapsto s_{n+1})
  \\
  &\mapsto
  \cut
  {{\color{blue} \Left s_{n+1}}}
  {[\Head\alpha, \alpha]}
  &(\beta_{\Head})
  \\
  &\mapsto
  \cut{s}{\Head\alpha}
  &(\beta_+)
\end{align*}
In total, this computation incurs additional $\beta_{\Tail}$ steps linearly
proportional to $n+2$, on top of any additional work needed to compute the same
number of $\Tail$ projections for each
$\cut{s_i}{\Tail E} \srds \cut{s_{i+1}}{E}$ along the way.


%% file: sec_corec-mach-correct.tex
\section{Safety and Termination}
\label{sec:safety&termination}

Just like System T \cref{thm:t-type-safety,thm:t-termination}, the
(co)\-recursive abstract machine is both \emph{type safe} (meaning that
well-typed commands never get stuck) and \emph{terminating} (meaning that
well-typed commands always cannot execute forever).  In order to prove this
fact, we can build a model of type safety and termination rooted in the idea of
Girard's reducibility candidates \cite{ProofsAndTypes}, but which matches closer
to the structure of the abstract machine.  In particular, we will use a model
from \cite{ClassicalStrongNormalization,DualityOfIntersectonUnionTypes} which
first identifies a set of commands that are \emph{safe} to execute (in our case,
commands which terminate on a valid final configuration).  From there, types are
modeled as a combination of terms and co\-terms that embed this safety property
of executable commands.

\subsection{Safety and Candidates}

To begin, we derive our notion of safety from the conclusion of
\cref{thm:t-termination}.  Ultimately we will only run commands that are closed,
save for one free co\-variable $\alpha$ taking a $\Nat$, so we expect all such
\emph{safe} commands to eventually finish execution in a valid final state that
provides that $\alpha$ with a $\Nat$ construction: either a $\Zero$ or
$\Succ$essor.
\begin{definition}[Safety]
\label{def:termination-pole}
The set of \emph{safe} commands, $\Bot$, is:
\begin{align*}
  \Bot
  &\defeq
  \{
  c
  \mid
  \exists c' \in \mathit{Final}.~
  c \srds c'
  \}
  \\
  \mathit{Final}
  &\defeq
  \{\cut{\Zero}{\alpha} \mid \alpha \in \mathit{CoVar}\}
  \cup
  \{\cut{\Succ V}{\alpha} \mid V \in \mathit{Value}, \alpha \in \mathit{CoVar}\}
\end{align*}
\end{definition}

From here, we will model types as collections of terms and co\-terms that work
well together.  A sensible place to start is to demand \emph{soundness}: all of
these terms and co\-terms can only form safe commands (that is, ones found in
$\Bot$).  However, we will quickly find that we also need \emph{completeness}:
any terms and co\-terms that do not break safety are included.
\begin{definition}[Candidates]
\label{def:pre-candidate}
\label{def:sound-candidate}
\label{def:complete-candidate}
\label{def:reducibility-candidate}

A \emph{pre-candidate} is any pair $\sem{A}=(\sem{A}^+,\sem{A}^-)$ where
$\sem{A}^+$ is a set of terms, and $\sem{A}^-$ is a set of co\-terms, \ie
$\sem{A} \in \wp(\mathit{Term}) \times \wp(\mathit{CoTerm})$.

\noindent
A \emph{sound} (pre-)candidate satisfies this additional
requirement:
\begin{itemize}
\item \emph{Soundness:} for all $v \in \sem{A}^+$ and $e \in \sem{A}^-$, the
  command $\cut{v}{e}$ is safe (\ie $\cut{v}{e} \in \Bot$).
\end{itemize}
A \emph{complete} (pre-)candidate satisfies these two \emph{completeness}
requirements:
\begin{itemize}
\item \emph{Positive completeness:} if $\cut{v}{E}$ is safe (\ie
  $\cut{v}{E} \in \Bot$) for all $E \in \sem{A}^-$ then $v \in \sem{A}^+$.
\item \emph{Negative completeness:} if $\cut{V}{e}$ is safe (\ie
  $\cut{V}{e} \in \Bot$) for all $V \in \sem{A}^+$ then $e \in \sem{A}^-$.
\end{itemize}
A \emph{reducibility candidate} is any sound and complete (pre-)candidate.
$\precands$ denotes the set of all pre-candidates, $\soundcands$ denotes the set
of sound ones, $\complcands$ the set of complete ones, and $\redcands$ denotes
the set of all reducibility candidates.

As notation, given any pre-candidate $\sem{A}$, we will always write $\sem{A}^+$
to denote the first component of $\sem{A}$ (the term set $\pi_1(\sem{A})$) and
$\sem{A}^-$ to denote the second one (the co\-term set $\pi_2(\sem{A})$).  As
shorthand, given a reducibility candidate $\sem{A}=(\sem{A}^+,\sem{A}^-)$, we
write $v \in \sem{A}$ to mean $v \in \sem{A}^+$ and likewise $e \in \sem{A}$ to
mean $e \in \sem{A}^-$.  Given a set of terms $\sem{A}^+$, we will occasionally
write the pre-candidate $(\sem{A}^+,\{\})$ as just $\sem{A}^+$ when the
difference is clear from context.  Likewise, we will occasionally write the
pre-candidate $(\{\},\sem{A}^-)$ as just the co\-term set $\sem{A}^-$ when
unambiguous.
\end{definition}

The motivation behind soundness may seem straightforward.  It ensures that the
$Cut$ rule is safe.  But soundness is not enough, because the type system does
much more than $Cut$: it makes many promises that several terms and co\-terms
inhabit the different types.  For example, the function type contains
$\lambda$-abstractions and call stacks, and \emph{every} type contains $\mu$-
and $\tmu$-abstractions over free (co)\-variables of the type.  Yet, there is
nothing in soundness that keeps these promises.  For example, the trivial
pre-candidate $(\{\},\{\})$ is sound but it contains nothing, even though $ActR$
and $ActL$ promise many $\mu$- and $\tmu$-abstractions that are left out!  So to
fully reflect the rules of the type system, we require a more informative model.

Completeness ensures that every reducibility candidate has ``enough''
(co)\-terms that are promised by the type system.  For example, the completeness
requirements given in \cref{def:reducibility-candidate} are enough to guarantee
every complete candidate contains all the appropriate $\mu$- and
$\tmu$-abstractions that always step to safe commands for any allowed binding.
\begin{lemma}[Activation]
\label{thm:activation-sound}
For any complete candidate $\sem{A}$:
\begin{enumerate}
\item If $c\subst{\alpha}{E} \in \Bot$ for all $E \in \sem{A}$, then
  $\outp\alpha c \in \sem{A}$.
\item If $c\subst{x}{V} \in \Bot$ for all $V \in \sem{A}$, then
  $\inp x c \in \sem{A}$.
\end{enumerate}
\end{lemma}
\begin{proof}
  Consider the first fact (the second is perfectly dual to it and follows
  analogously), and assume that $c\subst{\alpha}{E} \in \Bot$ for all
  $E \in \sem{A}$.  In other words, the definition of
  $c\subst{\alpha}{E} \in \Bot$ says $c\subst{\alpha}{E} \srds c'$ for some
  valid command $c' \in \mathit{Final}$.  For any specific
  $E \in \sem{A}$, we have:
  \begin{align*}
    \cut{\outp{\alpha} c}{E}
    &\srd
    c\subst{\alpha}{E}
    \srds
    c'
    \in
    \mathit{Final}
  \end{align*}
  By definition of $\Bot$ (\cref{def:termination-pole}) and transitivity of
  reduction, $\cut{\outp\alpha c}{E} \in \Bot$ as well for any $E \in \sem{A}$.
  Thus, $\sem{A}$ \emph{must} contain $\outp{\alpha}c$, as required by positive
  completeness (\cref{def:reducibility-candidate}).
\end{proof}
Notice that the restriction to values and co\-values in the definition of
completeness (\cref{def:reducibility-candidate}) is crucial in proving
\cref{thm:activation-sound}.  We can easily show that the $\mu$-abstraction
steps to a safe command for every given (co)\-value, but if we needed to say the
same for every (co)\-term we would be stuck in call-by-name evaluation where the
$\mu$-rule might not fire.  Dually, it is always easy to show the $\tmu$ safely
steps for every value, but we cannot say the same for every term in
call-by-value.  The pattern in \cref{thm:activation-sound} of reasoning about
commands based on the ways they reduce is a crucial key to proving properties of
particular (co)\-terms of interest.  The very definition of safe commands $\Bot$
is \emph{closed under expansion} of the machine's operational semantics.
\begin{property}[Expansion]
\label{thm:safety-expansion}
If $c \srd c' \in \Bot$ then $c \in \Bot$.
\end{property}
\begin{proof}
  Follows from the definition of $\Bot$ (\cref{def:termination-pole}) and
  transitivity of reduction.
\end{proof}
This property is the key step used to conclude the proof
\cref{thm:activation-sound}, and can be used to argue that other (co)\-terms are
in specific reducibility candidates due to the way they reduce.

\subsection{Subtyping and Completion}

Before we delve into the interpretations of types as reducibility candidates, we
first need to introduce another important concept that the model revolves
around: \emph{subtyping}.  The type system
(\cref{fig:t-mach-type-system,fig:stream-mach-type-system}) for the abstract
machine has no rules for subtyping, but nevertheless, a semantic notion of
subtyping is useful for organizing and building reducibility candidates.  More
specifically, notice that there are exactly two basic ways to order
pre-candidates based on the inclusion of their underlying sets:
\begin{definition}[Refinement and Subtype Order] 
\label{def:subtype-order}
\label{def:refinement-order}  

The \emph{refinement order} ($\sem{A} \refines \sem{B}$) and \emph{subtype
  order} ($\sem{A} \leq \sem{B}$) between pre-candidates is:
\begin{align*}
  (\sem{A}^+,\sem{A}^-) \refines (\sem{B}^+,\sem{B}^-)
  &\defeq
  (\sem{A}^+ \subseteq \sem{B}^+)
  \conj
  (\sem{A}^- \subseteq \sem{B}^-)
  \\
  (\sem{A}^+,\sem{A}^-) \leq (\sem{B}^+,\sem{B}^-)
  &\defeq
  (\sem{A}^+ \subseteq \sem{B}^+)
  \conj
  (\sem{A}^- \supseteq \sem{B}^-)
\end{align*}
The reverse \emph{extension order} $\sem{A} \extends \sem{B}$ is defined as
$\sem{B} \refines \sem{A}$, and \emph{supertype order} $\sem{A} \geq \sem{B}$
is $\sem{B} \leq \sem{A}$.
\end{definition}

Refinement just expresses basic inclusion: $\sem{A}$ refines $\sem{B}$ when
everything in $\sem{A}$ (both term and co\-term) is contained within $\sem{B}$.
With subtyping, the underlying set orderings go in \emph{opposite} directions!
If $\sem{A}$ is a subtype of $\sem{B}$, then $\sem{A}$ can have \emph{fewer}
terms and \emph{more} co\-terms than $\sem{B}$.  While this ordering may seem
counter-intuitive, it closely captures the understanding of sound candidates
where co\-terms are \emph{tests} on terms.  If $\Bot$ expresses which terms pass
which tests (\ie co\-terms), the soundness requires that all of its terms passes
each of its tests.  If a sound candidate has fewer terms, then it might be able
to safely include more tests which were failed by the removed terms.  But if a
sound candidate has more terms, it might be required to remove some tests that
the new terms don't pass.

This semantics for subtyping formalizes Liskov's \emph{substitution principle}
\cite{LiskovSubstitutionPrinciple}: if $\sem{A}$ is a subtype of $\sem{B}$, then
terms of $\sem{A}$ are also terms of $\sem{B}$ because they can be safely used
in any context expecting inputs from $\sem{B}$ (\ie with any co\-term of
$\sem{B}$).  Interestingly, our symmetric model lets us express the logical dual
of this substitution principle: if $\sem{A}$ is a subtype of $\sem{B}$, then the
co\-terms of $\sem{B}$ (\ie contexts expecting inputs of $\sem{B}$) are
co\-terms of $\sem{A}$ because they
can be safely given any input from $\sem{A}$.
These two principles lead to a natural subtype ordering of reducibility
candidates, based on the sets of (co)\-terms they include.

The usefulness of this semantic, dual notion of subtyping comes from the way it
gives us a complete lattice, which makes it possible to combine and build new
candidates from other simpler ones.
\begin{restatable}[Subtype Lattice]{definition}
  {thmpresubtypelattice}
\label{def:pre-subtype-lattice}

There is a complete lattice of pre-candidates in $\precands$ with respect to
subtyping order, where the binary intersection ($\sem{A} \wedge \sem{B}$, \aka
\emph{meet} or \emph{greatest lower bound}) and union ($\sem{A} \vee \sem{B}$,
\aka \emph{join} or \emph{least upper bound}) are defined as:
\begin{align*}
  \sem{A} \wedge \sem{B}
  &\defeq
  (\sem{A}^+ \cap \sem{B}^+, \sem{A}^- \cup \sem{B}^-)
  &
  \sem{A} \vee \sem{B}
  &\defeq
  (\sem{A}^+ \cup \sem{B}^+, \sem{A}^- \cap \sem{B}^-)
\end{align*}
Moreover, these generalize to the intersection ($\bigwedge \mathcal{A}$, \aka
infimum) and union ($\bigvee \mathcal{A}$, \aka supremum) of any set
$\mathcal{A} \subseteq \precands$ of pre-candidates
\begin{align*}
  \bigwedge \mathcal{A}
  &\defeq
  \left(
    \bigcap \{\sem{A}^+ \mid \sem{A} \in \mathcal{A}\}
    ,
    \bigcup \{\sem{A}^- \mid \sem{A} \in \mathcal{A}\}
  \right)
  \\
  \bigvee \mathcal{A}
  &\defeq
  \left(
    \bigcup \{\sem{A}^+ \mid \sem{A} \in \mathcal{A}\}
    ,
    \bigcap \{\sem{A}^- \mid \sem{A} \in \mathcal{A}\}
  \right)
\end{align*}
\end{restatable}
The binary least upper bounds and greatest lower bounds have these standard
properties:
\begin{align*}
  \sem{A}\wedge\sem{B}
  &\leq
  \sem{A},\sem{B}
  \leq
  \sem{A}\vee\sem{B}
  &
  \sem{C} \leq \sem{A},\sem{B}
  &\implies
  \sem{C} \leq \sem{A}\wedge\sem{B}
  &
  \sem{A},\sem{B} \leq \sem{C}
  &\implies
  \sem{A}\vee\sem{B} \leq \sem{C}
\end{align*}
In the general case for the bounds on an entire set of pre-candidates, we know:
\begin{align*}
  \forall \sem{A} \in \mathcal{A}.~
  &
  \textstyle
  \left(
    \bigwedge \mathcal{A}
    \leq
    \sem{A}
  \right)
  &
  (\forall \sem{A} \in \mathcal{A}.~ \sem{C} \leq \sem{A})
  &\implies
  \textstyle
  \sem{C} \leq \bigwedge \mathcal{A}
  \\
  \forall \sem{A} \in \mathcal{A}.~
  &
  \textstyle
  \left(
    \sem{A}
    \leq
    \bigvee \mathcal{A}
  \right)
  &
  (\forall \sem{A} \in \mathcal{A}.~ \sem{A} \leq \sem{C})
  &\implies
  \textstyle
  \bigvee \mathcal{A} \leq \sem{C}
\end{align*}

These intersections and unions both preserve soundness, and so they form a
lattice of sound candidates in $\soundcands$ as well.  However, they do not
preserve completeness in general, so they do not form a lattice of reducibility
candidates, \ie sound and complete pre-candidates.  Completeness is not
preserved because $\sem{A} \vee \sem{B}$ might be missing some terms (such as
$\mu$s) which could be soundly included, and dually $\sem{A} \wedge \sem{B}$
might be missing some co\-terms (such as $\tmu$s).  So because all reducibility
candidates are sound pre-candidates, $\sem{A} \vee \sem{B}$ and
$\sem{A} \wedge \sem{B}$ are well-defined, but their results will only be sound
candidates (not another reducibility candidate).%
\footnote{The refinement lattice, on the other hand, interacts very differently
  with soundness and completeness.  The refinement union
  \begin{math}
    \sem{A} \sqcup \sem{B}
    \defeq
    (\sem{A}^+ \cup \sem{B}^+, \sem{A}^- \cup \sem{B}^-)
    ,
  \end{math}
  preserves completeness, but might break soundness by putting together a term
  of $\sem{A}$ which is incompatible with a (co)\-term of $\sem{B}$, or vice
  versa.  Dually, The refinement intersection of two pre-candidates,
  \begin{math}
    \sem{A} \sqcap \sem{B}
    \defeq
    (\sem{A}^+ \cap \sem{B}^+, \sem{A}^- \cap \sem{B}^-)
    ,
  \end{math}
  preserves soundness but can break completeness if a safe term or co\-term is
  left out of the underlying intersections.  So while refinement may admit a
  complete lattice for pre-candidates in $\precands$, we only get two dual
  refinement semi-lattices for sound candidates in $\soundcands$ and complete
  candidates in $\complcands$.}

What we need is a way to extend arbitrary sound candidates, adding ``just
enough'' to make them full-fledged reducibility candidates.  Since there are two
possible ways to do this (add the missing terms or add the missing co\-terms),
there are two completions which go in different directions.  Also, since the
completeness of which (co)\-terms are guaranteed to be in reducibility
candidates is specified up to (co)\-values, we cannot be sure that absolutely
everything ends up in the completed candidate.  Instead, we can only ensure that
the (co)\-values that started in the pre-candidate are contained in its
completion.  This restriction to just the (co)\-values of a pre-candidate,
written $\sem{A}^v$ and defined as
\begin{align*}
  (\sem{A}^+,\sem{A}^-)^v
  &\defeq
  (\{V \mid V \in \sem{A}^+\}, \{E \mid E \in \sem{A}^-\})
\end{align*}
becomes a pivotal part of the semantics of types.  With this in mind, it follows
there are exactly two ``ideal'' completions, the positive and negative ones,
which give a reducibility candidate that is the closest possible to the starting
point.
\begin{restatable}[Positive \& Negative Completion]{lemma}
  {thmposnegcompletion}
\label{thm:pos-completion}
\label{thm:neg-completion}

There are two completions, $\PosCand$ and $\NegCand$, of any sound candidate
$\sem{A}$ with these three properties:
\begin{enumerate}
\item \emph{They are reducibility candidates:} $\PosCand(\sem{A})$ and
  $\NegCand(\sem{A})$ are both sound and complete.
\item \emph{They are (co)\-value extensions:} Every (co)\-value of $\sem{A}$ is
  included in $\PosCand(\sem{A})$ and $\NegCand(\sem{A})$.
  \begin{align*}
    \sem{A}^v &\refines \PosCand(\sem{A})
    &
    \sem{A}^v &\refines \NegCand(\sem{A})
  \end{align*}
\item \emph{They are the least/greatest such candidates:} Any reducibility
  candidate that extends the (co)\-values of $\sem{A}$ lies between
  $\PosCand(\sem{A})$ and $\NegCand(\sem{A})$, with $\PosCand(\sem{A})$ being
  smaller and $\NegCand(\sem{A})$ being greater.  In other words, given any
  reducibility candidate $\sem{C}$ such that $\sem{A}^v \refines \sem{C}$:
  \begin{align*}
    \PosCand(\sem{A}) \leq \sem{C} \leq \NegCand(\sem{A})
  \end{align*}
\end{enumerate}
\end{restatable}
The full proof of these completions (and proofs of the other remaining
propositions not given in this section) are given in \cref{sec:model}.  We refer
to $\PosCand(\sem{A})$ as the \emph{positive completion of $\sem{A}$} because it
is based entirely on the values of $\sem{A}$ (\ie its positive components):
$\PosCand(\sem{A})$ collects the complete set of co\-terms that are safe with
$\sem{A}$'s values, then collects the terms that are safe with the co\-values
from the previous step, and so on until a fixed point is reached (taking three
rounds total).  As such, the co\-values included in $\sem{A}$ can't influence
the result of $\PosCand(\sem{A})$.  Dually, $\NegCand(\sem{A})$ is the
\emph{negative completion of $\sem{A}$} because it is based entirely on
$\sem{A}$'s co\-values in the same manner.
\begin{restatable}[Positive \& Negative Invariance]{lemma}
  {thmposneginvariance}
\label{thm:pos-invariance}
\label{thm:neg-invariance}

For any sound candidates $\sem{A}$ and $\sem{B}$:
\begin{itemize}
\item If the values of $\sem{A}$ and $\sem{B}$ are the same, then
  $\PosCand(\sem{A}) = \PosCand(\sem{B})$.
\item If the (co)\-values of $\sem{A}$ and $\sem{B}$ are the same, then
  $\NegCand(\sem{A}) = \NegCand(\sem{B})$.
\end{itemize}
\end{restatable}
This extra fact gives us another powerful completeness property.  We can reason
about co\-values in the positive candidate $\PosCand(\sem{A})$ purely in terms
of how they interact with $\sem{A}$'s values, ignoring the rest of
$\PosCand(\sem{A})$.  Dually, we can reason about values in the negative
candidate $\NegCand(\sem{A})$ purely based on of $\sem{A}$'s co\-values.
\begin{restatable}[Strong Positive \& Negative Completeness]{corollary}
  {thmposnegcompleteness}
\label{thm:pos-completeness}
\label{thm:neg-completeness}

For any sound candidate $\sem{A}$:
\begin{itemize}
\item $E \in \PosCand(\sem{A})$ if and only if $\cut{V}{E} \in \Bot$ for all
  $V \in \sem{A}$.
\item $V \in \NegCand(\sem{A})$ if and only if $\cut{V}{E} \in \Bot$ for all
  $E \in \sem{A}$.
\end{itemize}
\end{restatable}
\begin{proof}
  Follows directly from \cref{thm:pos-invariance,thm:pos-completion}.
\end{proof}

Now that we know how to turn sound pre-candidates $\sem{A}$ into reducibility
candidates $\PosCand(\sem{A})$ or $\NegCand(\sem{A})$---$\PosCand$ and
$\NegCand$ take anything sound and deliver something sound \emph{and}
complete---we can spell out precisely the subtype-based lattice of reducibility
candidates.

\begin{restatable}[Reducibility Subtype Lattice]{theorem}
  {thmsubtypelattice}
\label{thm:subtype-lattice}

There is a complete lattice of reducibility candidates in $\redcands$ with
respect to subtyping order, with this binary intersection
$\sem{A}\curlywedge\sem{B}$ and union $\sem{A}\curlyvee\sem{B}$
\begin{align*}
  \sem{A} \curlywedge \sem{B}
  &\defeq
  \NegCand(\sem{A} \wedge \sem{B})
  &
  \sem{A} \curlyvee \sem{B}
  &\defeq
  \PosCand(\sem{A} \vee \sem{B})
\end{align*}
and this intersection ($\bigcurlywedge\mathcal{A}$) and union
($\bigcurlyvee\mathcal{A}$) of any set $\mathcal{A} \subseteq \redcands$ of
reducibility candidates
\begin{align*}
  \bigcurlywedge\mathcal{A}
  &\defeq
  \NegCand\left(\bigwedge\mathcal{A}\right)
  &
  \bigcurlyvee\mathcal{A}
  &\defeq
  \PosCand\left(\bigvee\mathcal{A}\right)
\end{align*}
\end{restatable}

\subsection{Interpretation and Adequacy}

\begin{figure}
\centering
\begin{align*}
  \den{A \to B}
  &\defeq
  \NegCand\{\app{V}{E} \mid V \in \den{A}, E \in \den{B}\}
  \\
  &{\phantom{:}}=
  \bigcurlyvee
  \{
  \sem{C} \in \redcands
  \mid
  \forall~ V \in \den{A}, E \in \den{B}.~
  \app{V}{E} \in \sem{C}
  \}
  \\
  \den{\Nat}
  &\defeq
  \bigcurlywedge
  \{
  \sem{C} \in \redcands
  \mid
  (\Zero \in \sem{C})
  ~\conj~
  (\forall~ V \in \sem{C}.~ \Succ V \in \sem{C})
  \}
  \\
  \den{\Stream A}
  &\defeq
  \bigcurlyvee
  \{
  \sem{C} \in \redcands
  \mid
  (\forall E \in \den{A}.~ \Head E \in \sem{C})
  ~\conj~
  (\forall E \in \sem{C}.~ \Tail E \in \sem{C})
  \}
  \\
  \den{\Gamma}
  &\defeq
  \{
  \rho
  \in
  \<Subst>
  \mid
  (\forall (x{\givestype}A)\in\Gamma.~ x\subs{\rho} \in \den{A})
  ~\conj~
  (\forall (\alpha{\takestype}A)\in\Gamma.~ \alpha\subs{\rho} \in \den{A})
  \}
  \\
  \<Subst> \ni \rho
  &::= \asub{x}{V},\dots,\asub{\alpha}{E},\dots
\end{align*}
\begin{align*}
  \den{\Gamma \entails c \contra}
  &\defeq
  \forall \rho \in \den{\Gamma}.~ c\subs{\rho} \in \Bot
  \\
  \den{\Gamma \entails v \givestype A}
  &\defeq
  \forall \rho \in \den{\Gamma}.~ v\subs{\rho} \in \den{A}
  \\
  \den{\Gamma \entails e \takestype A}
  &\defeq
  \forall \rho \in \den{\Gamma}.~ e\subs{\rho} \in \den{A}
\end{align*}
\caption{Model of termination and safety of the abstract machine.}
\label{fig:termination-model}
\end{figure}

Using the subtyping lattice, we have enough infrastructure to define an
interpretation of types and type-checking judgments as given in
\cref{fig:termination-model}.  Each type is interpreted as a reducibility
candidate.
%
Even though we are dealing with recursive types ($\Nat$ and $\Stream$) the
candidates for them are defined in a non-recursive way based on Knaster-Tarski's
fixed point construction \cite{KnasterFixedPoint,TarskiFixedPoint}, and can be
read with these intuitions:
\begin{itemize}
\item $\den{A \to B}$ is the negatively-complete candidate containing all the
  call stacks built from $\den{A}$ arguments and $\den{B}$ return continuations.
  Note that this is the same thing (via
  \cref{thm:neg-completion,thm:neg-invariance}) as \emph{largest} candidate
  containing those call stacks.
\item $\den{\Nat}$ is the \emph{smallest} candidate containing $\Zero$ and
  closed $\Succ$essor constructors.
\item $\den{\Stream A}$ is the \emph{largest} candidate containing $\Head$
  projections expecting an $\den{A}$ element and closed under $\Tail$
  projections.
\end{itemize}
Typing environments ($\Gamma$) are interpreted as the set of valid substitutions
$\rho$ which map variables $x$ to values and co\-variables $\alpha$ to
co\-values.  The interpretation of the typing environment
$\Gamma, x \givestype A$ places an additional requirement on these
substitutions: a valid substitution $\rho \in \den{\Gamma, x \givestype A}$ must
substitute a value of $\den{A}$ for the variable $x$, \ie
$x\subs{\rho} \in \den{A}$.  Dually, a valid substitution
$\rho \in \den{\Gamma, \alpha \takestype A}$ must substitute a co\-value of
$\den{A}$ for $\alpha$, \ie $\alpha\subs{\rho} \in \den{A}$.  Finally, typing
judgments (\eg $\Gamma \entails c \contra$) are interpreted as statements which
assert that the command or (co)\-term belongs to the safe set $\Bot$ or the
assigned reducibility candidate for any substitution allowed by the environment.
The key lemma is that typing derivations of a judgment ensure that the statement
they correspond to holds true.
\begin{restatable}[Adequacy]{lemma}
  {thmadequacy}
\label{thm:adequacy}
\begin{enumerate}
\item If $\Gamma \entails c$ is derivable then $\den{\Gamma \entails c}$ is
  true.
\item If $\Gamma \entails v \givestype A$ is derivable then
  $\den{\Gamma \entails v \givestype A}$ is true.
\item If $\Gamma \entails e \takestype A$ is derivable then
  $\den{\Gamma \entails e \takestype A}$ is true.
\end{enumerate}
\end{restatable}

In order to prove adequacy (\cref{thm:adequacy}), we need to know something more
about which (co)\-terms are in the interpretation of types.  For example, how do
we know that the well-typed call stacks and $\lambda$-abstractions given by the
rules in \cref{fig:t-mach-type-system} end up in $\den{A \to B}$?
Intuitively, function types themselves are non-recursive.  In
\cref{fig:termination-model}, the union over the possible candidates $\sem{C}$
defining $\den{A \to B}$ requires that certain call stacks must be in each
$\sem{C}$, but it does \emph{not} quantify over the (co)\-values already in
$\sem{C}$ to build upon them.  Because of this, $\den{A \to B}$ is equivalent to
the negative candidate $\NegCand\{\app V E \mid V \in \den{A}, E \in \den{B}\}$,
as noted in \cref{fig:termination-model}.  It follows from
\cref{thm:neg-completeness} that $\den{A \to B}$ must contain any value which is
compatible with just these call stacks, regardless of whatever else might be in
$\den{A \to B}$.  This means we can use expansion (\cref{thm:safety-expansion})
to prove that $\den{A \to B}$ contains all $\lambda$-abstractions that, when
given one of these call stacks, step via $\beta_\to$ reduction to a safe
command.
\begin{lemma}[Function Abstraction]
\label{thm:function-abstraction}

If $v\subst{x}{V} \in \den{B}$ for all $V {\in} \den{A}$, then
$\lambda x. v \in \den{A {\to} B}$.
\end{lemma}
\begin{proof}
  Observe that, for any $V \in \den{A}$ and $E \in \den{B}$:
  \begin{align*}
    \cut{\lambda x. v}{\app V E}
    &\srd
    \cut{v\subst{x}{V}}{E}
    \in
    \Bot
  \end{align*}
  where $\cut{v\subst{x}{V}}{E} \in \Bot$ is guaranteed due to soundness for the
  reducibility candidate $\den{B}$.  By expansion (\cref{thm:safety-expansion}),
  we know that $\cut{\lambda x. v}{\app V E} \in \Bot$ as well.  So from
  \cref{thm:neg-completeness}:
  \begin{equation*}
    \lambda x. v
    \in
    \NegCand\{\app V E \mid V \in \den{A}, E \in \den{B}\}
    =
    \den{A \to B}
    \qedhere
  \end{equation*}
\end{proof}

But what about $\den{\Nat}$ and $\den{\Stream A}$?  The interpretations of
(co)\-inductive types are not exactly instances of $\PosCand$ or $\NegCand$ as
written, because unlike $\den{A \to B}$, they quantify over elements in the
possible $\sem{C}$s they are made from.  This lets us say these $\sem{C}$s are
closed under $\Succ$ or $\Tail$, but it means that we cannot identify \emph{a
  priori} a set of (co)\-values that generate the candidate independent of each
$\sem{C}$.

A solution to this conundrum is to instead describe these (co)\-inductive types
incrementally, building them step-by-step instead of all at once \ala Kleene's
fixed point construction \cite{KleeneFixedPoint}.  For example, the set of the
natural numbers can be defined incrementally from a series of finite
approximations by beginning with the empty set, and then at each step adding the
number $0$ and the successor of the previous set:
\begin{align*}
  \mathbb{N}_0 &\defeq \{\}
  &
  \mathbb{N}_{i+1} &\defeq \{0\} \cup \{n+1 \mid n \in \mathbb{N}_i\}
  &
  \mathbb{N} &\defeq \bigcup_{i=0}^\infty \{\mathbb{N}_i\}
\end{align*}
So that each $\mathbb{N}_i$ contains only the numbers less than $i$.  The final
set of \emph{all} the natural numbers, $\mathbb{N}$, is then the union of each
approximation $\mathbb{N}_i$ along the way.  Likewise, we can do a similar
incremental definition of finite approximations of $\den{\Nat}$ like so:
\begin{align*}
  \den{\Nat}_0
  &\defeq
  \PosCand\{\}
  =
  \bigcurlywedge\redcands
  =
  \bigcurlyvee\{\} 
  \\
  \den{\Nat}_{i+1}
  &\defeq
  \PosCand(\{\Zero\} \vee \{\Succ V \mid V \in \den{\Nat}_i\})
\end{align*}
The starting approximation $\den{\Nat}_0$ is the \emph{smallest} reducibility
candidate, which is given by $\PosCand\{\}$.  From there, the next
approximations are given by the positive candidate containing at least $\Zero$
and the $\Succ$essor of every value of their predecessor.  This construction
mimics the approximations $\mathbb{N}_i$, where we start with the smallest
possible base and incrementally build larger and larger reducibility candidates
that more accurately approximate the limit.

The typical incremental co\-inductive definition is usually presented in the
reverse direction: start out with the ``biggest'' set (whatever that is), and
trim it down step-by-step.  Instead, the co\-inductive construction of
reducibility candidates is much more concrete, since they form a complete
lattice (with respect to subtyping).  There is a specific biggest candidate
$\bigcurlyvee \redcands$ (equal to $\bigcurlywedge\{\}$) containing every term
possible and the fewest co\-terms allowed.  From there, we can add explicitly
more co\-terms to each successive approximation, which shrinks the candidate by
ruling out some terms that do not run safely with them.  Thus, the incremental
approximations of $\den{\Stream A}$ are defined negatively as:
\begin{align*}
  \den{\Stream A}_0
  &\defeq
  \NegCand\{\}
  =
  \bigcurlyvee\redcands
  =
  \bigcurlywedge\{\}
  \\
  \den{\Stream A}_{i+1}
  &\defeq
  \NegCand
  (
  \{\Head E \mid E \in \den{A}\}
  \wedge
  \{\Tail E \mid E \in \den{\Stream A}_i\}
  )
\end{align*}
We start with the biggest possible candidate given by $\NegCand\{\}$.  From
there, the next approximations are given by the negative candidate containing at
least $\Head E$ (for any $E$ expecting an element of type $\den{A}$) and the
$\Tail$ of every co\-value in the previous approximation.  The net effect is
that $\den{\Stream A}_i$ definitely contains all continuations built from (at
most) $i$ $\Head$ and $\Tail$ destructors.  As with $\den{\Nat}_i$, the goal is
to show that $\den{\Stream A}$ is the limit of $\den{\Stream A}_i$, \ie that it
is the intersection of all finite approximations.
\begin{restatable}[(Co)Induction Inversion]{lemma}
  {thmsubtypecoinduction}
\label{thm:co-induction-inversion}
\label{thm:nat-inversion}
\label{thm:stream-inversion}

\begin{align*}
  \den{\Nat} &= \bigcurlyvee_{i = 0}^\infty \den{\Nat}_i
  &
  \den{\Stream A}
  &=
  \bigcurlywedge_{i=0}^\infty \den{\Stream A}_i
\end{align*}
\end{restatable}



This fact makes it possible to use expansion (\cref{thm:safety-expansion}) and
strong completeness (\cref{thm:pos-completeness,thm:neg-completeness}) to prove
that the recursor belongs to each of the approximations $\den{\Nat}_i$; and thus
also to $\bigcurlyvee_{i=0}^\infty\den{\Nat}_i=\den{\Nat}$ by induction on $i$.
Dually, the co\-recursor belongs to each approximation $\den{\Stream A}_i$, so
it is included in
$\bigcurlywedge_{i=0}^\infty\den{\Stream A}_i=\den{\Stream A}$.  This proof of
safety for (co)\-recursion is the final step in proving overall adequacy
(\cref{thm:adequacy}).  In turn, the ultimate type safety and termination
property we are after is a special case of adequacy.

\begin{restatable}[Type safety \& Termination]{theorem}
  {thmmachsafetytermination}
\label{thm:mach-corec-type-safety}
\label{thm:mach-corec-termination}
If $\alpha\takestype\Nat \entails c \contra$ in the (co)\-recursive abstract
machine then $c \srds \cut{\Zero}{\alpha}$ or $c \srds \cut{\Succ V}{\alpha}$
for some $V$.
\end{restatable}
\begin{proof}
  From \cref{thm:adequacy}, we have
  $\den{\alpha\takestype\Nat \entails c \contra}$.  That is, for all
  $E \in \den{\Nat}$, $c\subst{\alpha}{E} \in \Bot$.  Note that both
  $\cut{\Zero}{\alpha} \in \Bot$ and $\cut{\Succ V}{\alpha} \in \Bot$ (for any
  $V$ whatsoever) by definition of $\Bot$ (\cref{def:termination-pole}).  So by
  \cref{thm:pos-completeness}, $\alpha$ itself is a member of $\den{Nat}$, \ie
  $\alpha \in \den{\Nat}$.  Thus, $c\subst{\alpha}{\alpha} = c \in \Bot$, and
  the necessary reduction follows from the definition of $\Bot$.
\end{proof}


%% file: sec_related-work.tex
\section{Related Work}
\label{sec:related-work}

The co\-recursor presented here is a computational interpretation of the
categorical model of co\-recursion in a co\-algebra
\cite{GeuversIterationRecursion}.  A (weak) co\-algebra for a functor $F$ is
defined by a morphism $\alpha : A \to F(A)$ for some $A$.  A
\emph{co\-recursive co\-algebra} extends this idea, by commuting with other
morphisms of the form $X \to F(A + X)$.  Intuitively, the option $A + X$ in the
result is interpreted as a pure sum type by \cite{GeuversIterationRecursion}.
Here, we use a different interpretation of $A + X$ as multiple outputs,
represented concretely by multiple continuations.  The continuation-based
interpretation gives improved generality, and can express some co\-recursive
algorithms that the other interpretations cannot
\cite{ClassicalCorecursionProgramming}.

The co\-iterator, which we define as the restriction of the co\-recursor to
never short-cut co\-recursion, corresponds exactly to the Harper's
\texttt{strgen} \cite{HarperPFPL}.  In this sense, the co\-recursor is a
conservative extension of the purely functional co\-iterator.  Co\-iteration
with control operators is considered in \cite{CPSCoInductiveTypes}, which gives
a call-by-name CPS translation for a stream co\-iterator and constructor,
corresponding to $\CoIter$ and $\CoCase$, but not for $\CoRec$.  Here, the use
of an abstract machine serves a similar role as CPS---making explicit
information- and control-flow---but allows us to use the same translation for
both call-by-value and -name.  An alternative approach to (co)\-recursive
combinators is \emph{sized types} \cite{SizedTypes,AbelPhD}, which give the
programmer control over recursion while still ensuring termination, and have
been used for both purely functional \cite{WellfoundedCopatterns} and classical
\cite{StructuralRecursion} coinductive types.

Our investigation on evaluation strategy showed the (dual) impact of
call-by-value versus call-by-name evaluation
\cite{DualityOfComputation,CBVDualToCBN} on the efficiency of (co)\-recursion.
In contrast to having a monolithic evaluation strategy, another approach is to
use a hybrid evaluation strategy as done by call-by-push-value \cite{LevyPhD} or
polarized languages \cite{ZeilbergerPhD,MunchMaccagnoniPhD}.  With a hybrid
approach, we could define one language which has the efficient version of both
the recursor and co\-recursor.  Polarity also allows for incorporating other
evaluation strategies, such as call-by-need which shares the work of
computations \cite{BeyondPolarity,ExtendedCBPV}.  We leave the investigation of
a polarized version of co\-recursion to future work.


%% file: sec_conclusion.tex
\section{Conclusion}
\label{sec:conclusion}

This paper provides a foundational calculus for (co)\-recursion in programs
phrased in terms of an abstract machine language.  The impact of evaluation
strategy is also illustrated, where call-by-value and -name have (opposite)
advantages for the efficiency of co\-recursion and recursion, respectively.
These (co)\-recursion schemes are captured by (co)\-data types whose duality is
made apparent by the language of the abstract machine.  In particular, inductive
data types, like numbers, revolve around constructing concrete, finite values,
so that observations on numbers may be abstract and unbounded.  Dually,
co\-inductive co\-data types, like streams, revolve around concrete, finite
observations, so that values may be abstract and unbounded objects.  The
computational interpretation of this duality lets
us bring out hidden connections
underlying the implementation of recursion and co\-recursion.  For example, the
explicit ``seed'' or accumulator usually used to generate infinite streams is,
in fact, dual to the implicitly growing evaluation context of recursive calls.
To show that the combination of primitive recursion and co\-recursion is
well-behaved---that is, every program safely terminates with an answer---we
interpreted the type system as a form of classical (bi)orthogonality model
capable of handling first-class control effects, and extended with
(co)\-inductive reducibility candidates.  Our model reveals how the incremental
Kleene-style and wholesale Knaster-Tarski-style constructions of greatest and
least fixed points have different advantages for reasoning about program
behavior.  By showing the two fixed point constructions are the same---a
non-trivial task for types of effectful computation---we get a complete picture
of the mechanics of classical (co)\-recursion.


%% file: sec_model.tex
\section{Proof of Type Safety and Termination}
\label{sec:model}

Here we give the full details to the proof of the main result,
\cref{thm:mach-corec-type-safety,thm:mach-corec-termination}, ensuring both
safety and termination for all well-typed, executable commands.  We use a proof
technique suitable for abstract machines based on
\cite{ClassicalStrongNormalization,DualityOfIntersectonUnionTypes}, which we
extend with the inductive type $\Nat$ and the co\-inductive type $\Stream A$.
To begin in \cref{sec:orthogonality,sec:candidate-completion,sec:lattices}, we
give a self-contained introduction and summary of the fundamental concepts and
results from \cite{ClassicalStrongNormalization,DualityOfIntersectonUnionTypes}.
\Cref{sec:candidate-completion} in particular gives a new account of
\emph{positive} and \emph{negative completion} which simplifies the sections
that follow.  From there, \cref{sec:co-inductive-candidates} establishes the
definition and properties of the (co)\-inductive types $\Nat$ and $\Stream A$ in
this model, which lets us prove the fundamental \emph{adequacy} lemma in
\cref{sec:adequacy}.

\subsection{Orthogonal Fixed-Point Candidates}
\label{sec:orthogonality}

Our proof technique revolves around \emph{pre-candidates}
(\cref{def:pre-candidate}) and their more informative siblings
\emph{reducibility candidates}.  The first, and most important, operation on
pre-candidates is \emph{orthogonality}. Intuitively, on the one side
orthogonality identifies \emph{all} the terms which are safe with
\emph{everything} in a given set of co\-terms, and on the other side it identifies
the co\-terms which are safe with a set of terms.  These two dual operations
converting back and forth between terms and co\-terms naturally extends to a
single operation on pre-candidates.
\begin{definition}[Orthogonality]
\label{def:orthogonality}

The \emph{orthogonal} of any set of terms, $\sem{A}^+$, written
$\sem{A}^{+\Bot}$, is the set of co\-terms that form safe commands (\ie in
$\Bot$) with all of $\sem{A}^+$:
\begin{align*}
  \sem{A}^{+\Bot}
  &\defeq
  \{ e \mid \forall v \in \sem{A}^+.~ \cut{v}{e} \in \Bot \}
\end{align*}
Dually, the \emph{orthogonal} of any set of co\-terms $\sem{A}^-$, also written
$\sem{A}^{-\Bot}$ and disambiguated by context, is the set of terms that form
safe commands with all of $\sem{A}^-$:
\begin{align*}
  \sem{A}^{-\Bot}
  &\defeq
  \{ v \mid \forall e \in \sem{A}^-.~ \cut{v}{e} \in \Bot \}
\end{align*}
Finally, the \emph{orthogonal} of any pre-candidate
$\sem{A} = (\sem{A}^+, \sem{A}^-)$ is:
\begin{align*}
  (\sem{A}^+, \sem{A}^-)^\Bot
  &\defeq
  (\sem{A}^{-\Bot}, \sem{A}^{+\Bot})
\end{align*}

As a shorthand for mapping over sets, given any set of pre-candidates
$\set{A} \subseteq \precands$, we write $\set{A}^{\Bot*}$ for the set of
orthogonals to each pre-candidate in $\set{A}$:
\begin{align*}
  \set{A}^{\Bot*}
  &\defeq
  \{\sem{A}^\Bot \mid \sem{A}^\Bot \in \set{A}\}
\end{align*}
We use the same notation for the orthogonals of any set of term-sets
($\set{A} \subseteq \wp(\mathit{Term})$) or co\-term-sets
($\set{A} \subseteq \wp(\mathit{CoTerm})$), individually.
\end{definition}

Orthogonality is interesting primarily because of the logical structure it
creates among pre-candidates.  In particular, orthogonality behaves very much
like \emph{intuitionistic negation} ($\neg$).  Intuitionistic logic rejects
\emph{double negation elimination} ($\neg \neg A \iff A$) in favor of the weaker
principle of \emph{double negation introduction} ($A \implies \neg \neg A$).
This fundamental property of intuitionistic negation is mimicked by
pre-candidate orthogonality.
\begin{property}[Orthogonal Negation]
\label{thm:orthogonal-negation}
\label{thm:orthogonal-contrapositive}
\label{thm:orthogonal-doi}
\label{thm:orthogonal-toe}

The following holds for any pre-candidates $\sem{A}$ and $\sem{B}$:
\begin{enumerate}
\item \emph{Contrapositive (\ie antitonicity):} $\sem{A} \refines \sem{B}$
  implies $\sem{B}^{\Bot} \refines \sem{A}^{\Bot}$.
\item \emph{Double orthogonal introduction (DOI):}
  $\sem{A} \refines \sem{A}^{\Bot\Bot}$.
\item \emph{Triple orthogonal elimination (TOE):}
  $\sem{A}^{\Bot\Bot\Bot} = \sem{A}^{\Bot}$.
\end{enumerate}
\end{property}
\begin{proof}
  \begin{enumerate}
  \item \emph{Contrapositive}: Let $v \in \sem{B}^{\Bot}$ and $e \in \sem{A}$.
    we know $e \in \sem{B}$ (because $\sem{A} \refines \sem{B}$ implies
    $\sem{A}^- \subseteq \sem{B}^-$) and thus $\cut{v}{e} \in \Bot$ (because
    $v \in \sem{B}^{-\Bot}$).  Therefore, $e \in \sem{A}^{\Bot}$ by definition
    of orthogonality (\cref{def:orthogonality}).  Dually, given any
    $e \in \sem{B}^{\Bot}$ and $v \in \sem{A}$, we know $v \in \sem{B}$ and thus
    $\cut{v}{e} \in \Bot$, so $e \in \sem{B}^{\Bot}$ as well.
  \item \emph{DOI}: Suppose $v \in \sem{A}$.  For any $e \in \sem{A}^{\Bot}$, we
    know $\cut{v}{e} \in \Bot$ by definition of orthogonality
    (\cref{def:orthogonality}).  Therefore, $v \in \sem{A}^{\Bot\Bot}$ also by
    definition of orthogonality.  Dually, every $e \in \sem{A}$ yields
    $\cut{v}{e} \in \Bot$ for all $v \in \sem{A}^\Bot$, so
    $e \in \sem{A}^{\Bot\Bot}$ as well.
  \item \emph{TOE}: Note that $\sem{A} \refines \sem{A}^{\Bot\Bot}$ is an
    instance of double orthogonal introduction above for $\sem{A}$, so by
    contrapositive, $\sem{A}^{\Bot\Bot\Bot} \refines \sem{A}^{\Bot}$.  Another
    instance of double orthogonal introduction for $\sem{A}^{\Bot}$ is
    $\sem{A}^{\Bot} \refines \sem{A}^{\Bot\Bot\Bot}$.  Thus the two
    pre-candidates are equal.
    \qedhere
  \end{enumerate}
\end{proof}

The second operation on pre-candidates is the \emph{(co)\-value restriction}.
This just limits a given pre-candidate to only the values and co\-values
contained within it, and gives us a way to handle the chosen evaluation strategy
(here, call-by-name or call-by-value) in the model.  In particular, the
(co)\-value restriction is useful for capturing the \emph{completeness}
requirement of reducibility candidates (\cref{def:reducibility-candidate}),
which only tests (co)\-terms with respect to the (co)\-values already in the
candidate.
\begin{definition}[(Co)\-value Restriction]
\label{def:co-value-restriction}

The \emph{(co)\-value restriction} of a set of terms $\sem{A}^+$, set of
co\-terms $\sem{A}^-$, and pre-candidates $\sem{A} = (\sem{A}^+,\sem{A}^-)$ is:
\begin{align*}
  \sem{A}^{+v}
  &\defeq
  \{V \mid V \in \sem{A}^+\}
  &
  \sem{A}^{-v}
  &\defeq
  \{E \mid E \in \sem{A}^-\}
  &
  (\sem{A}^+,\sem{A}^-)^v
  &\defeq
  (\sem{A}^{+v},\sem{A}^{-v})
\end{align*}
As another shorthand, given any set of pre-candidates $\set{A}$, we will
occasionally write $\set{A}^{v*}$ to be the set of (co)\-value restrictions
of each pre-candidate in $\set{A}$:
\begin{align*}
  \set{A}^{v*}
  &\defeq
  \{\sem{A}^v \mid \sem{A} \in \set{A}\}
\end{align*}
We use the same notation for the (co)\-value restriction of any set of term-sets
or co\-term-sets.
\end{definition}

\begin{property}[Monotonicity]
\label{thm:subtype-monotonicity}

Given any pre-candidates $\sem{A}$ and $\sem{B}$,
\begin{enumerate}
\item $\sem{A} \leq \sem{B}$ implies $\sem{A}^\Bot \leq \sem{B}^\Bot$, and
\item $\sem{A} \leq \sem{B}$ implies $\sem{A}^v \leq \sem{B}^v$.
\item $\sem{A} \refines \sem{B}$ implies $\sem{A}^v \refines \sem{B}^v$.
\end{enumerate}
\end{property}
\begin{proof}
  Subtype monotonicity of orthogonality follows from contrapositive
  (\cref{thm:orthogonal-contrapositive}) and the opposed definitions of
  refinement versus subtyping.  Specifically, $\sem{A} \leq \sem{B}$ means the
  same thing as $(\sem{A}^+,\sem{B}^-) \refines (\sem{B}^+,\sem{A}^-)$, which
  contrapositive (\cref{thm:orthogonal-contrapositive}) turns into
  \begin{math}
    (\sem{B}^+,\sem{A}^-)^\Bot
    =
    (\sem{A}^{-\Bot},\sem{B}^{+\Bot})
    \refines
    (\sem{B}^{-\Bot},\sem{A}^{+\Bot})
    =
    (\sem{A}^+,\sem{B}^-)^\Bot
  \end{math}
  which is equivalent to $\sem{A}^\Bot \leq \sem{B}^\Bot$.  Monotonicity of the
  (co)\-value restriction with respect to both subtyping and refinement follows
  directly from its definition.
\end{proof}

Putting the two operations together, \emph{(co)\-value restricted orthogonality}
($\sem{A}^{v\Bot}$) becomes our primary way of handling reducibility candidates.
This combined operation shares essentially the same negation-inspired properties
of plain orthogonality (\cref{thm:orthogonal-negation}), but is restricted to
just (co)\-values rather than general (co)\-terms.
\begin{property}[Restricted Orthogonal Negation]
\label{thm:restricted-orthogonal-negation}
\label{thm:restriction-idempotent}
\label{thm:restricted-orthogonal}
\label{thm:restricted-doi}
\label{thm:restricted-toe}

Given any pre-candidate $\sem{A}$:
\begin{enumerate}
\item \emph{Restriction idempotency:}
  $\sem{A}^{vv} = \sem{A}^{v} \refines \sem{A}$
\item \emph{Restricted orthogonal:}
  $\sem{A}^{\Bot} \refines \sem{A}^{v\Bot}$
\item \emph{Restricted double orthogonal introduction (DOI):}
  $\sem{A}^{v} \refines \sem{A}^{v\Bot v\Bot v}$.
\item \emph{Restricted triple orthogonal elimination (TOE):}
  $\sem{A}^{v\Bot v\Bot v\Bot v} = \sem{A}^{v\Bot v}$.
\end{enumerate}
\end{property}
\begin{proof}
  \begin{enumerate}
  \item Because $V \in \sem{A}$ if and only if $V \in \sem{A}^{v}$ (and
    symmetrically for co\-values).
  \item Follows from contrapositive (\cref{thm:orthogonal-contrapositive}) of
    the above fact that $\sem{A}^{v} \refines \sem{A}$.
  \item Double orthogonal introduction (\cref{thm:orthogonal-doi}) on
    $\sem{A}^{v}$ gives
    $\sem{A}^{v} \refines \sem{A}^{v \Bot \Bot}$.  The restricted
    orthogonal (above) of $\sem{A}^{v\Bot}$ implies
    $\sem{A}^{v\Bot \Bot} \refines \sem{A}^{v\Bot v\Bot}$.  Thus from
    monotonicity (\cref{thm:subtype-monotonicity}) and restriction idempotency,
    we have:
    \begin{math}
      \sem{A}^{v}
      \refines
      \sem{A}^{v \Bot \Bot v}
      \refines
      \sem{A}^{v \Bot v \Bot v}
      .
    \end{math}
  \item Follows similarly to triple orthogonal elimination in
    (\cref{thm:orthogonal-toe}).
    $\sem{A}^{v} \refines \sem{A}^{v\Bot v\Bot v}$ is an instance of restricted
    double orthogonal introduction above, and by contrapositive
    (\cref{thm:orthogonal-contrapositive}) and monotonicity
    (\cref{thm:subtype-monotonicity}),
    $\sem{A}^{v\Bot v\Bot v} \refines \sem{A}^{v \Bot v}$.  Another instance of
    restricted double orthogonal introduction on $\sem{A}^{v \Bot v}$ is
    $\sem{A}^{v \Bot v} \refines \sem{A}^{v\Bot v\Bot v}$.  Thus, the two sets
    are equal.
    \qedhere
  \end{enumerate}
\end{proof}



With these restricted logical properties, we can recast the \emph{soundness} and
\emph{completeness} properties of reducibility candidates
(\cref{def:reducibility-candidate}) in terms of orthogonality to show
reducibility candidates are exactly the same as fixed points of (co)\-value
restricted orthogonality.
\begin{lemma}[Fixed-Point Candidates]
\label{thm:fixed-point-candidate}
\

\begin{enumerate}
\item A pre-candidate $\sem{A}$ is sound if and only if
  $\sem{A} \refines \sem{A}^\Bot$.
\item A pre-candidate $\sem{A}$ is complete if and only if
  $\sem{A}^{v\Bot} \refines \sem{A}$.
\item A pre-candidate $\sem{A}$ is a reducibility candidate if and only if
  $\sem{A} = \sem{A}^{v\Bot}$.
\item Every reducibility candidate is a fixed point of orthogonality:
  $\sem{A} \in \redcands$ implies $\sem{A} = \sem{A}^\Bot$.
\end{enumerate}
\end{lemma}
\begin{proof}
  Unfolding the definitions of orthogonality (\cref{def:orthogonality}) and the
  (co)-value restriction (\cref{def:co-value-restriction}) shows that the first
  two refinements are equivalent to soundness and completeness from
  \cref{def:sound-candidate,def:complete-candidate}.

  For the last fact, first recall $\sem{A}^{\Bot} \refines \sem{A}^{v\Bot}$
  (\cref{thm:restricted-orthogonal}).  So if a pre-candidate $\sem{A}$ is both
  sound and complete, $\sem{A} = \sem{A}^{v\Bot} = \sem{A}^\Bot$ because
  \begin{equation*}
    \sem{A}^{\Bot} \refines \sem{A}^{v\Bot}
    \refines
    \sem{A}
    \refines
    \sem{A}^{\Bot} \refines \sem{A}^{v\Bot}
  \end{equation*}
  Going the other way, suppose that $\sem{A} = \sem{A}^{v\Bot}$.  Completeness
  is guaranteed by definition, but what of soundness?  Suppose that $v$ and $e$
  come from the fixed point pre-candidate
  $\sem{A}=\sem{A}^{v\Bot}=(\sem{A}^{-v\Bot},\sem{A}^{+v\Bot})$.  The reason why
  $\sem{A}=\sem{A}^{v\Bot}$ forces $\cut{v}{e} \in \Bot$ depends on the
  evaluation strategy.
  \begin{itemize}
  \item \emph{Call-by-value}, where every co\-term is a co\-value.  Thus, the
    positive requirement on terms of reducibility candidates is equivalent to:
    $v \in \sem{A}^+$ if and only if, for all $e \in \sem{A}^-$,
    $\cut{v}{e} \in \Bot$.
  \item \emph{Call-by-name}, where every term is a value.  Thus, the negative
    requirement on coterms of reducibility candidates is equivalent to:
    $e \in \sem{A}^-$ if and only if, for all $v \in \sem{A}^+$,
    $\cut{v}{e} \in \Bot$.
  \end{itemize}
  In either case, $\cut{v}{e} \in \Bot$ for one of the above reasons, since
  $v,e\in\sem{A}=\sem{A}^{v\Bot}$.
\end{proof}

\subsection{Positive and Negative Completion}
\label{sec:candidate-completion}

Now that we know reducibility candidates are the same thing as fixed points of
(co)\-value restricted orthogonality ($\blank^{v\Bot}$), we have a direct method
to define the completion of a sound pre-candidate into a sound \emph{and}
complete one.  To complete some $\sem{A}$, there are two opposite points of
view: ($\PosCand$) start with the terms of $\sem{A}$ and build everything else
around those, or ($\NegCand$) start with the co\-terms of $\sem{A}$ and build
around them.  Both of these definitions satisfy all the defining criteria
promised by
\cref{thm:pos-completion,thm:neg-completion,thm:pos-invariance,thm:neg-invariance}
due to the logical properties of orthogonality
(\cref{thm:restricted-orthogonal-negation}).
\begin{definition}[Positive \& Negative Reducibility Candidates]
\label{def:pos-candidate}
\label{def:neg-candidate}

For any sound candidate $\sem{A} \in \soundcands$, the positive
($\PosCand(\sem{A})$) and the negative ($\NegCand(\sem{A})$) completion of
$\sem{A}$ are:
\begin{align*}
  \PosCand(\sem{A})
  &=
  (\sem{A}^+, \sem{A}^{+v\Bot})^{v\Bot v\Bot}
  =
  (\sem{A}^{+v\Bot v\Bot}, \sem{A}^{+v\Bot v\Bot v\Bot})
  \\
  \NegCand(\sem{A})
  &=
  (\sem{A}^{-v\Bot}, \sem{A}^-)^{v\Bot v\Bot}
  =
  (\sem{A}^{-v\Bot v\Bot v\Bot}, \sem{A}^{-v\Bot v\Bot})
\end{align*}
\end{definition}

\thmposnegcompletion*
\begin{proof}
  The definitions given in \cref{def:pos-candidate,def:neg-candidate} satisfy
  all three requirements:
  \begin{enumerate}
  \item \emph{They are reducibility candidates:} Observe that by restricted
    triple orthogonal elimination (\cref{thm:restricted-toe}),
    $\PosCand(\sem{A})$ and $\NegCand(\sem{A})$ are reducibility candidates
    because they are fixed-points of $\blank^{v\Bot}$
    (\cref{thm:fixed-point-candidate}):
    \begin{align*}
      &
      (\PosCand(\sem{A}))^{v\Bot}
      &
      &
      (\NegCand(\sem{A}))^{v\Bot}
      \\
      &\quad=
      (\sem{A}^{+v\Bot v\Bot}, \sem{A}^{+v\Bot v\Bot v\Bot})^{v\Bot}
      &
      &\quad=
      (\sem{A}^{-v\Bot v\Bot v\Bot}, \sem{A}^{-v\Bot v\Bot})^{v\Bot}
      &\text{(\cref{def:pos-candidate,def:neg-candidate})}
      \\
      &\quad=
      (\sem{A}^{+v\Bot v\Bot v\Bot v\Bot}, \sem{A}^{+v\Bot v\Bot v\Bot})
      &
      &\quad=
      (\sem{A}^{-v\Bot v\Bot v\Bot}, \sem{A}^{-v\Bot v\Bot v\Bot v\Bot})
      &\text{(\cref{def:orthogonality})}
      \\
      &\quad=
      (\sem{A}^{+v\Bot v\Bot}, \sem{A}^{+v\Bot v\Bot v\Bot})
      &
      &\quad=
      (\sem{A}^{-v\Bot v\Bot v\Bot}, \sem{A}^{-v\Bot v\Bot})
      &\text{(\cref{thm:restricted-toe})}
      \\
      &\quad=
      \PosCand(\sem{A})
      &
      &\quad=
      \NegCand(\sem{A})
      &\text{(\cref{def:pos-candidate,def:neg-candidate})}
    \end{align*}
  \item \emph{They are (co)\-value extensions:} First, note that
    \begin{align*}
      \sem{A}^{+v}
      \subseteq
      \sem{A}^{+v\Bot v\Bot v}
      &=
      \PosCand(\sem{A})^{+v}
      \subseteq
      \PosCand(\sem{A})^+
      \\
      \sem{A}^{-v}
      \subseteq
      \sem{A}^{-v\Bot v\Bot v}
      &=
      \NegCand(\sem{A})^{-v}
      \subseteq
      \NegCand(\sem{A})^-
    \end{align*}
    by restricted double orthogonal introduction (\cref{thm:restricted-doi}).
    Furthermore, soundness of $\sem{A}$ means $\sem{A} \refines \sem{A}^\Bot$
    (\ie $\sem{A}^+ \subseteq \sem{A}^{-\Bot}$ and
    $\sem{A}^- \subseteq \sem{A}^{+\Bot}$), so again by
    \cref{thm:restricted-doi}:
    \begin{align*}
      \sem{A}^{-v}
      \subseteq
      \sem{A}^{+\Bot v}
      \subseteq
      \sem{A}^{+v\Bot v}
      \subseteq
      \sem{A}^{+v\Bot v\Bot v\Bot v}
      &=
      \PosCand(\sem{A})^{-v}
      \subseteq
      \PosCand(\sem{A})^-
      \\
      \sem{A}^{+v}
      \subseteq
      \sem{A}^{-\Bot v}
      \subseteq
      \sem{A}^{-v\Bot v}
      \subseteq
      \sem{A}^{-v\Bot v\Bot v\Bot v}
      &=
      \NegCand(\sem{A})^{+v}
      \subseteq
      \NegCand(\sem{A})^+
    \end{align*}
  \item \emph{They are the least/greatest such candidates:} Suppose there is a
    reducibility candidate $\sem{C}$ such that $\sem{A}^v \refines \sem{C}$.
    Because $\sem{C}$ is a fixed point of $\blank^{v\Bot}$
    (\cref{thm:fixed-point-candidate}), iterating contrapositive
    (\cref{thm:orthogonal-contrapositive}) on this refinement gives:
    \begin{align*}
      \sem{C}
      =
      \sem{C}^{v\Bot}
      &\refines
      \sem{A}^{v v\Bot}
      =
      \sem{A}^{v\Bot}
      &
      \sem{A}^{v\Bot v\Bot}
      &\refines
      \sem{C}^{v\Bot}
      =
      \sem{C}
      &
      \sem{C}
      =
      \sem{C}^{v\Bot}
      &\refines
      \sem{A}^{v\Bot v\Bot v\Bot}
    \end{align*}
    Expanding the definition of $\PosCand$, $\NegCand$, and refinement, this
    means:
    \begin{align*}
      \PosCand(\sem{A})^+
      =
      \sem{A}^{+v\Bot v\Bot}
      &\subseteq
      \sem{C}
      &
      \PosCand(\sem{A})^-
      =
      \sem{A}^{+v\Bot v\Bot v\Bot}
      &\supseteq
      \sem{C}
      \\
      \NegCand(\sem{A})^+
      =
      \sem{A}^{-v\Bot v\Bot v\Bot}
      &\supseteq
      \sem{C}
      &
      \NegCand(\sem{A})^-
      =
      \sem{A}^{-v\Bot v\Bot}
      &\subseteq
      \sem{C}
    \end{align*}
    Or in other words,
    \begin{math}
      \PosCand(\sem{A})
      \leq
      \sem{C}
      \leq
      \NegCand(\sem{A})
      .
    \end{math}
    \qedhere
  \end{enumerate}
\end{proof}

\thmposneginvariance*
\begin{proof}
  Because the definition of $\PosCand(\sem{A})$ depends only on $\sem{A}^+$ and
  not $\sem{A}^-$, and dually the definition of $\NegCand(\sem{A})$ depends only
  on $\sem{A}^-$.
\end{proof}





In addition to these defining properties of $\PosCand$ and $\NegCand$, the two
completions are also \emph{idempotent} (\ie they are closure operations, because
multiple applications are the same as just one) and \emph{monotonic} (\ie they
preserve the subtyping order, by converting any two sound subtype candidates to
two sound and complete subtype reducibility candidates).
\begin{corollary}[Idempotency]
\label{thm:pos-idempotency}
\label{thm:neg-idempotency}

For all reducibility candidates $\sem{A}$,
$\PosCand(\sem{A}) = \sem{A} = \NegCand(\sem{A})$.  It follows that, for all
sound candidates $\sem{A}$:
\begin{align*}
  \PosCand(\PosCand(\sem{A}))
  &=
  \PosCand(\sem{A})
  &
  \NegCand(\NegCand(\sem{A}))
  &=
  \NegCand(\sem{A})
\end{align*}
\end{corollary}
\begin{proof}
  $\PosCand(\sem{A}) = \sem{A} = \NegCand(\sem{A})$ follows from
  \cref{def:pos-candidate,def:neg-candidate} because the reducibility candidate
  $\sem{A}$ is a fixed point of $\blank^{v\Bot}$
  (\cref{thm:fixed-point-candidate}).  The idempotency of $\PosCand$ and
  $\NegCand$ is then immediate from the fact that they produce reducibility
  candidates from any sound candidate.
\end{proof}

\begin{lemma}[Monotonicity]
\label{thm:pos-monotonicity}
\label{thm:neg-monotonicity}

Given any sound candidates $\sem{A} \leq \sem{B}$:
\begin{enumerate*}[]
\item $\PosCand(\sem{A}) \leq \PosCand(\sem{B})$, and
\item $\NegCand(\sem{A}) \leq \NegCand(\sem{B})$.
\end{enumerate*}
\end{lemma}
\begin{proof}
  Given $\sem{A} \leq \sem{B}$,
  \cref{thm:pos-completion,thm:neg-completion,thm:pos-invariance,thm:neg-invariance}
  imply that
  \begin{align*}
    \PosCand(\sem{A})
    =
    \PosCand(\sem{A}^+,\{\})
    &\leq
    \PosCand(\sem{B}^+,\{\})
    =
    \PosCand(\sem{B})
    \\
    \NegCand(\sem{A})
    =
    \NegCand(\{\},\sem{A}^-)
    &\leq
    \NegCand(\{\},\sem{B}^-)
    =
    \NegCand(\sem{B})
  \end{align*}
  because $\sem{A}^v \leq \sem{B}^v$ (\cref{thm:subtype-monotonicity}), which
  means $\sem{A}^{+v} \subseteq \sem{B}^{+v}$ and
  $\sem{A}^{-v} \supseteq \sem{B}^{-v}$ by definition of subtyping.  Thus from
  \cref{thm:pos-completion}, we know that
  $(\sem{A}^{+v},\{\})\refines(\sem{B}^{+v},\{\})\refines\PosCand(\sem{B}^+,\{\})$
  and $\PosCand(\sem{A}^+,\{\})$ is the least one to do so, forcing
  $\PosCand(\sem{A}^+,\{\}) \leq \PosCand(\sem{B}^+,\{\})$.  Likewise from
  \cref{thm:neg-completion}, we know that
  $(\{\},\sem{B}^{-v})\refines(\{\},\sem{A}^{v-})\refines\NegCand(\{\},\sem{A}^{-v})$
  and $\NegCand(\sem{B}^-,\{\})$ is the greatest one to do so, forcing
  $\NegCand(\{\},\sem{A}^-) \leq \NegCand(\{\},\sem{B}^-)$.
\end{proof}

\subsection{Refinement and Subtyping Lattices}
\label{sec:lattices}

Because pre-candidates have two different orderings
(\cref{def:refinement-order,def:subtype-order}), they also have two very
different lattice structures.  We are primarily interested in the
\emph{subtyping lattice} because it is compatible with both soundness and
completeness in both directions.  In particular, the na\"ive subtype lattice
as-is always preserves soundness, and combined with the dual completions
($\PosCand$ and $\NegCand$) the subtype lattice preserves completeness as well.
This gives us a direct way to assemble complex reducibility candidates from
simpler ones.

\begin{theorem}[Sound Subtype Lattice]
\label{thm:sound-subtype-lattice}

The subtype intersection $\bigwedge$ and union $\bigvee$ forms a complete
semi-lattice over sound candidates in $\soundcands$.
\end{theorem}
\begin{proof}
  Let $\set{A} \subseteq \soundcands$ be a set of sound candidates, and suppose
  $v, e \in \bigwedge\set{A}$.  By definition:
  \begin{itemize}
  \item for all $\sem{A} \in \set{A}$, $v \in \sem{A}$, and
  \item there exists an $\sem{A} \in \set{A}$ such that $e \in \sem{A}$.
  \end{itemize}
  Therefore, we know that $v \in \sem{A}$ for the particular sound candidate
  that $e$ inhabits, and thus $\cut{v}{e}$ by soundness of $\sem{A}$.  Soundness
  of $\bigvee\set{A}$ follows dually, because $v, e \in \bigvee \set{A}$
  implies:
  \begin{itemize}
  \item there exists an $\sem{A} \in \set{A}$ such that $v \in \sem{A}$, and
  \item for all $\sem{A} \in \set{A}$, $e \in \sem{A}$.
    \qedhere
  \end{itemize}
\end{proof}

\thmsubtypelattice*
\begin{proof}
  Let $\set{A} \subseteq \redcands$ be any set of reducibility candidates.
  First, note that $\bigwedge\set{A}$ and $\bigvee\set{A}$ are sound
  (\cref{thm:sound-subtype-lattice}) because every reducibility candidate is
  sound.  Thus, for all $\sem{A} \in \set{\mathcal{A}}$, monotonicity
  (\cref{thm:pos-monotonicity,thm:neg-monotonicity}) and idempotency
  (\cref{thm:pos-idempotency,thm:neg-idempotency}) of $\PosCand$ and $\NegCand$
  implies:
  \begin{align*}
    \bigwedge\set{A} &\leq \sem{A}
    &
    \bigvee\set{A} &\leq \sem{A}
    \\
    \bigcurlywedge\set{A} = \PosCand\bigwedge\set{A}
    &\leq
    \PosCand(\sem{A}) = \sem{A}
    &
    \bigcurlyvee\set{A} = \NegCand\bigvee\set{A}
    &\leq
    \NegCand(\sem{A}) = \sem{A}
  \end{align*}
  \Cref{thm:pos-monotonicity,thm:neg-monotonicity,thm:pos-idempotency,thm:neg-idempotency}
  also imply that these are the tightest such bounds.  Suppose there are
  reducibility candidates $\sem{B}$ and $\sem{C}$ such that
  \begin{align*}
    \forall \sem{A} \in \set{A}.~
    \sem{B} &\leq \sem{A}
    &
    \forall \sem{A} \in \set{A}.~
    \sem{A} &\leq \sem{C}
  \end{align*}
  From the lattice properties of $\bigwedge$ and $\bigvee$, monotonicity, and
  idempotency, we have:
  \begin{align*}
    \sem{B} &\leq \bigwedge\set{A}
    &
    \bigvee\set{A} &\leq \sem{C}
    \\
    \sem{B} = \PosCand(\sem{B})
    &\leq
    \PosCand\bigwedge\set{A} = \bigcurlywedge\set{A}
    &
    \bigcurlyvee\set{A} = \NegCand\bigvee\set{A}
    &\leq
    \NegCand(\sem{C}) = \sem{C}
    \qedhere
  \end{align*}
\end{proof}

The other lattice is based on refinement, instead of subtyping.  In contrast,
the refinement lattice has a opposing relationship with soundness and
completeness: one direction of the lattice preserves only soundness, and the
other one preserves only completeness.
\begin{definition}[Refinement Lattice]
\label{def:refinement-lattice}

There is a complete lattice of pre-candidates in $\precands$ with respect to
refinement order, where the binary intersection ($\sem{A} \sqcap \sem{B}$ \aka
\emph{meet}) and union ($\sem{A} \sqcap \sem{B}$, \aka \emph{join}) are defined
as:
  \begin{align*}
    \sem{A} \sqcap \sem{B}
    &\defeq
    (\sem{A}^+ \cap \sem{B}^+, \sem{A}^- \cap \sem{B}^-)
    &
    \sem{A} \sqcup \sem{B}
    &\defeq
    (\sem{A}^+ \cup \sem{B}^+, \sem{A}^- \cup \sem{B}^-)
  \end{align*}
  Moreover, these generalize to the intersection ($\bigsqcap\sem{A}$, \aka
  infimum) and union ($\bigsqcup\sem{A}$, \aka supremum) of any set
  $\set{A} \in \precands$ of pre-candidates
  \begin{align*}
    \bigsqcap\set{A}
    &\defeq
    \left(
      \bigcap\{\sem{A}^+ \mid \sem{A} \in \set{A}\}
      ,
      \bigcap\{\sem{A}^- \mid \sem{A} \in \set{A}\}    
    \right)
    \\
    \bigsqcup\set{A}
    &\defeq
    \left(
      \bigcup\{\sem{A}^+ \mid \sem{A} \in \set{A}\}
      ,
      \bigcup\{\sem{A}^- \mid \sem{A} \in \set{A}\}    
    \right)
  \end{align*}
\end{definition}

\begin{theorem}[Sound and Complete Refinement Semi-Lattices]
\label{thm:sound-refinement-semi-lattice}
\label{thm:complete-refinement-semi-lattice}

The refinement intersection $\bigsqcap$ forms a meet semi-lattice over sound
candidates in $\soundcands$, and the refinement union $\bigsqcup$ forms a join
semi-lattice over complete candidates in $\complcands$.
\end{theorem}
\begin{proof}
  Let $\set{A} \in \soundcands$ be a set of sound candidates, \ie for all
  $\sem{A} \in \set{A}$, we know $\sem{A} \refines \sem{A}^\Bot$.  In the
  refinement lattice on pre-candidates, de Morgan duality
  (\cref{thm:orthogonal-de-morgan}) implies:
  \begin{equation*}
    \forall \sem{A} \in \set{A}.~
    \bigsqcap\set{A}
    \refines
    \sem{A}
    \refines
    \sem{A}^\Bot
    \refines
    \bigsqcup(\set{A}^{\Bot*})
    \refines
    \left(\bigsqcap\set{A}\right)^\Bot
  \end{equation*}
  So that $\bigsqcap\set{A}$ is also sound.

  Let $\set{A} \in \complcands$ be a set of complete candidates, \ie for all
  $\sem{A} \in \set{A}$, we know $\sem{A}^{v\Bot} \refines \sem{A}$.  In the
  refinement lattice on pre-candidates, de Morgan duality
  (\cref{thm:orthogonal-de-morgan}) and the fact that the (co)\-value
  restriction $\blank^v$ distributes over unions implies:
  \begin{equation*}
    \forall \sem{A} \in \set{A}.~
    \bigsqcup\set{A}
    \extends
    \sem{A}
    \extends
    \sem{A}^{v\Bot}
    \extends
    \bigsqcap(\set{A}^{v*\Bot*})
    =
    \bigsqcup(\set{A})^{v\Bot}
  \end{equation*}
  So that $\bigsqcup\set{A}$ is also complete.
\end{proof}

Because soundness and completeness are each broken by different directions of
this refinement lattice, it doesn't give us a complete lattice for assembling
new reducibility candidates.  However, what it does give us is additional
insight into the logical properties of orthogonality.  That is, while
orthogonality behaves like intuitionistic negation, the intersections
($\bigsqcap$) and unions ($\bigsqcup$) act like conjunction and disjunction,
respectively.  Together, these give us properties similar to the familiar de
Morgan laws of duality intuitionistic logic.
\begin{property}[Orthogonal de Morgan]
\label{thm:orthogonal-de-morgan}

Given any pre-candidates $\sem{A}$ and $\sem{B}$:
\begin{enumerate}
\item
  \begin{math}
    (\sem{A} \sqcup \sem{B})^\Bot
    =
    (\sem{A}^{\Bot}) \sqcap (\sem{B}^{\Bot})
    .
  \end{math}
\item
  \begin{math}
    (\sem{A} \sqcap \sem{B})^\Bot
    \extends
    (\sem{A}^{\Bot}) \sqcup (\sem{B}^{\Bot})
    .
  \end{math}
\item
  \begin{math}
    (\sem{A}^{\Bot} \sqcap \sem{B}^{\Bot})^{\Bot\Bot}
    =
    (\sem{A}^{\Bot}) \sqcap (\sem{B}^{\Bot})
    =
    (\sem{A} \sqcup \sem{B})^{\Bot}
    =
    (\sem{A}^{\Bot\Bot} \sqcup \sem{B}^{\Bot\Bot})^\Bot
    .
  \end{math}
\end{enumerate}
Furthermore, given any set of pre-candidates $\set{A} \subseteq \precands$:
\begin{enumerate}
\item
  \begin{math}
    (\bigcup\set{A})^\Bot
    =
    \bigcap(\set{A}^{\Bot*})
    .
  \end{math}
\item
  \begin{math}
    (\bigcap\set{A})^\Bot
    \extends
    \bigcup(\set{A}^{\Bot*})
    .
  \end{math}
\item 
  \begin{math}
    (\bigcap(\set{A}^{\Bot*}))^{\Bot\Bot}
    =
    \bigcap(\set{A}^{\Bot*})
    =
    \left(\bigcup\set{A}\right)^{\Bot}
    =
    \left(\bigcup\set{A}^{\Bot*\Bot*}\right)^{\Bot}
    .
  \end{math}
\end{enumerate}
\end{property}
\begin{proof}
  We will show only the de Morgan properties for union and intersection over any
  sets of pre-candidate $\set{A}$; the binary versions are special cases of
  these.  Note that the union and intersection of the refinement lattice on
  pre-candidates have these lattice properties:
  \begin{align*}
    \forall \sem{A} \in \set{A}.~
    &
    \sem{A} \refines \bigsqcup\set{A}
    &
    (\forall \sem{A} \in \set{A}.~ \sem{A} \refines \sem{C})
    &\implies
    \bigsqcup\set{A} \refines \sem{C}
    \\
    \forall \sem{A} \in \set{A}.~
    &
    \bigsqcap\set{A} \refines \sem{A}
    &
    (\forall \sem{A} \in \set{A}.~ \sem{C} \refines \sem{A})
    &\implies
    \sem{C} \refines \bigsqcap\set{A} 
  \end{align*}
  Taking the contrapositive (\cref{thm:orthogonal-contrapositive}) to the facts
  on the left, and instantiating the facts on the right to $\set{A}^{\Bot*}$,
  gives:
  \begin{align*}
    \forall \sem{A} \in \set{A}.~
    &
    \left(\bigsqcup\set{A}\right)^\Bot \refines \sem{A}^{\Bot}
    &
    (\forall \sem{A} \in \set{A}.~ \sem{A}^{\Bot} \refines \sem{C})
    &\implies
    \bigsqcup(\set{A}^{\Bot*}) \refines \sem{C}
    \\
    \forall \sem{A} \in \set{A}.~
    &
    \sem{A}^{\Bot} \refines \left(\bigsqcap\set{A}\right)^\Bot
    &
    (\forall \sem{A} \in \set{A}.~ \sem{C} \refines \sem{A}^{\Bot})
    &\implies
    \sem{C} \refines \bigsqcap(\set{A}^{\Bot*})
  \end{align*}
  
  \begin{enumerate}
  \item We know $(\bigsqcup\set{A})^\Bot \refines \sem{A}^{\Bot}$ (for each
    $\sem{A}^{\Bot} \in \set{A}$), so
    $(\bigsqcup\set{A})^\Bot \refines \bigsqcap(\set{A}^{\Bot*})$.  In the
    reverse direction, suppose $v \in \bigsqcap(\set{A}^{\Bot*})$.  For every
    $e \in \bigsqcup\set{A}$, we know there is (at least) one
    $\sem{A} \in \set{A}$ such that $e \in \sem{A}$.  So since
    $v \in \bigsqcap(\set{A}^{\Bot*}) \refines \sem{A}^{\Bot}$, we know
    $\cut{v}{e} \in \Bot$ by definition of orthogonality
    (\cref{def:orthogonality}).  Therefore, $v \in (\bigsqcup\set{A})^\Bot$ as
    well.  Dually, for every $e \in \bigsqcap(\set{A}^{\Bot*})$ and
    $v \in \bigsqcup\set{A}$, there it at least one $v \in \sem{A} \in \set{A}$,
    forcing $\cut{v}{e} \in \Bot$ and thus $e \in (\bigsqcup\set{A})^\Bot$.  So
    in general $\bigsqcap(\set{A}^{\Bot*}) \refines (\bigsqcup\set{A})^\Bot$,
    making the two sets
    equal.
  \item We know $\sem{A}^{\Bot} \refines (\bigsqcap\set{A})^\Bot$ (for each
    $\sem{A}^{\Bot} \in \set{A}$), so
    $\bigsqcup(\set{A}^{\Bot*}) \refines (\bigsqcap\set{A})^\Bot$.  But the
    reverse direction may not be true:
    $(\bigsqcap\set{A})^\Bot \not\refines \bigsqcup(\set{A}^{\Bot*})$.  Suppose
    that $e \in (\bigsqcap\set{A})^\Bot$.  Consider the possibility that each
    $\sem{A} \in \set{A}$ might contain a term $v_{\sem{A}}$ incompatible with
    $e$ (\ie $\cut{v_{\sem{A}}}{e} \notin \Bot$), and yet each such
    $v_{\sem{A}}$ might not end up in the intersection of $\set{A}$
    ($v_{\sem{A}} \notin \bigsqcap\set{A}$).  In this case, $e$ is still
    orthogonal to every term in $\bigsqcap\set{A}$, but there is no individual
    $\sem{A} \in \set{A}$ such that $e \in \sem{A}^{\Bot}$ because each one
    has an associated counter-example $v_{\sem{A}}$ ruling it out.
    
  \item The last fact follows from the above and triple orthogonal elimination
    (\cref{thm:orthogonal-toe}).
    \begin{align*}
      \left(\bigsqcap(\set{A}^{\Bot*})\right)^{\Bot\Bot}
      &=
      \left(\bigsqcup\set{A}\right)^{\Bot\Bot\Bot}
      =
      \left(\bigsqcup\set{A}\right)^{\Bot}
      =
      \bigsqcap(\set{A}^{\Bot*})
      \\
      \left(\bigsqcup(\set{A}^{\Bot*\Bot*})\right)^{\Bot}
      &=
      \bigsqcap(\set{A}^{\Bot*\Bot*\Bot*})
      =
      \bigsqcap(\set{A}^{\Bot*})
      =
      \left(\bigsqcup\set{A}\right)^{\Bot}
      \qedhere
    \end{align*}
  \end{enumerate}
\end{proof}
Take note that the missing direction in the asymmetric property (2)
($(\sem{A}\sqcap\sem{B})^\Bot \not\refines (\sem{A}^\Bot)\sqcup(\sem{B}^\Bot)$)
exactly corresponds to the direction of the de Morgan laws which is rejected by
intuitionistic logic (the negation of a conjunction does not imply the
disjunction of the negations).  Instead, we have a weakened version of (2)
presented in (3), adding additional applications of orthogonality to restore the
symmetric equality rather than an asymmetric refinement.  As with other
properties like triple orthogonal elimination, this also has a (co)\-value
restricted variant.
\begin{lemma}[Restricted de Morgan]
\label{thm:restricted-de-morgan}

For any set of pre-candidates $\set{A} \subseteq \precands$:
\begin{align*}
  \left(\bigsqcap(\set{A}^{v*\Bot* v*})\right)^{\Bot v \Bot v}
  &=
  \bigsqcap(\set{A}^{v*\Bot* v*})
  =
  \left(\bigsqcup\set{A}\right)^{v\Bot v}
  =
  \left(\bigsqcup(\set{A}^{v*\Bot* v*\Bot*})\right)^{v\Bot v}
\end{align*}
\end{lemma}
\begin{proof}
  Follows from the de Morgan laws (\cref{thm:orthogonal-de-morgan}), restricted
  triple orthogonal elimination (\cref{thm:restricted-toe}), and the fact that
  the (co)\-value restriction $\blank^v$ distributes over intersection and
  unions:
  \begin{align*}
    \left(\bigsqcap(\set{A}^{v*\Bot*v*})\right)^{\Bot v\Bot v}
    &=
    \left(\bigsqcup\set{A}\right)^{v\Bot v\Bot v\Bot v}
    =
    \left(\bigsqcup\set{A}\right)^{v\Bot v}
    =
    \bigsqcap(\set{A}^{v*\Bot*v*})
    \\
    \left(\bigsqcup(\set{A}^{v*\Bot*v*\Bot*})\right)^{v\Bot v}
    &=
    \bigsqcap(\set{A}^{v*\Bot*v*\Bot*v*\Bot*v*})
    =
    \bigsqcap(\set{A}^{v*\Bot*v})
    =
    \left(\bigsqcup\set{A}\right)^{v\Bot v}
    \qedhere
  \end{align*}
\end{proof}

With these de Morgan properties of intersection and union, we can be more
specific about how the subtype lattice operations $\bigwedge$ and $\bigvee$ on
pre-candidates are lifted into the complete versions $\bigcurlywedge$ and
$\bigcurlyvee$ that form the subtype lattice among reducibility candidates.
\begin{lemma}
\label{thm:subtype-intersection-orthogonality}
\label{thm:subtype-union-orthogonality}

Let $\set{A} \subseteq \redcands$ be any set of reducibility candidates.
\begin{align*}
  \bigcurlywedge\set{A}
  &=
  \left(\bigwedge\set{A}\right)^{v\Bot v\Bot}
  =
  \left(
    \left(\bigcap\{\sem{A}^+ \mid \sem{A} \in \set{A}\}\right)^{v\Bot v\Bot}
    ,
    \left(\bigcap\{\sem{A}^+ \mid \sem{A} \in \set{A}\}\right)^{v\Bot}
  \right)
  \\
  \bigcurlyvee\set{A}
  &=
  \left(\bigvee\set{A}\right)^{v\Bot v\Bot}
  =
  \left(
    \left(\bigcap\{\sem{A}^- \mid \sem{A} \in \set{A}\}\right)^{v\Bot}
    ,
    \left(\bigcap\{\sem{A}^- \mid \sem{A} \in \set{A}\}\right)^{v\Bot v\Bot}
  \right)
\end{align*}
\end{lemma}
\begin{proof}
  Follows from de Morgan duality
  (\cref{thm:orthogonal-de-morgan,thm:restricted-de-morgan}) and the fact that
  reducibility candidates are fixed points of $\blank^{v\Bot}$
  (\cref{thm:fixed-point-candidate}).  Let
  $\set{A}^+ = \{\sem{A}^+ \mid \sem{A} \in \set{A}\}$ and
  $\set{A}^- = \{\sem{A}^- \mid \sem{A} \in \set{A}\}$ in the following:
  \begin{equation*}
    \arraycolsep=0.5ex
    \def\arraystretch{1.25}
  \begin{array}[b]{rclrcl}
    \bigcurlywedge\set{A}
    &=&
    \NegCand\bigwedge\set{A}
    &
    \bigcurlyvee\set{A}
    &=&
    \PosCand\bigvee\set{A}
    \\
    &=&
    (
      (\bigcup\set{A}^-)^{v\Bot v\Bot v\Bot}
      ,
      (\bigcup\set{A}^-)^{v\Bot v\Bot}
    )
    &
    &=&
    (
      (\bigcup\set{A}^+)^{v\Bot v\Bot}
      ,
      (\bigcup\set{A}^+)^{v\Bot v\Bot v\Bot}
    )
    \\
    &=&
    (
      (\bigcap\set{A}^{-v*\Bot*})^{v\Bot v\Bot}
      ,
      (\bigcap\set{A}^{-v*\Bot*})^{v\Bot}
    )
    &
    &=&
    (
      (\bigcap\set{A}^{+v*\Bot*})^{v\Bot}
      ,
      (\bigcap\set{A}^{+v*\Bot*})^{v\Bot v\Bot}
    )
    \\
    &=&
    (
      (\bigcap\set{A}^{+})^{v\Bot v\Bot}
      ,
      (\bigcap\set{A}^{+})^{v\Bot}
    )
    &
    &=&
    (
      (\bigcap\set{A}^{-})^{v\Bot}
      ,
      (\bigcap\set{A}^{-})^{v\Bot v\Bot}
    )
    \\
    &=&
    (
      (\bigcap\set{A}^+)^{v\Bot v\Bot}
      ,
      (\bigcap\set{A}^{-v*\Bot*})^{v\Bot}
    )
    &
    &=&
    (
      (\bigcap\set{A}^{+v*\Bot*})^{v\Bot}
      ,
      (\bigcap\set{A}^-)^{v\Bot v\Bot}
    )
    \\
    &=&
    (
      (\bigcap\set{A}^+)^{v\Bot v\Bot}
      ,
      (\bigcup\set{A}^-)^{v\Bot v\Bot}
    )
    &
    &=&
    (
      (\bigcup\set{A}^+)^{v\Bot v\Bot}
      ,
      (\bigcap\set{A}^-)^{v\Bot v\Bot}
    )
    \\
    &=&
    (
      \bigwedge\set{A}
    )^{v\Bot v\Bot}
    &
    &=&
    (
      \bigvee\set{A}
    )^{v\Bot v\Bot}
  \end{array}
  \qedhere
  \end{equation*}
\end{proof}

\subsection{(Co)Induction and (Co)Recursion}
\label{sec:co-inductive-candidates}

We now examine how the (co)\-inductive types $\Nat$ and $\Stream A$ are properly
defined as reducibility candidates in this orthogonality-based, symmetric model.
As shorthand, we will use these two functions on reducibility candidates
\begin{align*}
  N &: \redcands \to \redcands
  &
  S &: \redcands \to \redcands \to \redcands
  \\
  N(\sem{C})
  &\defeq
  \PosCand(\{\Zero\} \vee \Succ(\sem{C}))
  &
  S_{\sem{A}}(\sem{C})
  &\defeq
  \NegCand(\Head(\sem{A}) \wedge \Tail(\sem{C}))
\end{align*}
defined in terms of these operations that lift constructors and destructors onto
candidates:
\begin{align*}
  \Succ(\sem{C}) &\defeq \{\Succ V \mid V \in \sem{C}\}
  &
  \Tail(\sem{C}) &\defeq \{\Tail E \mid E \in \sem{C}\}
  &
  \Head(\sem{A}) &\defeq \{\Head E \mid E \in \sem{A}\}
\end{align*}
These capture the (co)\-inductive steps for the iterative definitions of
$\den{\Nat}_i$ and $\den{\Stream A}_i$:
\begin{align*}
  \den{\Nat}_{i+1} &= N(\den{\Nat}_i)
  &
  \den{\Stream A}_{i+1} &= S_{\den{A}}(\den{\Stream A}_i)
\end{align*}
They also capture the all-at-once definitions of $\den{\Nat}$ and
$\den{\Stream A}$ as
\begin{align*}
  \den{\Nat}
  &=
  \bigcurlywedge\{\sem{C} \in \redcands \mid N(\sem{C}) \leq \sem{C}\}
  &
  \den{\Stream A}
  &=
  \bigcurlyvee\{\sem{C} \in \redcands \mid \sem{C} \leq S_{\den{A}}(\sem{C})\}
\end{align*}
due to the fact that their closure conditions (under $\Zero,\Succ$ and
$\Head,\Tail$, respectively) are equivalent to these subtyping conditions.
\begin{lemma}
\label{thm:nat-subtype-closure}
\label{thm:stream-subtype-closure}

For all reducibility candidates $\sem{C}$
\begin{enumerate}
\item $N(\sem{C}) \leq \sem{C}$ if and only if $\Zero \in \sem{C}$ and
  $\Succ V \in \sem{C}$ (for all $V \in \sem{C}$).
\item $\sem{C} \leq S_{\sem{A}}(\sem{C})$ if and only if $\Head E \in \sem{C}$
  (for all $E \in \sem{A}$) and $\Tail E \in \sem{C}$ (for all $E \in \sem{C}$).
\end{enumerate}
\end{lemma}
\begin{proof}
  The ``only if'' directions follow directly from
  \cref{thm:pos-completion,thm:neg-completion} by the definitions of $N$,
  $S_{\sem{A}}$, and subtyping.  That is, we know that $\Zero \in N(\sem{C})$
  and $\Succ V \in N(\sem{C})$ (for all $V \in \sem{C}$) by definition of $N$ in
  terms of $\PosCand$, and thus they must be in $\sem{C} \geq N(\sem{C})$ by
  subtyping.  Similarly,
  $\Head E, \Tail F \in S_{\sem{A}}(\sem{C}) \leq \sem{C}$ (for all
  $E \in \sem{A}$ and $F \in \sem{C}$) by subtyping and the definition of $S$ in
  terms of $\NegCand$.

  The ``if'' direction follows from the universal properties of $\PosCand$ and
  $\NegCand$ (\cref{thm:pos-completion,thm:neg-completion}): for any
  reducibility candidate $\sem{C} \extends \sem{B}^v$
  \begin{math}
    \PosCand(\sem{B})
    \leq
    \sem{C}
    \leq
    \NegCand(\sem{B})
    .
  \end{math}
  Now note that
  \begin{align*}
    N(\sem{C})
    &=
    \PosCand(\{\Zero\} \cup \{\Succ V \mid V \in \sem{C}\},\{\})
    \\
    S_{\sem{A}}(\sem{C})
    &=
    \NegCand
    (\{\},\{\Head E \mid E \in \sem{A}\} \cup \{\Tail F \mid F \in \sem{C}\})
  \end{align*}
  Therefore, if $\Zero \in \sem{C}$ and $\Succ V \in \sem{C}$ (for all
  $V \in \sem{C}$) then
  \begin{equation*}
    N(\sem{C})
    \leq
    \sem{C}
    \extends
    (\{\Zero\} \cup \{\Succ V \mid V \in \sem{C}\},\{\})
  \end{equation*}
  Likewise if $\Head E \in \sem{C}$ and $\Tail F \in \sem{C}$ (for all
  $E \in \sem{A}$ and $F \in \sem{C}$) then
  \begin{equation*}
    (\{\},\{\Head E \mid E \in \sem{A}\} \cup \{\Tail F \mid F \in \sem{C}\})
    \refines
    \sem{C}
    \leq
    S_{\sem{A}}(\sem{C})
    \qedhere
  \end{equation*}
\end{proof}

First, consider the model of the $\Nat$ type in terms of the infinite union of
approximations: $\bigcurlyvee_{i=0}^\infty\den{\Nat}_i$.  This reducibility
candidate should contain safe instances of the recursor.  The reason it does is
because the presence of a recursor is preserved by the $N$ stepping operation on
reducibility candidates, and so it must remain in the final union
$\bigcurlyvee_{i=0}^\infty\den{\Nat}_i$ because it is in each approximation
$\den{\Nat}_i$.
\begin{lemma}[$\Nat$ Recursion Step]
\label{thm:nat-recursion-step}

For any reducibility candidates $\sem{B}$ and $\sem{C}$,
\begin{equation*}
  \Rec \{\Zero \to v \mid \Succ x \to y.w\} \With E
  \in
  N(\sem{C})
\end{equation*}
for all $E \in \sem{B}$ whenever the following conditions hold:
\begin{itemize}
\item $v \in \sem{B}$,
\item $w\subs{\asub{x}{V},\asub{y}{W}} \in \sem{B}$ for all $V \in \sem{C}$ and
  $W \in \sem{B}$, and
\item
\begin{math}
  \Rec \{\Zero \to v \mid \Succ x \to y.w\} \With E
  \in
  \sem{C}
\end{math}
for all $E \in \sem{B}$.
\end{itemize}
\end{lemma}
\begin{proof}
  Note that the accumulator continuation $E$ changes during the successor step,
  so we will need to generalize over it.  As shorthand, let
  $E^{\Rec}_{E} \defeq \Rec \{\Zero \to v \mid \Succ x \to y.w\} \With E$ where
  $E'$ stands for the given continuation accumulator.  It suffices to show (via
  \cref{thm:pos-completeness}) that for all $E \in \sem{B}$,
  $\cut{\Zero}{E^{\Rec}_{E}}$ and $\cut{\Succ V}{E^{\Rec}_{E}}$ (for all
  $V \in \sem{C}$).  Observe that, given any $E \in \sem{B}$ and
  $V \in \sem{C}$, we have these two possible reductions:
  \begin{align*}
    \cut{\Zero}{E^{\Rec}_{E}}
    &\srd
    \cut{v}{E}
    &
    \cut{\Succ V}{E^{\Rec}_{E}}
    &\srd
    \cut
    {\outp\alpha \cut{V}{E^{\Rec}_{\alpha}}}
    {\inp y \cut{w\subst{x}{V}}{E}}
  \end{align*}
  Now, we note the following series facts:
  \begin{enumerate}
  \item $\cut{v}{E} \in \Bot$ because $v, E \in \sem{B}$ by assumption.
  \item $w\subs{\asub{x}{V},\asub{y}{W}} \in \sem{B}$ for all $W \in \sem{B}$
    because $V \in \sem{C}$.
  \item $\cut{w\subs{\asub{x}{V},\asub{y}{W}}}{E} \in \Bot$ for all
    $W \in \sem{B}$.
  \item $\inp y \cut{w\subs{\asub{x}{V}}}{E} \in \sem{B}$ by activation
    (\cref{thm:activation-sound}).
  \item $\cut{V}{E^{\Rec}_{E}} \in \Bot$ for all $E \in \sem{B}$ by the
    assumption $V, E^{\Rec}_{E} \in \sem{C}$.
  \item $\outp\alpha \cut{V}{E^{\Rec}_{\alpha}} \in \sem{B}$ by activation
    (\cref{thm:activation-sound}).
  \item
    \begin{math}
      \cut
      {\outp\alpha \cut{V}{E^{\Rec}_{\alpha}}}
      {\inp y \cut{w\subst{x}{V}}{E}}
      \in
      \Bot
      .
    \end{math}
  \end{enumerate}
  Therefore, both $\cut{V}{E^{\Rec}_{E}}$ reduces to to a command in $\Bot$ for
  any $V \in N(\sem{C})$.  It follows from
  \cref{thm:safety-expansion,thm:pos-completeness} that
  \begin{math}
    E^{\Rec}_{E}
    \in
    N(\sem{C})
    .
  \end{math}
  %
\end{proof}

\begin{lemma}[$\Nat$ Recursion]
\label{thm:nat-recursion}

For any reducibility candidate $\sem{B}$, if
\begin{itemize}
\item $v \in \sem{B}$,
\item $w\subs{\asub{x}{V},\asub{y}{W}} \in \sem{B}$ for all
  $V \in \bigcurlyvee_{i=0}^\infty\den{\Nat}_i$ and $W \in \sem{B}$, and
\item $E \in \sem{B}$,
\end{itemize}
then
\begin{math}
  \Rec \{\Zero \to v \mid \Succ x \to y.w\} \With E
  \in
  \bigcurlyvee_{i=0}^\infty\den{\Nat}_i
  .
\end{math}
\end{lemma}
\begin{proof}
  By induction on $i$,
  \begin{math}
    \Rec \{\Zero \to v \mid \Succ x \to y.w\} \With E'
    \in
    \den{\Nat}_i
  \end{math}
  for all $E' \in \sem{B}$:
  \begin{itemize}
  \item ($0$) $\den{\Nat}_0 = \PosCand\{\}$ is the least reducibility candidate
    w.r.t subtyping, \ie with the fewest terms and the most co\-terms, so
    $\Rec \{\Zero \to v \mid \Succ x \to y.w\} \With E' \in \den{\Nat}_0$
    trivially.
  \item ($i+1$) Assume the inductive hypothesis:
    \begin{math}
      \Rec \{\Zero \to v \mid \Succ x \to y.w\} \With E'
      \in
      \den{\Nat}_i
    \end{math}
    for all $E' \in \sem{B}$.  Applying \cref{thm:nat-recursion-step} to
    $\den{\Nat}_i$, we have (for all $E' \in \sem{B}$):
    \begin{align*}
      \Rec \{\Zero \to v \mid \Succ x \to y.w\} \With E'
      \in
      N(\den{\Nat}_{i+1})
      =
      \den{\Nat}_{i+1}
    \end{align*}
  \end{itemize}
  Thus, we know that the least upper bound of all $\den{\Nat}_i$ in the subtype
  lattice (\cref{thm:subtype-lattice}) contains the instance of this recursor
  with $E' = E \in \sem{B}$ because it is a co\-value
  (\cref{thm:pos-completion}):
  \begin{equation*}
    \Rec \{\Zero \to v \mid \Succ x \to y.w\} \With E
    \in
    \bigcurlyvee_{i=0}^\infty \den{\Nat}_i
    \qedhere
  \end{equation*}
\end{proof}



Showing that the union $\bigcurlyvee_{i=0}^\infty\den{\Nat}_i$ is closed under
the $\Succ$ constructor is more difficult.  The simple union
$\bigvee_{i=0}^\infty\den{\Nat}_i$ is clearly closed under $\Succ$: every value
in $\bigvee_{i=0}^\infty\den{\Nat}_i$ must come from some individual
approximation $\den{\Nat}_n$, so its successor is in the next one
$\den{\Nat}_{n+1}$.  However, $\bigcurlyvee_{i=0}^\infty\den{\Nat}_i$ is not
just a simple union; it has been completed by $\PosCand$, so it
might---hypothetically---contain \emph{more} terms which do not come from any
individual $\den{\Nat}_n$.  Thankfully, this does not happen.  It turns out the
two unions are one in the same---the $\PosCand$ completion cannot add anything
more because of infinite recursors which can inspect numbers of any size---which
lets us show that $\bigcurlyvee_{i=0}^\infty\den{\Nat}_i$ is indeed closed under
$\Succ$.
\begin{lemma}[Nat Choice]
\label{thm:nat-choice}

\begin{math}
  \bigcurlyvee_{i=0}^\infty\den{\Nat}_i
  =
  \bigvee_{i=0}^\infty\den{\Nat}_i
  .
\end{math}
\end{lemma}
\begin{proof}
  We already know that
  $\bigvee_{i=0}^\infty\den{\Nat}_i \leq \bigcurlyvee_{i=0}^\infty\den{\Nat}_i$
  by definition (since all reducibility candidates are pre-candidates), so it
  suffices to show the reverse:
  $\bigcurlyvee_{i=0}^\infty\den{\Nat}_i \leq \bigvee_{i=0}^\infty\den{\Nat}_i$.
  
  Note from \cref{thm:subtype-union-orthogonality} that
  \begin{align*}
    \textstyle
    \bigcurlyvee_{i=0}^\infty\den{\Nat}_i
    &=
    \textstyle
    \left(\bigvee_{i=0}^\infty\den{\Nat}_i\right)^{v\Bot v\Bot}
    =
    \left(
      \left(\bigvee_{i=0}^\infty\den{\Nat}_i\right)^{-v\Bot}
      ,
      \left(\bigvee_{i=0}^\infty\den{\Nat}_i\right)^{-v\Bot v\Bot}
    \right)
  \end{align*}
  We will proceed by showing there is an
  $E \in \bigvee_{i=0}^\infty\den{\Nat}_i$ such that $\cut{V}{E} \in \Bot$
  forces $V \in \bigvee_{i=0}^\infty\den{\Nat}_i$.  Since we know that every
  $V \in \bigcurlyvee_{i=0}^\infty\den{\Nat}_i$ and
  $E \in \bigvee_{i=0}^\infty\den{\Nat}_i$ forms a safe command
  $\cut{V}{E} \in \Bot$, this proves the result.
  
  First, observe that $\bigcurlyvee_{i=0}^\infty\den{\Nat}_i$ contains the
  following instance of the recursor (\cref{thm:nat-recursion}):
  \begin{align*}
    \Rec_\infty
    &\defeq
    \Rec\{\Zero \to \Zero \mid \Succ \blank \to x. x\} \With \alpha
    \in
    \bigcurlyvee_{i=0}^\infty\den{\Nat}_i
  \end{align*}
  which has these reductions with $\Zero$ and $\Succ$:
  \begin{align*}
    \cut{\Zero}{\Rec_\infty} &\srd \cut{\Zero}{\alpha}
    &
    \cut{\Succ V}{\Rec_\infty}
    &\srd
    \cut
    {\outp\alpha \cut{V}{\Rec_\infty}}
    {\inp x \cut{x}{\alpha}}
  \end{align*}
  The reason why this co\-value forces values of
  $\bigcurlyvee_{i=0}^\infty\den{\Nat}_i$ into values of
  $\bigvee_{i=0}^\infty\den{\Nat}_i$ depends on the evaluation strategy.

  \emph{In call-by-value}, the first successor reduction proceeds as:
  \begin{align*}
    \cut{\Succ V}{\Rec_\infty}
    &\srds
    \cut{V}{\Rec_\infty\subst{\alpha}{\inp x \cut{x}{\alpha}}}
  \end{align*}
  where the continuation accumulator has been $\tmu$-expanded.  In general, an
  arbitrary step in this reduction sequence looks like:
  \begin{align*}
    \cut{\Succ V}{\Rec_\infty\subst{\alpha}{E}}
    &\srds
    \cut{V}{\Rec_\infty\subst{\alpha}{\inp x \cut{x}{E}}}
  \end{align*}
  The only values (in call-by-value) which do not get stuck with $\Rec_\infty$
  (\ie values $V$ such that $\cut{V}{\Rec_\infty} \srds c \in \mathit{Final}$)
  have the form $\Succ^n \Zero$ which is in
  $\den{\Nat}_{n+1} \leq \bigvee_{i=0}^\infty\den{\Nat}_i$.

  \emph{In call-by-name}, the successor reduction proceeds as:
  \begin{align*}
    \cut{\Succ V}{\Rec_\infty}
    &\srds
    \cut{V}{\Rec_\infty}
  \end{align*}
  Call-by-name includes another form of value, $\outp\beta c$, which is not
  immediately stuck with $\Rec_\infty$.  We now need to show that
  $\cut{V}{\Rec_\infty} \in \Bot$, \ie
  $\cut{V}{\Rec_\infty} \srds c \in \mathit{Final}$, forces $V \in \den{\Nat}_n$
  for some $n$.  Let's examine this reduction more closely and check the
  intermediate results by abstracting out $\Rec_\infty$ with a fresh
  co\-variable $\beta$: $\cut{V}{\Rec_\infty} \srds c \in \mathit{Final}$
  because $\cut{V}{\Rec_\infty} \srds c'\subst{\beta}{\Rec_\infty}$ for some
  $\cut{V}{\beta} \srds c' \not\srd$ and then
  $c'\subst{\beta}{\Rec_\infty} \srds c \in \mathit{Final}$.  We now proceed by
  (strong) induction on the length of the remaining reduction sequence (\ie the
  number of steps in $c'\subst{\beta}{\Rec_\infty} \srds c$) and by cases on the
  shape of the intermediate $c'$:
  \begin{itemize}
  \item $c' \neq \cut{W}{\beta}$.  Then
    \begin{math}
      c'\subst{\beta}{\Rec_\infty}
      \in
      \mathit{Final}
    \end{math}
    already.  In this case, $\cut{W}{E} \in \Bot$ for any $E$ whatsoever by
    expansion (\cref{thm:safety-expansion}), and so
    $W \in \den{\Nat}_0 = \PosCand\{\}$.
  \item $c' = \cut{W}{\beta}$.  Then we know that
    \begin{math}
      c'\subst{\beta}{\Rec_\infty}
      =
      \cut{W\subst{\beta}{\Rec_\infty}}{\Rec_\infty}
      \srds
      c
      \in
      \mathit{Final}
      .
    \end{math}
    Since $\cut{W}{\beta} \not\srd$ we know $W\subst{\beta}{\Rec_\infty}$ is not
    a $\mu$-abstraction.  The only other possibilities for
    $W\subst{\beta}{\Rec_\infty}$, given the known reduction to $c$, are $\Zero$
    or $\Succ V'$ for some $V'$.  $\Zero \in \den{\Nat}_1$ by definition.  In
    the case of $\Succ V'$, we have the (non-reflexive) reduction sequence
    \begin{math}
      c'\subst{\beta}{\Rec_\infty}
      =
      \cut{\Succ V'}{\Rec_\infty}
      \srds
      \cut{V'}{\Rec_\infty}
      \srds
      c
      .
    \end{math}
    The inductive hypothesis on the smaller reduction
    $\cut{V'}{\Rec_\infty} \srds c$ ensures $V' \in \den{\Nat}_n$ for some $n$,
    so that $\Succ V' \in \den{\Nat}_{n+1}$ by definition, and thus
    $V \in \den{\Nat}_{n+1}$ as well by expansion.
  \end{itemize}

  So in both call-by-value and call-by-name, we have
  \begin{math}
    \left(\bigcap_{i=0}^\infty\den{\Nat}_i^-\right)^{v\Bot}
    \subseteq
    \bigcup_{i=0}^\infty\den{\Nat}_i^+
    .
  \end{math}
  De Morgan duality (\cref{thm:fixed-point-candidate}) ensures the reverse, so
  \begin{math}
    \left(\bigcap_{i=0}^\infty\den{\Nat}_i^-\right)^{v\Bot}
    =
    \bigcup_{i=0}^\infty\den{\Nat}_i^+
    .
  \end{math}
  Finally, because all reducibility candidates are fixed points of
  $\blank^{v\Bot}$ (\cref{thm:fixed-point-candidate}), de Morgan duality further
  implies:
  \begin{small}
  \begin{align*}
    \textstyle
    \bigcurlyvee_{i=0}^\infty\den{\Nat}_i
    &=
    \textstyle
    \left(
      \left(\bigcap_{i=0}^\infty\den{\Nat}_i^-\right)^{v\Bot}
      ,
      \left(\bigcap_{i=0}^\infty\den{\Nat}_i^-\right)^{v\Bot v\Bot}
    \right)
    =
    \left(
      \bigcup_{i=0}^\infty\den{\Nat}_i^+
      ,
      \left(\bigcup_{i=0}^\infty\den{\Nat}_i^+\right)^{v\Bot}
    \right)
    \\
    &=
    \textstyle
    \left(
      \bigcup_{i=0}^\infty\den{\Nat}_i^+
      ,
      \bigcap_{i=0}^\infty\den{\Nat}_i^{+v\Bot}
    \right)
    =
    \left(
      \bigcup_{i=0}^\infty\den{\Nat}_i^+
      ,
      \bigcap_{i=0}^\infty\den{\Nat}_i^-
    \right)
    =
    \bigvee_{i=0}^\infty\den{\Nat}_i
    \qedhere
  \end{align*}
  \end{small}
\end{proof}



Due to the symmetry of the model, the story for $\Stream A$ is very much the
same as $\Nat$.  We can show that the intersection of
approximations---$\bigcurlywedge_{i=0}^\infty\den{\Stream A}_i$---contains safe
instances of the co\-recursor because its presence is preserved by the stepping
function $S$.  The task of showing that this intersection is closed under the
$\Tail$ destructor is more challenging in the same way as $\Succ$ closure.  We
solve it with the dual method: the presence of co\-recursors which ``inspect''
stream continuations of any size ensures that there are no new surprises that
are not already found in one of the approximations $\den{\Stream A}_n$.
\begin{lemma}[$\Stream$ Co\-recursion Step]
\label{thm:stream-corecursion-step}

For any reducibility candidates $\sem{A}$, $\sem{B}$ and $\sem{C}$,
\begin{align*}
  \CoRec \{\Head\alpha \to e \mid \Tail\beta \to \gamma.f\} \With V
  \in
  S_{\sem{A}}(\sem{C})
\end{align*}
for all $V \in \sem{B}$ whenever the following conditions hold:
\begin{itemize}
\item $e\subst{\alpha}{E} \in \sem{B}$ for all $E \in \sem{A}$,
\item $f\subs{\asub{\beta}{E},\asub{\gamma}{F}} \in \sem{B}$ for all
  $E \in \sem{C}$ and $F \in \sem{B}$, and
\item
  \begin{math}
    \CoRec \{\Head\alpha \to e \mid \Tail\beta \to \gamma.f\} \With V
    \in
    \sem{C}
  \end{math}
  for all $V \in \sem{B}$.
\end{itemize}
\end{lemma}
\begin{proof}
  Since the value accumulator $V$ will change in the tail step, we have to
  generalize over it.  Let
  \begin{math}
    V^{\CoRec}_{V}
    \defeq
    \CoRec \{\Head\alpha \to e \mid \Tail\beta \to \gamma.f\} \With V
    ,
  \end{math}
  where $V$ stands for a given value from $\sem{B}$.  It suffices to show (via
  \cref{thm:neg-completeness}) that for all $V \in \sem{B}$,
  $\cut{V^{\CoRec}_{V}}{\Head E}$ (for all $E \in \sem{A}$) and
  $\cut{V^{\CoRec}_{V}}{\Tail E'}$ (for all $E' \in \sem{C}$).  We have these
  two possible reductions:
  \begin{align*}
    \cut{V^{\CoRec}_{V}}{\Head E}
    &\srd
    \cut{V}{e\subst{\alpha}{E}}
    \\
    \cut{V^{\CoRec}_{V}}{\Tail E'}
    &\srd
    \cut
    {\outp\gamma \cut{V}{f\subst{\beta}{E'}}}
    {\inp x \cut{V^{\CoRec}_x}{E'}}
  \end{align*}
  Now, we note the following series of facts
  \begin{enumerate}
  \item $e\subst{\alpha}{E} \in \sem{B}$ by assumption because $E \in \sem{A}$.
  \item $\cut{V}{e\subst{\alpha}{E}} \in \Bot$ because
    $V, e\subst{\alpha}{E} \in \sem{B}$ by assumption.
  \item $f\subs{\asub{\beta}{E'},\asub{\gamma}{F}} \in \sem{B}$ for all
    $F \in \sem{B}$ by because
    \begin{math}
      E'
      \in
      \sem{C}
      .
    \end{math}
  \item $\cut{V}{f\subs{\asub{\beta}{E'},\asub{\gamma}{F}}} \in \Bot$ for all
    $F \in \sem{B}$.
  \item $\outp\gamma\cut{V}{f\subs{\asub{\beta}{E'}}} \in \sem{B}$ by
    activation (\cref{thm:activation-sound}).
  \item $\cut{V^{\CoRec}_{V'}}{E'} \in \Bot$ for all $V' \in \sem{B}$ by the
    inductive hypothesis since $E' \in \den{\Stream A}_i$.
  \item $\inp x \cut{V^{\CoRec}_x}{E'} \in \sem{B}$ by activation
    (\cref{thm:activation-sound}).
  \item
    \begin{math}
      \cut
      {\outp\gamma\cut{V}{f\subs{\asub{\beta}{E'}}}}
      {\inp x \cut{V^{\CoRec}_x}{E'}}
      \in
      \Bot
      .
    \end{math}
  \end{enumerate}
  Therefore, both $\cut{V^{\CoRec}_{V}}{\Head E}$ and
  $\cut{V^{\CoRec}_{V}}{\Tail E'}$ reduce to a command in $\Bot$ for any
  $E \in \sem{A}$ and $E' \in \sem{C}$.  It follows from
  \cref{thm:safety-expansion,thm:neg-completeness} that
  \begin{math}
    V^{\CoRec}_{V'}
    \in
    S_{\sem{A}}(\sem{C})
    .
  \end{math}
\end{proof}



\begin{lemma}[$\Stream$ Co\-recursion]
\label{thm:stream-corecursion}

For any reducibility candidate $\sem{B}$, if
\begin{itemize}
\item $e\subst{\alpha}{E} \in \sem{B}$ for all $E \in \den{A}$,
\item $f\subs{\asub{\beta}{E},\asub{\gamma}{F}} \in \sem{B}$ for all
  $E \in \bigcurlywedge_{i=0}^\infty\den{\Stream A}_i$ and $F \in \sem{B}$, and
\item $V \in \sem{B}$,
\end{itemize}
then
\begin{math}
  \CoRec \{\Head\alpha \to e \mid \Tail\beta \to \gamma.f\} \With V
  \in
  \bigcurlywedge_{i=0}^\infty\den{\Stream A}_i
  .
\end{math}
\end{lemma}
\begin{proof}
  By induction on $i$,
  \begin{math}
    \CoRec \{\Head\alpha \to e \mid \Tail\beta \to \gamma.f\} \With V'
    \in
    \den{\Stream A}_i
  \end{math}
  for all $V' \in \sem{B}$:
  \begin{itemize}
  \item ($0$) $\den{\Stream A}_0 = \NegCand\{\}$ is the greatest reducibility
    candidate w.r.t subtyping , \ie with the most terms, so
    \begin{math}
      \CoRec \{\Head\alpha \to e \mid \Tail\beta \to \gamma.f\} \With V'
      \in
      \den{\Stream A}_0
    \end{math}
    trivially.
  \item ($i+1$) Assume
    \begin{math}
      \CoRec \{\Head\alpha \to e \mid \Tail\beta \to \gamma.f\} \With V'
      \in
      \den{\Stream A}_i
    \end{math}
    for all $V' \in \sem{B}$.  Applying \cref{thm:stream-corecursion-step} to
    $\sem{\Stream A}_i$, we have (for all $V' \in \sem{B}$):
    \begin{equation*}
      \CoRec \{\Head\alpha \to e \mid \Tail\beta \to \gamma.f\} \With V'
      \in
      S_{\den{A}}(\Tail\den{\Stream A}_i)
      =
      \den{\Stream A}_{i+1}
    \end{equation*}
  \end{itemize}
  Thus, we know that in the greatest lower bound of all $\den{\Stream A}_i$ in
  the subtype lattice (\cref{thm:subtype-lattice}) contains this corecursor with
  $V' = V \in \sem{B}$ because it is a value (\cref{thm:neg-completion}):
  \begin{equation*}
    \CoRec \{\Head\alpha \to e \mid \Tail\beta \to \gamma.f\} \With V
    =
    \bigcurlywedge_{i=0}^\infty\den{\Stream A}_{i}
    \qedhere
  \end{equation*}
\end{proof}

\begin{lemma}[Stream Choice]
\label{thm:stream-choice}

\begin{math}
  \bigcurlywedge_{i=0}^\infty\den{\Stream A}_i
  =
  \bigwedge_{i=0}^\infty\den{\Stream A}_i
\end{math}
\end{lemma}
\begin{proof}
  Note from \cref{thm:subtype-union-orthogonality} that
  \begin{small}
  \begin{align*}
    \textstyle
    \bigcurlywedge_{i=0}^\infty\den{\Stream A}_i
    &=
    \textstyle
    \left(\bigwedge_{i=0}^\infty\den{\Stream A}_i\right)^{v\Bot v\Bot}
    =
    \left(
      \left(\bigwedge_{i=0}^\infty\den{\Stream }_i\right)^{+v\Bot v\Bot}
      ,
      \left(\bigwedge_{i=0}^\infty\den{\Stream }_i\right)^{+v\Bot}
    \right)
  \end{align*}
  \end{small}%
  We will proceed by showing there is a
  $V \in \bigwedge_{i=0}^\infty\den{\Stream A}_i$ such that
  $\cut{V}{E} \in \Bot$ forces $E \in \bigwedge_{i=0}^\infty\den{\Stream A}_i$.
  Since we know that every $E \in \bigcurlywedge_{i=0}^\infty\den{\Stream A}_i$
  and $V \in \bigwedge_{i=0}^\infty\den{\Stream A}_i$ forms a safe command
  $\cut{V}{E} \in \Bot$, this proves the result.

  First, we define the following corecursive term:
  \begin{align*}
    \CoRec_\infty[V]
    &\defeq
    \CoRec
    \{\Head \alpha \to \alpha \to \Tail \blank \to \gamma. \gamma\}
    \With V
  \end{align*}
  and observe that
  \begin{math}
    \CoRec_\infty[V]
    \in
    \bigcurlywedge_{i=0}^\infty\den{\Stream A}_i  
  \end{math}
  (\cref{thm:stream-corecursion}) for all $V \in \den{A}$.  In general,
  $\CoRec_\infty[V]$ has these reductions with $\Head$ and $\Tail$:
  \begin{align*}
    \cut{\CoRec_\infty[V]}{\Head E}
    &\srd
    \cut{V}{\alpha}
    &
    \cut{\CoRec_\infty[V]}{\Tail E}
    &\srd
    \cut
    {\outp\gamma \cut{V}{\gamma}}
    {\inp x \cut{\CoRec_\infty[x]}{E}}
  \end{align*}
  The reason why this value forces co\-values of
  $\bigcurlywedge_{i=0}^\infty\den{\Stream A}_i$ into co\-values of
  $\bigwedge_{i=0}^\infty\den{\Stream A}_i$ depends on the evaluation strategy.

  \emph{In call-by-name}, the tail reduction proceeds as:
  \begin{align*}
    \cut{\CoRec_\infty[V]}{\Tail E}
    &\srds
    \cut{\CoRec_\infty[\outp\gamma \cut{V}{\gamma}]}{E}
  \end{align*}
  where the value accumulator has been $\mu$-expanded.  The only co\-values (in
  call-by-name) which do not get stuck with $\CoRec_\infty$ (\ie co\-values $E$
  such that $\cut{\CoRec_\infty[V]}{E} \srds c \in \mathit{Final}$) have the
  form $\Tail^n (\Head E)$ with $E \in \den{A}$, which is in
  $\den{\Stream A}_{n+1} \geq \bigcurlywedge_{i=0}^\infty\den{\Stream A}_i$.

  \emph{In call-by-value}, the tail reduction proceeds as:
  \begin{align*}
    \cut{\CoRec_\infty[V]}{\Tail E}
    &\srds
    \cut{\CoRec_\infty[V]}{E}
  \end{align*}
  Call-by-value includes another form of co\-value, $\inp y c$, which is not
  immediately stuck with $\CoRec_\infty$.  We now need to show that if
  $V \in \den{A}$ then $\cut{\CoRec_\infty[V]}{E} \in \Bot$, \ie
  $\cut{\CoRec_\infty[V]}{E} \srds c \in \mathit{Final}$, forces
  $E \in \den{\Stream A}_n$ for some $n$.  Let's look at the intermediate
  results of this reduction sequence by abstracting out $\CoRec_\infty$ with a
  fresh variable $y$: $\cut{\CoRec_\infty[V]}{E} \srds c$ because
  $\cut{\CoRec_\infty[V]}{E} \srds c'\subst{y}{\CoRec_\infty[V]}$ for some
  $\cut{y}{E} \srds c' \not\srd$ and then
  $c'\subst{y}{\CoRec_\infty[V]} \srds c \in \mathit{Final}$.  We now proceed by
  (strong) induction on the length of the remain reduction sequence (\ie the
  number of steps in $c'\subst{y}{\CoRec_\infty[V]} \srds c$) and by cases on
  the shape of the intermediate $c'$:
  \begin{itemize}
  \item $c' \neq \cut{y}{F}$.  Then
    $c'\subst{y}{\CoRec_\infty[V]} \in \mathit{Final}$ already.  In this case,
    $\cut{W}{F} \in \Bot$ for any $W$ whatsoever by expansion
    (\cref{thm:safety-expansion}), and so
    $F \in \den{\Stream A}_0 = \NegCand\{\}$.
  \item $c' = \cut{y}{F}$.  Then
    \begin{math}
      c'\subst{y}{\CoRec_\infty[V]}
      =
      \cut{\CoRec_\infty[V]}{F\subst{y}{\CoRec_\infty[V]}}
      \srds
      c
      \in
      \mathit{Final}
      .
    \end{math}
    Since $\cut{y}{F} \not\srd$, we know $F\subst{y}{\CoRec_\infty[V]}$ is not a
    $\tmu$-abstraction.  The only other possibilities for
    $F\subst{y}{\CoRec_\infty[V]}$, given the known reduction to $c$, are
    $\Head F'$ or $\Tail E'$.  In the first case, we have
    \begin{math}
      \cut{\CoRec_\infty[V]}{\Head F'}
      \srd
      \cut{V}{F'}
      \in
      \Bot
    \end{math}
    for all $V \in \den{A}$; so $F' \in \den{A}$ by completion of $\den{A}$ and
    thus $\Head F' \in \den{\Stream A}_1$.  In the second case, we have the
    (non-reflexive) reduction sequence
    \begin{math}
      c'\subst{y}{\CoRec_\infty[V]}
      =
      \cut{\CoRec_\infty[V]}{\Tail E'}
      \srds
      \cut{\CoRec_\infty[V]}{E'}
      \srds
      c
    \end{math}
    The inductive hypothesis on the smaller reduction
    $\cut{\CoRec_\infty[V]}{E'} \srds c$ ensures $E' \in \den{\Stream A}_n$ for
    some $n$, so that $\Head E' \in \den{\Stream A}_{n+1}$ by definition, and
    thus $E \in \den{\Stream A}_{n+1}$ as well by expansion.
  \end{itemize}

  So in both call-by-name and -value, we have
  \begin{math}
    \left(\bigcap_{i=0}^\infty\den{\Stream A}_i^+\right)^{v\Bot}
    \subseteq
    \bigcup_{i=0}^\infty\den{\Stream A}_i^-
    .
  \end{math}
  De Morgan duality (\cref{thm:fixed-point-candidate}) ensures
  \begin{math}
    \left(\bigcap_{i=0}^\infty\den{\Stream A}_i^+\right)^{v\Bot}
    =
    \bigcup_{i=0}^\infty\den{\Stream A}_i^-
    .
  \end{math}
  Finally, because all reducibility candidates are fixed points of
  $\blank^{v\Bot}$ (\cref{thm:fixed-point-candidate}), de Morgan duality further
  implies:
  \begin{align*}
    \textstyle
    \bigcurlywedge_{i=0}^\infty\den{\Stream A}_i
    &=
    \textstyle
    \left(
      \left(\bigcap_{i=0}^\infty\den{\Stream A}_i^+\right)^{v\Bot v\Bot}
      ,
      \left(\bigcap_{i=0}^\infty\den{\Stream A}_i^+\right)^{v\Bot}
    \right)
    \\
    &=
    \textstyle
    \left(
      \left(\bigcup_{i=0}^\infty\den{\Stream A}_i^-\right)^{v\Bot}
      ,
      \bigcup_{i=0}^\infty\den{\Stream A}_i^-
    \right)
    \\
    &=
    \textstyle
    \left(
      \bigcap_{i=0}^\infty\den{\Stream A}_i^{-v\Bot}
      ,
      \bigcup_{i=0}^\infty\den{\Stream A}_i^-
    \right)
    \\
    &=
    \textstyle
    \left(
      \bigcap_{i=0}^\infty\den{\Stream A}_i^+
      ,
      \bigcup_{i=0}^\infty\den{\Stream A}_i^-
    \right)
    =
    \bigwedge_{i=0}^\infty\den{\Stream A}_i
    \qedhere
  \end{align*}
\end{proof}

Now that we know that the iterative interpretations of $\Nat$ and $\Stream A$
contain all the expected parts---the (de)\-constructors and
(co)\-recursors---we are ready to show that they are the same as the all-at-once
definition given in \cref{fig:termination-model}.  More specifically, the
iterative $\bigcurlyvee_{i=0}^\infty\den{\Nat}_i$ and
$\bigcurlywedge_{i=0}^\infty\den{\Stream A}_i$ correspond to the Kleene notion
of (least and greatest, respectively) fixed points.  Instead, the all-at-once
$\den{\Nat}$ and $\den{\Stream A}$ correspond to the Knaster-Tarski fixed point
definitions.  These two correspond because the generating functions $N$ and $S$
are monotonic, and due to the fact that we can choose which approximation each
value of $\bigcurlyvee_{i=0}^\infty\den{\Nat}_i$ and co\-value of
$\bigcurlywedge_{i=0}^\infty\den{\Stream A}_i$ comes from
(\cref{thm:nat-choice,thm:stream-choice}).
\begin{lemma}[Monotonicity]
\label{thm:nat-monotonicity}
\label{thm:stream-monotonicity}

Given reducibility candidates $\sem{A}$, $\sem{B}$, and $\sem{C}$, if
$\sem{B} \leq \sem{C}$ then $N(\sem{B}) \leq N(\sem{C})$ and
$S_{\sem{A}}(\sem{B}) \leq S_{\sem{A}}(\sem{C})$.
\end{lemma}
\begin{proof}
  Because each of the (de)constructors, $\{\Zero\}$, $\Succ(\sem{C})$,
  $\Head(\sem{A})$, and $\Succ(\sem{C})$ are monotonic w.r.t subtyping, as are
  unions, intersections, $\PosCand$, and $\NegCand$
  (\cref{thm:pos-monotonicity,thm:neg-monotonicity}).
\end{proof}

\thmsubtypecoinduction*
\begin{proof}
  First note that the values of $\bigcurlyvee_{i=0}^\infty\den{\Nat}_i$ is
  closed under $\Zero$ a and $\Succ$:
  \begin{itemize}
  \item $\Zero \in \den{\Nat}_1 \leq \bigcurlyvee_{i=0}^\infty\den{\Nat}_i$ by
    definition.
  \item Given $V \in \bigcurlyvee_{i=0}^\infty\den{\Nat}_i$, we know
    $V \in \bigvee_{i=0}^\infty\den{\Nat}_i$ (\cref{thm:nat-choice}), and thus
    $V \in \den{\Nat}_n$ for some $n$.  So
    $\Succ V \in \den{\Nat}_{n+1} \leq \bigcurlyvee_{i=0}^\infty\den{\Nat}_i$ by
    definition.
  \end{itemize}
  Similarly, the co\-values of $\bigcurlywedge_{i=0}^\infty\den{\Stream A}_i$ is
  closed under $\Head$ and $\Tail$:
  \begin{itemize}
  \item For all $E \in \den{A}$,
    \begin{math}
      \Head E
      \in
      \den{\Stream A}_1
      \leq
      \bigcurlywedge_{i=0}^\infty\den{\Stream A}_i
    \end{math}
    by definition.
  \item Given $E \in \bigcurlywedge_{i=0}^\infty\den{\Stream A}_i$, we know
    $E \in \bigwedge_{i=0}^\infty\den{\Stream A}_i$ (\cref{thm:stream-choice}),
    and thus $E \in \den{\Stream A}_n$ for some $n$.  So
    \begin{math}
      \Tail E
      \in
      \den{\Stream A}_{n+1}
      \leq
      \bigcurlywedge_{i=0}^\infty\den{\Stream A}_i
    \end{math}
    by definition.
  \end{itemize}
  Because of these closure facts, we know from the definition of
  $\bigcurlywedge$ and $\bigcurlyvee$, respectively, that
  \begin{align*}
    \bigcurlyvee_{i=0}^\infty\den{\Nat}_i
    &\geq
    \bigcurlywedge
    \{
    \sem{C}
    \mid
    (\Zero \in \sem{C})
    \conj
    (\forall V \in \sem{C}.~ \Succ V \in \sem{C})
    \}
    =
    \den{\Nat}
    \\[-1ex]
    \bigcurlywedge_{i=0}^\infty\den{\Stream A}_i
    &\leq
    \bigcurlyvee
    \{
    \sem{C}
    \mid
    (\forall E \in \den{A}.~ \Head E \in \sem{C})
    \conj
    (\forall E \in \sem{C}.~ \Tail E \in \sem{C})
    \}
    =
    \den{\Stream A}
  \end{align*}

  Going the other way, it we need to show that each approximation $\den{\Nat}_i$
  is a subtype of the $\sem{C}$s that make up $\den{\Nat}$, and dually that each
  approximation $\den{\Stream A}_i$ is a supertype of the $\sem{C}$s that make
  up $\den{\Stream A}$.  Suppose that $\sem{C}$ is any reducibility candidate
  such that $N(\sem{C}) \leq \sem{C}$.  Then $\den{\Nat}_i \leq \sem{C}$ follows
  by induction on $i$:
  \begin{itemize}
  \item ($0$) $\den{\Nat}_0 = \PosCand\{\}$ is the least reducibility candidate
    w.r.t subtyping, so $\den{\Nat}_0 \leq \sem{C}$.
  \item ($i+1$) Assume that $\den{\Nat}_i \leq \sem{C}$.  The next approximation
    is $\den{\Nat}_{i+1} = N(\den{\Nat}_i)$.  Therefore,
    \begin{math}
      \den{\Nat}_{i+1}
      =
      N(\den{\Nat}_i)
      \leq
      N(\sem{C})
      \leq
      \sem{C}
    \end{math}
    by monotonicity of $N$ (\cref{thm:nat-monotonicity}).
  \end{itemize}

  Similarly, suppose that $\sem{C}$ is any reducibility candidate such that
  $S_{\den{A}}(\sem{C}) \geq \sem{C}$.  Then $\den{\Stream A}_i \geq \sem{C}$
  follows by induction on $i$:
  \begin{itemize}
  \item ($0$) $\den{\Stream A}_0 = \NegCand\{\}$ is the greatest reducibility
    candidate w.r.t subtyping, so $\den{\Stream A}_0 \geq \sem{C}$ trivially.
  \item ($i+1$) Assume that $\den{\Stream A}_i \geq \sem{C}$.  The next
    approximation is $\den{\Stream A}_{i+1} = S_{\den{A}}(\den{\Stream A}_i)$.
    Therefore,
    \begin{math}
      \den{\Stream A}_{i+1}
      =
      S_{\den{A}}(\den{\Stream A}_i)
      \geq
      S_{\den{A}}(\sem{C})
      \geq
      \sem{C}
    \end{math}
    by monotonicity of $S_{\den{A}}$ (\cref{thm:stream-monotonicity}).
  \end{itemize}

  In other words, we know that every $\sem{C} \geq N(\sem{C})$ is an upper bound
  of all approximations $\den{\Nat}_i$, and every
  $\sem{C} \leq S_{\den{A}}(\sem{C})$ is a lower bound of all approximations
  $\den{\Stream A}_i$.  So because $\bigcurlyvee$ is the \emph{least} upper
  bound and $\bigcurlywedge$ is the \emph{greatest} lower bound, we have
  \begin{align*}
    \bigcurlyvee_{i=0}^\infty\den{\Nat}_i
    &\leq
    \sem{C}
    &\text{(if $\sem{C} \geq N(\sem{C})$)}
    &&
    \sem{C}
    &\leq
    \bigcurlywedge_{i=0}^\infty\sem{\Stream A}_i
    &\text{(if $\sem{C} \leq S_{\den{A}}(\sem{C})$)}
  \end{align*}
  In other words, $\bigcurlyvee_{i=0}^\infty\den{\Nat}_i$ is a lower bound of
  the $\sem{C}$s that make up $\den{\Nat}$ and
  $\bigcurlywedge_{i=0}^\infty\sem{\Stream A}_i$ is an upper bound of the
  $\sem{C}$s that make up $\den{\Stream A}$
  (\cref{thm:nat-subtype-closure,thm:stream-subtype-closure}).  Again, since
  $\bigcurlywedge$ is the \emph{greatest} lower bound and $\bigcurlyvee$ is the
  \emph{least} upper bound, we have
  \begin{small}
  \begin{align*}
    \bigcurlyvee_{i=0}^\infty\den{\Nat}_i
    &\leq
    \bigcurlywedge
    \{\sem{C} \mid \sem{C} \geq N(\sem{C})\}
    =
    \den{\Nat}
    &
    \bigcurlywedge_{i=0}^\infty\sem{\Stream A}_i
    &\geq
    \bigcurlyvee
    \{\sem{C} \mid \sem{C} \leq S_{\den{A}}(\sem{C})\}
    =
    \den{\Stream A}
    &&
    \qedhere
  \end{align*}
  \end{small}
\end{proof}

\subsection{Adequacy}
\label{sec:adequacy}

To conclude, we give the full proof of adequacy (\cref{thm:adequacy}) here.
With the lemmas that precede in \cref{sec:co-inductive-candidates}, the
remaining details are now totally standard.  Soundness ensures the safety of the
$\mathit{Cut}$ rule and completeness ensures that the terms of each type are
included in their interpretations as reducibility candidates.  The main role of
adequacy is to show that the guarantees given by the premises of each typing
rule are strong enough to prove their conclusion, and that the notion of
substitution corresponds to the interpretation of typing environments.

\thmadequacy*
\begin{proof}
  By (mutual) induction on the given typing derivation for the command or
  (co)\-term:
  \begin{itemize}
  \item ($\mathit{Cut}$) \emph{Inductive Hypothesis}:
    $\den{\Gamma \entails v \givestype A}$ and
    $\den{\Gamma \entails e \takestype A}$.

    Let $\rho \in \den{\Gamma}$, so $v\subs{\rho}\in\den{A}$ and
    $e\subs{\rho}\in\den{A}$ by the inductive hypothesis.  Observe that
    \begin{math}
      \cut{v}{e}\subs{\rho}
      =
      \cut{v\subs{\rho}}{e\subs{\rho}}
      \in
      \Bot
    \end{math}
    because all reducibility candidates are sound.  In other words,
    $\den{\Gamma \entails \cut{v}{e} \contra}$.

  \item ($\mathit{VarR}$ and $\mathit{VarL}$) $x\subs{\rho} \in \den{A}$ for any
    $\rho \in \den{\Gamma, x \givestype A}$ by definition.  Dually,
    $\alpha\subs{\rho} \in \den{A}$ for any
    $\rho \in \den{\Gamma, \alpha \takestype A}$ by definition.  In other words,
    $\den{\Gamma, x \givestype A \entails x \givestype A}$ and
    $\den{\Gamma, \alpha \takestype A \entails \alpha \takestype A}$.

  \item ($\mathit{ActR}$) \emph{Inductive Hypothesis}:
    $\den{\Gamma, \alpha \takestype A \entails c \contra}$.

    Let $\rho \in \den{\Gamma}$, so that for all $E \in \den{A}$,
    $\asub{\alpha}{E},\rho \in \den{\Gamma, \alpha \takestype A}$ by definition
    and
    \begin{math}
      c\subs{\rho}\subst{\alpha}{E}
      =
      c\subs{\rho,\asub{\alpha}{E}}
      \in
      \Bot
    \end{math}
    by the inductive hypothesis.  Thus,
    \begin{math}
      (\outp\alpha c)\subs{\rho}
      =
      \outp\alpha (c\subs{\rho})
      \in
      \den{A}
    \end{math}
    by activation (\cref{thm:activation-sound}).  In other words,
    $\den{\Gamma \entails \outp\alpha c \givestype A}$.

  \item ($\mathit{ActL}$)  Dual to $\mathit{ActR}$ above.

  \item (${\to}R$) \emph{Inductive Hypothesis}:
    $\den{\Gamma, x \givestype A \entails v \givestype B}$.

    Let $\rho \in \den{\Gamma}$, so for all $W \in \den{A}$,
    $\asub{x}{W},\rho \in \den{\Gamma, x \givestype A}$ by definition and
    \begin{math}
      v\subs{\rho}\subst{x}{W}
      =
      v\subs{\rho,\asub{x}{W}}
      \in
      \den{B}
    \end{math}
    by the inductive hypothesis.  Thus,
    \begin{math}
      (\lambda x. v)\subs{\rho}
      =
      \lambda x. (v\subs{\rho})
      \in
      \den{A \to B}
    \end{math}
    by \cref{thm:function-abstraction}.  In other words
    $\den{\Gamma \entails \lambda x. v \givestype A \to B}$.

  \item (${\to}L$) \emph{Inductive Hypothesis}:
    $\den{\Gamma \entails V \givestype A}$ and
    $\den{\Gamma \entails E \takestype B}$.

    Let $\rho \in \den{\Gamma}$, so $V\subs{\rho} \in \den{A}$ and
    $E\subs{\rho} \in \den{B}$ by the inductive hypothesis.  Thus,
    \begin{math}
      (\app V E) \subs \rho
      =
      \app{V\subs{\rho}}{E\subs{\rho}}
      \in
      \den{A \to B}
    \end{math}
    by definition of $\den{A \to B}$ and \cref{thm:neg-completion}.  In other
    words, $\den{\Gamma \entails \app V E \takestype A \to B}$.

  \item (${\Nat}R_{\Zero}$): For any substitution $\rho$,
    \begin{math}
      \Zero\subs{\rho}
      =
      \Zero
      \in
      \den{\Nat}
    \end{math}
    by \cref{thm:pos-completion}.  In other words,
    $\den{\Gamma \entails \Zero \givestype \Nat}$.

  \item (${\Nat}R_{\Succ}$) \emph{Inductive Hypothesis}:
    $\den{\Gamma \entails V : \Nat}$.

    Let $\rho \in \den{\Gamma}$, so that $V\subs{\rho} \in \den{\Nat}$ by the
    inductive hypothesis.  Thus,
    \begin{math}
      (\Succ V)\subs{\rho}
      =
      \Succ (V\subs{\rho})
      \in
      \den{\Nat}
    \end{math}
    by \cref{thm:pos-completion}.  In other words,
    $\den{\Gamma \entails \Succ V \givestype \Nat}$.

  \item (${\Nat}L$) \emph{Inductive Hypothesis}:
    $\den{\Gamma \entails v \givestype A}$,
    $\den{\Gamma, x\givestype\Nat, y \givestype A \entails w \givestype A}$,
    and $\den{\Gamma \entails E \takestype A}$.

    Let $\rho \in \den{\Gamma}$, so that by the inductive hypothesis:
    \begin{itemize}
    \item $E\subs{\rho} \in \den{A}$,
    \item $v \subs{\rho} \in \den{A}$, and
    \item $w\subs{\rho}\subs{\asub{x}{V},\asub{y}{W}} = w\subs{\rho,\asub{x}{V},\asub{y}{W}} \in \den{A}$ for all
      $V \in \den{\Nat}$ and $W \in \den{A}$.
    \end{itemize}
    Thus,
    \begin{align*}
      &
      \Rec
      \{
      \Zero \to v\subs{\rho}
      \mid
      \Succ x \to y. w\subs{\rho}
      \}
      \With E\subs{\rho}
      \\[-2ex]
      &=
      (
      \Rec
      \{
      \Zero \to v
      \mid
      \Succ x \to y. w
      \}
      \With E
      )
      \subs{\rho}
      \in
      \bigcurlyvee_{i=0}^\infty\den{\Nat}_i
      =
      \den{\Nat}
    \end{align*}
    by \cref{thm:nat-inversion,thm:nat-recursion}.  In other words
    \begin{align*}
      \den{
        \Gamma
        \entails
        \Rec
        \{
        \Zero \to v
        \mid
        \Succ x \to y. w
        \}
        \With E
        \takestype
        \Nat
      }
    \end{align*}
    
  \item (${\Stream}R$) \emph{Inductive Hypothesis}:
    $\den{\Gamma \entails E \takestype A}$.

    Let $\rho \in \den{\Gamma}$, so that $E\subs{\rho} \in \den{A}$ by the
    inductive hypothesis.  Thus,
    \begin{math}
      (\Head E)\subs{\rho}
      =
      \Head (E\subs{\rho})
      \in
      \den{\Stream A}
    \end{math}
    by \cref{thm:neg-completion}. In other words,
    $\den{\Gamma \entails \Head E \takestype \Stream A}$.

  \item (${\Stream}L_{\Head}$) \emph{Inductive Hypothesis}:
    $\den{\Gamma \entails E \takestype \Stream A}$.

    Let $\rho \in \den{\Gamma}$, so that $E\subs{\rho} \in \den{\Stream A}$ by
    the inductive hypothesis.  Thus,
    \begin{math}
      (\Tail E)\subs{\rho}
      =
      \Tail (E\subs{\rho})
      \in
      \den{\Stream A}
    \end{math}
    by \cref{thm:neg-completion}. In other words,
    $\den{\Gamma \entails \Tail E \takestype \Stream A}$.

  \item (${\Stream}L_{\Tail}$) \emph{Inductive Hypothesis}:
    $\den{\Gamma, \alpha \takestype A \entails e \takestype B}$,
    $\den{\Gamma, \beta \takestype \Stream A, \gamma \takestype B \entails f \takestype B}$,
    and
    $\den{\Gamma \entails V \givestype B}$.

    Let $\rho \in \den{\Gamma}$, so that by the inductive hypothesis:
    \begin{itemize}
    \item $V\subs{\rho} \in \den{B}$,
    \item $e\subs{\rho}\subst{\alpha}{E} = e\subs{\rho,\asub{\alpha}{E}} \in \den{B}$ for all $E \in \den{A}$, and
    \item $f\subs{\rho}\subs{\asub{\beta}{E},\asub{\gamma}{F}} = f\subs{\rho,\asub{\beta}{E},\asub{\gamma}{F}} \in \den{B}$ for all $E \in \den{\Stream A}$ and $F \in \den{B}$.
    \end{itemize}
    Thus,
    \begin{align*}
      &
      \CoRec
      \{
      \Head\alpha \to e\subs{\rho}
      \mid
      \Tail\beta \to \gamma. f\subs{\rho}
      \}
      \With V\subs{\rho}
      \\[-2ex]
      &=
      (
      \CoRec
      \{
      \Head\alpha \to e
      \mid
      \Tail\beta \to \gamma. f
      \}
      \With V
      )
      \subs{\rho}
      \in
      \bigcurlywedge_{i=0}^\infty\den{\Stream A}_i
      =
      \den{\Stream A}
    \end{align*}
    by \cref{thm:stream-inversion,thm:stream-corecursion}.  In other words
    \begin{equation*}
      \den
      {
        \Gamma
        \entails
        \CoRec
        \{
        \Head\alpha \to e
        \mid
        \Tail\beta \to \gamma. f
        \}
        \With V
        \givestype
        \Stream A
      }
      \qedhere
    \end{equation*}
  \end{itemize}
\end{proof}


%% file: corec-calc.bbl
\begin{thebibliography}{}

\bibitem[\protect\citename{Abel, }2006]{AbelPhD}
Abel, A. (2006)
\newblock {\em A Polymorphic Lambda Calculus with Sized Higher-Order Types}.
\newblock {Ph.D.} thesis, Ludwig-Maximilians-Universit{\"a}t M{\"u}nchen.

\bibitem[\protect\citename{Abel {\em et~al.}\relax, }2013]{Copatterns}
Abel, A., Pientka, B., Thibodeau, D. and Setzer, A. (2013)
\newblock Copatterns: Programming infinite structures by observations.
\newblock  {\em Proceedings of the 40th Annual {ACM} {SIGPLAN-SIGACT} Symposium
  on Principles of Programming Languages}.
\newblock POPL~'13, pp.  27--38.
\newblock ACM.

\bibitem[\protect\citename{Abel \& Pientka, }2013]{WellfoundedCopatterns}
Abel, A.~M. and Pientka, B. (2013)
\newblock Wellfounded recursion with copatterns: A unified approach to
  termination and productivity.
\newblock  {\em Proceedings of the 18th {ACM} {SIGPLAN} International
  Conference on Functional Programming}.
\newblock ICFP~'13, pp.  185--196.
\newblock ACM.

\bibitem[\protect\citename{Ariola {\em et~al.}\relax, }2009]{SequentMachines}
Ariola, Z.~M., Bohannon, A. and Sabry, A. (2009)
\newblock Sequent calculi and abstract machines.
\newblock {\em {ACM} Transactions on Programming Languages and Systems} {\bf
  31}(4):13:1--13:48.

\bibitem[\protect\citename{Barthe \& Uustalu, }2002]{CPSCoInductiveTypes}
Barthe, G. and Uustalu, T. (2002)
\newblock {CPS} translating inductive and coinductive types.
\newblock  {\em Proceedings of the 2002 ACM SIGPLAN Workshop on Partial
  Evaluation and Semantics-Based Program Manipulation}.
\newblock PEPM ’02, p.  131–142.
\newblock Association for Computing Machinery.

\bibitem[\protect\citename{B{\"o}hm \& Berarducci, }1985]{Bhm1985AutomaticSO}
B{\"o}hm, C. and Berarducci, A. (1985)
\newblock Automatic synthesis of typed lambda-programs on term algebras.
\newblock {\em Theoretical Computer Science} {\bf 39}:135--154.

\bibitem[\protect\citename{Curien \& Herbelin, }2000]{DualityOfComputation}
Curien, P.-L. and Herbelin, H. (2000)
\newblock The duality of computation.
\newblock  {\em Proceedings of the Fifth {ACM} {SIGPLAN} International
  Conference on Functional Programming}.
\newblock ICFP~'00, pp.  233--243.
\newblock ACM.

\bibitem[\protect\citename{Downen \& Ariola, }2018a]{BeyondPolarity}
Downen, P. and Ariola, Z.~M. (2018a)
\newblock Beyond polarity: Towards a multi-discipline intermediate language
  with sharing.
\newblock  {\em 27th {EACSL} Annual Conference on Computer Science Logic, {CSL}
  2018, September 4-7, 2018, Birmingham, {UK}}.
\newblock LIPIcs 119, pp.  21:1--21:23.
\newblock Schloss Dagstuhl - Leibniz-Zentrum f{\"{u}}r Informatik.

\bibitem[\protect\citename{Downen \& Ariola, }2018b]{SequentTutorial}
Downen, P. and Ariola, Z.~M. (2018b)
\newblock A tutorial on computational classical logic and the sequent calculus.
\newblock {\em Journal of Functional Programming} {\bf 28}:e3.

\bibitem[\protect\citename{Downen \& Ariola,
  }2021]{ClassicalCorecursionProgramming}
Downen, P. and Ariola, Z.~M. (2021)
\newblock {\em Classical (Co)Recursion: Programming}.
\newblock \url{https://arxiv.org/abs/2103.06913}.

\bibitem[\protect\citename{Downen {\em et~al.}\relax,
  }2015]{StructuralRecursion}
Downen, P., Johnson-Freyd, P. and Ariola, Z.~M. (2015)
\newblock Structures for structural recursion.
\newblock  {\em Proceedings of the 20th {ACM} {SIGPLAN} International
  Conference on Functional Programming}.
\newblock ICFP~'15, pp.  127--139.
\newblock ACM.

\bibitem[\protect\citename{Downen {\em et~al.}\relax,
  }2019]{DualityOfIntersectonUnionTypes}
Downen, P., Ariola, Z.~M. and Ghilezan, S. (2019)
\newblock The duality of classical intersection and union types.
\newblock {\em Fundamenta Informaticae} {\bf 170}(1-3):39--92.

\bibitem[\protect\citename{Downen {\em et~al.}\relax,
  }2020]{ClassicalStrongNormalization}
Downen, P., Johnson-Freyd, P. and Ariola, Z.~M. (2020)
\newblock Abstracting models of strong normalization for classical calculi.
\newblock {\em Journal of Logical and Algebraic Methods in Programming} {\bf
  111}:100512.

\bibitem[\protect\citename{Felleisen \& Friedman, }1986]{CEK}
Felleisen, M. and Friedman, D.~P. (1986)
\newblock Control operators, the {SECD} machine, and the {$\lambda$}-calculus.
\newblock  {\em Proceedings of the {IFIP TC 2/WG2.2} Working Conference on
  Formal Descriptions of Programming Concepts Part {III}} pp.  193--219.

\bibitem[\protect\citename{Gentzen, }1935]{Gentzen1935UULS1}
Gentzen, G. (1935)
\newblock Untersuchungen {\"u}ber das logische schlie{\ss}en. {I}.
\newblock {\em Mathematische Zeitschrift} {\bf 39}(1):176--210.

\bibitem[\protect\citename{Geuvers, }1992]{GeuversIterationRecursion}
Geuvers, H. (1992)
\newblock Inductive and coinductive types with iteration and recursion.
\newblock  {\em Proceedings of the 1992 Workshop on Types for Proofs and
  Programs, Bastad} pp.  193--217.

\bibitem[\protect\citename{Gibbons, }2003]{OrigamiProgramming}
Gibbons, J. (2003)
\newblock {\em The Fun of Programming}.
\newblock Chap. Origami programming, pp.  41--60.

\bibitem[\protect\citename{Girard, }1987]{LinearLogic}
Girard, J.-Y. (1987)
\newblock Linear logic.
\newblock {\em Theoretical Computer Science} {\bf 50}(1):1--101.

\bibitem[\protect\citename{Girard {\em et~al.}\relax, }1989]{ProofsAndTypes}
Girard, J.-Y., Taylor, P. and Lafont, Y. (1989)
\newblock {\em Proofs and Types}.
\newblock Cambridge University Press.

\bibitem[\protect\citename{G{\"o}del, }1980]{SystemT}
G{\"o}del, K. (1980)
\newblock On a hitherto unexploited extension of the finitary standpoint.
\newblock {\em Journal of Philosophical Logic} {\bf 9}(2):133--142.

\bibitem[\protect\citename{Hagino, }1987]{HaginoCodata}
Hagino, T. (1987)
\newblock A typed lambda calculus with categorical type constructors.
\newblock  {\em Category Theory and Computer Science} pp.  140--157.
\newblock Springer Berlin Heidelberg.

\bibitem[\protect\citename{Harper, }2016]{HarperPFPL}
Harper, R. (2016)
\newblock {\em Practical Foundations for Programming Languages}. 2nd edn.
\newblock Cambridge University Press.

\bibitem[\protect\citename{Hinze {\em et~al.}\relax, }2013]{Hinze13}
Hinze, R., Wu, N. and Gibbons, J. (2013)
\newblock Unifying structured recursion schemes.
\newblock  {\em Proceedings of the 18th ACM SIGPLAN International Conference on
  Functional Programming}.
\newblock ICFP '13, p.  209–220.
\newblock Association for Computing Machinery.

\bibitem[\protect\citename{Hughes {\em et~al.}\relax, }1996]{SizedTypes}
Hughes, J., Pareto, L. and Sabry, A. (1996)
\newblock Proving the correctness of reactive systems using sized types.
\newblock  {\em Proceedings of the 23rd ACM SIGPLAN-SIGACT Symposium on
  Principles of Programming Languages}.
\newblock POPL ’96, p.  410–423.
\newblock Association for Computing Machinery.

\bibitem[\protect\citename{Kleene, }1971]{KleeneFixedPoint}
Kleene, S.~C. (1971)
\newblock {\em Introduction to Metamathematics}.
\newblock Bibliotheca Mathematica, a Series of Monographs on Pure and.
\newblock Wolters-Noordhoff.

\bibitem[\protect\citename{Knaster, }1928]{KnasterFixedPoint}
Knaster, B. (1928)
\newblock Un theoreme sur les functions d'ensembles.
\newblock {\em Ann. Soc. Polon. Math.} {\bf 6}:133--134.

\bibitem[\protect\citename{Krivine, }2007]{KrivineMachine}
Krivine, J.-L. (2007)
\newblock A call-by-name lambda-calculus machine.
\newblock {\em Higher-Order and Symbolic Computation} {\bf 20}(3):199--207.

\bibitem[\protect\citename{Levy, }2001]{LevyPhD}
Levy, P.~B. (2001)
\newblock {\em Call-By-Push-Value}.
\newblock PhD thesis, Queen Mary and Westfield College, University of London.

\bibitem[\protect\citename{Liskov, }1987]{LiskovSubstitutionPrinciple}
Liskov, B. (1987)
\newblock Keynote address-data abstraction and hierarchy.
\newblock  {\em Addendum to the Proceedings on Object-oriented Programming
  Systems, Languages and Applications (Addendum)}.
\newblock OOPSLA '87, pp.  17--34.
\newblock ACM.

\bibitem[\protect\citename{Malcom, }1990]{Malcom90}
Malcom, G. (1990)
\newblock Data structures and program transformation.
\newblock {\em Science of computer programming} {\bf 14}(2):255--279.

\bibitem[\protect\citename{McDermott \& Mycroft, }2019]{ExtendedCBPV}
McDermott, D. and Mycroft, A. (2019)
\newblock Extended call-by-push-value: Reasoning about effectful programs and
  evaluation order.
\newblock  {\em Programming Languages and Systems - 28th European Symposium on
  Programming, {ESOP} 2019, Held as Part of the European Joint Conferences on
  Theory and Practice of Software, {ETAPS} 2019, Prague, Czech Republic, April
  6-11, 2019, Proceedings}.
\newblock Lecture Notes in Computer Science 11423, pp.  235--262.
\newblock Springer.

\bibitem[\protect\citename{Meertens, }1992]{Meertens92}
Meertens, L. (1992)
\newblock Paramorphisms.
\newblock {\em Formal Aspects of Computing} {\bf 4}(09).

\bibitem[\protect\citename{Meijer {\em et~al.}\relax, }1991]{MFP91}
Meijer, E., Fokkinga, M. and Paterson, R. (1991)
\newblock Functional programming with bananas, lenses, envelopes and barbed
  wire.
\newblock  {\em Proceedings of the 5th ACM Conference on Functional Programming
  Languages and Computer Architecture} p.  124–144.
\newblock Springer-Verlag.

\bibitem[\protect\citename{Mendler, }1987]{Mendler87}
Mendler, N.~P. (1987)
\newblock Recursive types and type constraints in second-order lambda calculus.
\newblock  {\em Logic in Computer Science}.

\bibitem[\protect\citename{Mendler, }1988]{Mendler88}
Mendler, N.~P. (1988)
\newblock {\em Inductive Definition in Type Theory}.
\newblock {Ph.D.} thesis, Cornell University.

\bibitem[\protect\citename{Munch-Maccagnoni, }2013]{MunchMaccagnoniPhD}
Munch-Maccagnoni, G. (2013)
\newblock {\em Syntax and Models of a non-Associative Composition of Programs
  and Proofs}.
\newblock PhD thesis, Universit\'e Paris Diderot.

\bibitem[\protect\citename{Roy, }2003]{RoyOPLSSCoinduction}
Roy, C. (2003)
\newblock Coinduction and bisimilarity.
\newblock  {\em Oregon Programming Languages Summer School}.
\newblock OPLSS.

\bibitem[\protect\citename{Rutten, }2019]{RuttenMethodofCoalgebra}
Rutten, J. (2019)
\newblock {\em The Method of Coalgebra: Exercises in coinduction}.
\newblock CWI, Amsterdam, The Netherlands.

\bibitem[\protect\citename{Sangiorgi, }2011]{SangiorgiIntroCoinduction}
Sangiorgi, D. (2011)
\newblock {\em Introduction to Bisimulation and Coinduction}.
\newblock Cambridge University Press.

\bibitem[\protect\citename{Tarski, }1955]{TarskiFixedPoint}
Tarski, A. (1955)
\newblock {A lattice-theoretical fixpoint theorem and its applications.}
\newblock {\em Pacific Journal of Mathematics} {\bf 5}(2):285 -- 309.

\bibitem[\protect\citename{Vene \& Uustalu, }1998]{Vene98functionalprogramming}
Vene, V. and Uustalu, T. (1998)
\newblock Functional programming with apomorphisms (corecursion).
\newblock  {\em Proceedings of the Estonian Academy of Sciences: Physics,
  Mathematics} pp.  147--161.

\bibitem[\protect\citename{Wadler, }2003]{CBVDualToCBN}
Wadler, P. (2003)
\newblock Call-by-value is dual to call-by-name.
\newblock  {\em Proceedings of the Eighth {ACM} {SIGPLAN} International
  Conference on Functional Programming} pp.  189--201.
\newblock ACM.

\bibitem[\protect\citename{Zeilberger, }2009]{ZeilbergerPhD}
Zeilberger, N. (2009)
\newblock {\em The Logical Basis of Evaluation Order and Pattern-Matching}.
\newblock PhD thesis, Carnegie Mellon University.

\end{thebibliography}
